\documentclass[preprint,twocolumn,3p,times]{elsarticle}
\RequirePackage{multirow,booktabs,subfigure,color,array,hhline,makecell}
\usepackage{amssymb}
\usepackage{amsmath}
\usepackage{graphicx}
\usepackage{amsthm}
\usepackage{mathrsfs}
\usepackage{indentfirst}
\usepackage[colorlinks,citecolor=blue,urlcolor=blue]{hyperref}
\allowdisplaybreaks
\usepackage[table]{xcolor}
\usepackage{tikz-network}
\usepackage{pgf}
\usepackage{tikz}
\usetikzlibrary{arrows, decorations.pathmorphing, backgrounds, positioning, fit, shadows, petri, matrix, automata}
\RequirePackage{algorithm,algpseudocode}
\usepackage{changes}
\usepackage{caption}
\usepackage{bbm}

\newtheorem{thm}{Theorem}

\newtheorem{lem}{Lemma}
\newtheorem{assum}{Assumption}
\newtheorem{rem}{Remark}
\newtheorem{cor}{Corollary}

\allowdisplaybreaks[4]
\AtBeginDocument{%
	\providecommand\BibTeX{{%
			\normalfont B\kern-0.5em{\scshape i\kern-0.25em b}\kern-0.8em\TeX}}}
\journal{Applied Soft Computing}
\begin{document}
\captionsetup[figure]{labelfont={bf},labelformat={default},labelsep=period,name={Fig.}}
\begin{frontmatter}
\title{Community detection by spectral methods in multi-layer networks}
\author[label1]{Huan Qing\corref{cor1}}
\ead{qinghuan@u.nus.edu 
\&~qinghuan07131995@163.com
\&~qinghuan@cqut.edu.cn}
\cortext[cor1]{Corresponding author.}
\address[label1]{School of Economics and Finance, Chongqing University of Technology, Chongqing, 400054, China}
\begin{abstract}
Community detection in multi-layer networks is a crucial problem in network analysis. In this paper, we analyze the performance of two spectral clustering algorithms for community detection within the framework of the multi-layer degree-corrected stochastic block model (MLDCSBM) framework. One algorithm is based on the sum of adjacency matrices, while the other utilizes the debiased sum of squared adjacency matrices. We also provide their accelerated versions through subsampling to handle large-scale multi-layer networks. We establish consistency results for community detection of the two proposed methods under MLDCSBM as the size of the network and/or the number of layers increases. Our theorems demonstrate the advantages of utilizing multiple layers for community detection. Our analysis also indicates that spectral clustering with the debiased sum of squared adjacency matrices is generally superior to spectral clustering with the sum of adjacency matrices. Furthermore, we provide a strategy to estimate the number of communities in multi-layer networks by maximizing the averaged modularity. Substantial numerical simulations demonstrate the superiority of our algorithm employing the debiased sum of squared adjacency matrices over existing methods for community detection in multi-layer networks, the high computational efficiency of our accelerated algorithms for large-scale multi-layer networks, and the high accuracy of our strategy in estimating the number of communities. Finally, the analysis of several real-world multi-layer networks yields meaningful insights.
\end{abstract}
\begin{keyword}
Multi-layer networks\sep community detection \sep spectral clustering \sep bias-adjusted\sep multi-layer degree-corrected stochastic block model
\end{keyword}
\end{frontmatter}
\section{Introduction}\label{sec1}
Multi-layer networks, also known as multiplex networks, refer to a type of complex network structure that comprises multiple layers, each representing a distinct type of relationship or interaction between nodes \citep{mucha2010community,kivela2014multilayer,boccaletti2014structure}. These networks provide a powerful framework for modeling real-world systems, where entities are often involved in diverse types of relationships that span multiple dimensions. The significance of multi-layer networks lies in their ability to capture the rich nature of complex systems, enabling more accurate and comprehensive analyses \citep{pilosof2017multilayer,de2023more}. Compared to single-layer networks, multi-layer networks contain richer information and offer a more nuanced understanding of system structure and dynamics. In real-world applications, multi-layer networks have found widespread use in various fields, including social science, biology, transportation, neuroscience, etc. For instance, in social networks, individuals may be connected through email, messaging, social media, etc., each represented by a separate layer \citep{papalexakis2013more, jia2023clustering}. In gene co-expression multi-layer networks, genes co-express at different developmental stages of the animal, and each developmental state represents a layer \citep{narayanan2010simultaneous,bakken2016comprehensive,zhang2017finding}. In transportation systems, roads, railways, and airways may be viewed as separate layers, interconnected and influencing each other \citep{boccaletti2014structure}. In this paper, we focus on a specific type of multi-layer network with the same nodes set of each layer and nodes being solely connected within their respective layers, precluding cross-layer connections.

One of the key challenges in multi-layer networks is the problem of community detection, which is significant in the fields of statistical science, social science, and computer science, as it offers valuable insights into the underlying structure of complex systems. In recent years, community detection in multi-layer networks has garnered significant attention, particularly under the multi-layer stochastic block model (MLSBM). Within this framework, the network for each layer is generated from the classical stochastic block model (SBM) \citep{holland1983stochastic}. \citep{han2015consistent} investigated the consistency of community detection for the spectral method using the sum of adjacency matrices, in scenarios where the number of layers increases while the number of nodes remains fixed under MLSBM. \citep{pensky2019spectral} proposed a variant of the spectral method for nodes that might switch their community memberships across consecutive layers in MLSBM. \citep{paul2020spectral} analyzed the consistency of the spectral method using the sum of adjacency matrices and the matrix factorization method within the MLSBM framework. \citep{jing2021community} studied the consistency of a regularized tensor decomposition approach under a mixed multi-layer SBM, assuming that both nodes and layers have community structures. \citep{lei2020consistent} delved into least-squares estimation and established its consistency under MLSBM. More recently, \citep{lei2023bias} introduced a novel bias-adjusted (i.e., debiased) sum of squared adjacency matrices and proved the consistency of bias-adjusted spectral clustering in MLSBM. Subsequently, \citep{su2023spectral} extended this bias-adjusted spectral method to detect communities in directed networks, establishing its consistency under the directed version of MLSBM. In this paper, we focus specifically on the problem of community detection in multi-layer networks that share a common community structure.

Though MLSBM is powerful in modeling multi-layer networks with community structures, similar to SBM, it cannot capture multi-layer networks in which nodes have heterogeneous degrees since it assumes that the probability of generating an edge between two nodes merely depends on their communities. For single networks, the degree-corrected stochastic block model (DCSBM) \citep{karrer2011stochastic} extended SBM by considering a heterogeneity parameter such that the probability of generating an edge between two nodes depends both on their communities and themselves. Many effective methods have been developed to detect communities under DCSBM \citep{qin2013regularized,lei2015consistency, SCORE,chen2018convexified,jing2022community,qing2024applications}. In this paper, we consider the community detection problem under the multi-layer degree-corrected stochastic block model by showing the consistency of two spectral methods. Meanwhile, though numerical results in \citep{lei2023bias,su2023spectral} show that spectral methods using the bias-adjusted sum of squared adjacency matrices outperform spectral methods using the sum of adjacency matrices, a theoretical understanding of this phenomenon is still a mystery. In this paper, we will reveal this mystery theoretically. Our main contributions are as follows:
\begin{itemize}
  \item We propose two spectral methods for community detection within the framework of the multi-layer degree-corrected stochastic block model. The first approach relies on the conventional sum of adjacency matrices, whereas the second incorporates the new bias-adjusted sum of squared adjacency matrices.We also provide their accelerated versions based on the idea of subsampling to detect communities in large-scale multi-layer networks.
  \item We derive theoretical upper bounds on the error rates for the two methods and establish their consistency as the number of nodes and/or layers increases under the multi-layer degree-corrected stochastic block model. Our theoretical findings underscore the advantages of leveraging multiple layers for community detection, as the sparsity requirement for multi-layer networks is significantly weaker compared to that of single networks, and community detection becomes more precise as the number of layers increases.
  \item We conduct a systematic comparison of the theoretical upper bounds on error rates for both methods and observe that, while the spectral method rooted in the sum of adjacency matrices may appear superior to the spectral method employing the bias-adjusted sum of squared adjacency matrices occasionally, it necessitates an unreasonably high number of layers, rendering it impractical in most real-world scenarios. This accounts for the consistent superiority of the latter approach.
  \item We estimate the number of communities in multi-layer networks by maximizing the average modularity \cite{paul2021null}. We then assess the performance of our two proposed methods against several existing approaches, evaluate the computational efficiency of their accelerated versions, and demonstrate the accuracy of our strategy in determining the number of communities on both simulated and real-world multi-layer networks.
\end{itemize}

The rest of this paper is organized as follows. Section \ref{sec2} describes the multi-layer degree-corrected stochastic block model and the clustering error considered in this paper. Section \ref{sec3} presents our spectral methods. Section \ref{sec4} describes the consistency results and the comparison results. Section \ref{secK} presents a strategy for the estimation of the number of communities in multi-layer networks. Section \ref{sec5} and Section \ref{sec6realdata} present our simulation study and real data analysis, respectively. Section \ref{sec7} concludes. All proofs are given in \ref{SecProofs}.

\emph{Notation.} For any positive integer $m$, $[m]$ denotes the set $\{1,2,\ldots,m\}$ and $I_{m\times m}$ denotes the $m$-by-$m$ identity matrix. For any value $a$, $\mathrm{round}(a)$ means the integer closest to $a$. For any vector $x$, $\|x\|_{q}$ denotes its $l_{q}$ norm for any $q>0$. For any matrix $M$, $M', \|M\|, \|M\|_{F}, \mathrm{rank}(M), M(i,:), M(:,\mathcal{I})$, and $\lambda_{k}(M)$ denote its transpose, spectral norm, Frobenius norm, rank, $i$-th row, columns in the index set $\mathcal{I}$ of $M$, and $k$-th largest eigenvalue ordered by the magnitude. $\mathbb{E}[\cdot]$ and $\mathbb{P}(\cdot)$ denote expectation and probability, respectively. For convenience, we summarize the main symbols used in this paper in Table \ref{table-symbol}.
\begin{table*}[ht]
\centering
\rowcolors{1}{white!22}{white!22}
\begin{tabular}{cc|cc}
\hline
Symbol&Meaning&Symbol&Meaning\\
\hline
$n$&Number of nodes&$L$&Number of layers\\
$A_{l}\in\{0,1\}^{n\times n}$&$l$-th adjacency matrix&$K$&Number of communities\\
$\mathcal{C}_{k}$&$k$-th community&$\hat{\mathcal{C}}_{k}$&$k$-th estimated community\\
$\ell$&$n\times1$ node label vector&$\hat{\ell}$&$n\times1$ estimated node label vector\\
$Z\in\{0,1\}^{n\times K}$&Community membership matrix&$n_{k}$&Size of $k$-th community\\
$n_{\mathrm{min}}$&Minimum community size&$n_{\mathrm{max}}$&Maximum community size\\
$\theta\in[0,1]^{n\times1}$&Degree heterogeneity parameter vector&$\Theta$&Diagonal form of $\theta$\\
$\theta_{\mathrm{min}}$&Minimum entry of $\theta$&$\theta_{\mathrm{max}}$&Maximum entry of $\theta$\\
$B_{l}\in[0,1]^{K\times K}$&$l$-th connectivity matrix&$\Omega_{l}$&$A_{l}$'s expectation\\
$\mathcal{B}$&$\{B_{1}, B_{2}, \ldots, B_{L}\}$&$\hat{f}$&Clustering error\\
$\mathcal{P}_{K}$&Set of all permutations of $\{1,2,\ldots,K\}$&$\pi$&Permutation\\
$A_{\mathrm{sum}}$&$\sum_{l\in[L]}A_{l}$&$\Omega_{\mathrm{sum}}$&$\sum_{l\in[L]}\Omega_{l}$\\
$D_{l}$&$n\times n$ diagonal matrix with $D_{l}(i,i)=\sum_{j\in[n]}A_{l}(i,j)$&$S_{\mathrm{sum}}$&$\sum_{l\in[L]}(A^{2}_{l}-D_{l})$\\
$\tilde{S}_{\mathrm{sum}}$&$\sum_{l\in[L]}\Omega^{2}_{l}$&$U\in\mathbb{R}^{n\times K}$&$K$ leading eigenvectors of $\Omega_{\mathrm{sum}}$\\
$U_{*}$&Row-normalized version of $U$&$X$&$K\times K$ matrix such that $U_{*}=ZX$\\
$\hat{U}\in\mathbb{R}^{n\times K}$&$K$ leading eigenvectors of $A_{\mathrm{sum}}$&$\hat{U}_{*}$&Row-normalized version of $\hat{U}$\\
$V\in\mathbb{R}^{n\times K}$&$K$ leading eigenvectors of $\tilde{S}_{\mathrm{sum}}$&$V_{*}$&Row-normalized version of $V$\\
$Y$&$K\times K$ matrix such that $V_{*}=ZY$&$\hat{V}\in\mathbb{R}^{n\times K}$&$K$ leading eigenvectors of $S_{\mathrm{sum}}$\\
$\hat{V}_{*}$&Row-normalized version of $\hat{V}$&$\hat{f}_{NSoA}$&NSoA's Clustering error\\
$n_{\mathrm{sample}}$&Subsample size of nodes&$\mathcal{S}_{\mathrm{nodes}}$&Set of selected nodes\\
$L_{\mathrm{sample}}$&Subsample size of layers&$\mathcal{S}_{\mathrm{layers}}$&Set of selected layers\\
$A^{\mathrm{sub}}_{\mathrm{sum}}$&Sample version of $A_{\mathrm{sum}}$&$S^{\mathrm{sub}}_{\mathrm{sum}}$&Sample version of $S_{\mathrm{sum}}$\\
$\varpi$&A value depending on $n$&$\tilde{U}\in\mathbb{R}^{n\times K}$&$K$ leading left-singular vectors of $A^{\mathrm{sub}}_{\mathrm{sum}}$\\
$\tilde{U}_{*}$&Row-normalized version of $\tilde{U}$&$\tilde{V}\in\mathbb{R}^{n\times K}$&$K$ leading left-singular vectors of $S^{\mathrm{sub}}_{\mathrm{sum}}$\\
$\tilde{V}_{*}$&Row-normalized version of $\tilde{V}$&$\nu$&Proportion of edges in real networks\\
$\hat{f}_{NDSoSA}$&NDSoSA's Clustering error&$\rho$&Sparsity parameter\\
$Q_{MNavrg}$&Averaged modularity metric&$Q$ and $Q_{s}$&Two K × K orthogonal matrices\\
\hline
\end{tabular}
\caption{Main symbols adopted throughout this paper.}
\label{table-symbol}
\end{table*}
\section{The multi-layer degree-corrected stochastic block model (MLDCSBM)}\label{sec2}
\begin{figure*}
\centering
\subfigure[]{\includegraphics[width=0.33\textwidth]{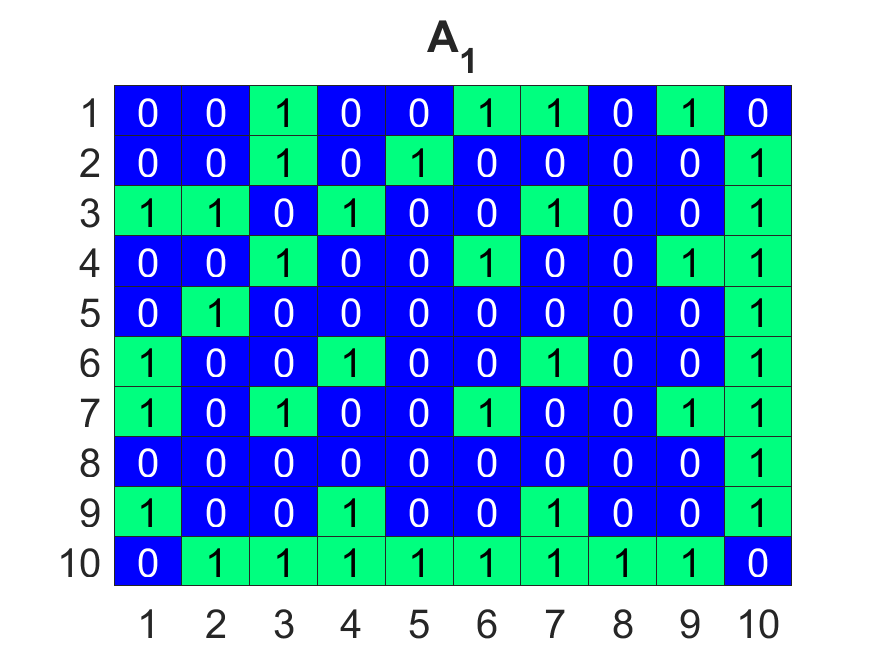}}
\subfigure[]{\includegraphics[width=0.33\textwidth]{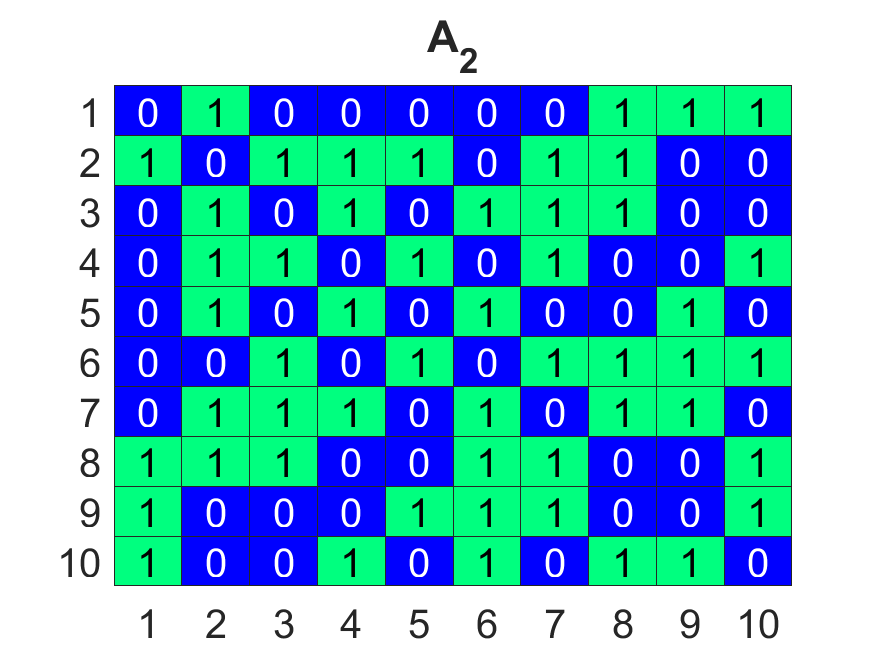}}
\subfigure[]{\includegraphics[width=0.33\textwidth]{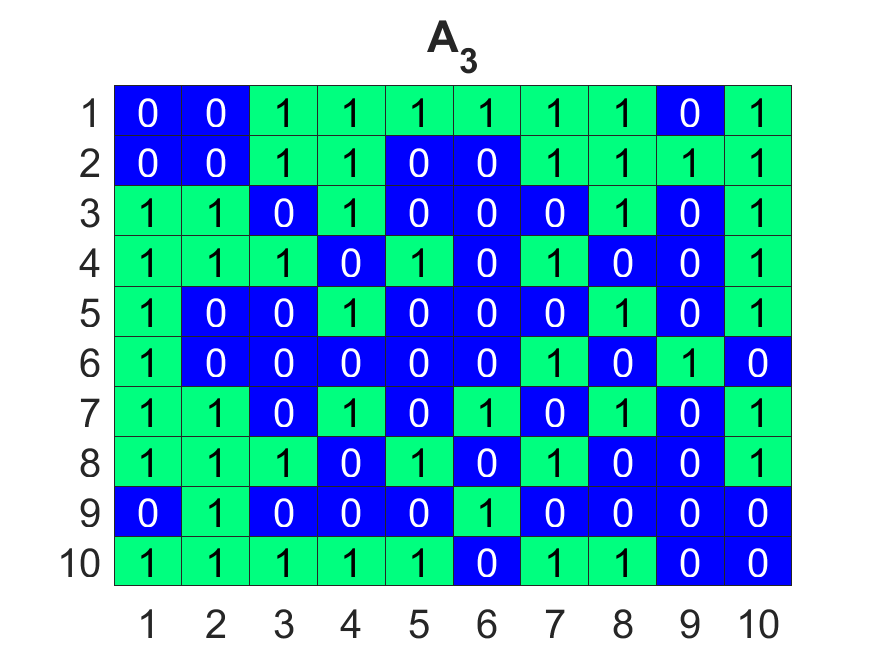}}
\caption{Adjacency matrices of a simple example of multi-layer network with 10 nodes and 3 layers.}
\label{Asim3} 
\end{figure*}
Consider the multi-layer network with $L$-layers and $n$ common nodes. For the $l$-th network, it can be represented by a binary symmetric matrix $A_{l}\in\{0,1\}^{n\times n}$ for all $l\in[L]$. In this paper, we allow the networks to have self-edges (loops). Fig.~\ref{Asim3} presents the adjacency matrices of a toy example of multi-layer network with 10 nodes and 3 layers. Assume that all the layers share a common community structure and all nodes belong to $K$ disjoint communities $\{\mathcal{C}_{1},\mathcal{C}_{2},\ldots, \mathcal{C}_{K}\}$, where $\mathcal{C}_{k}$ is a collection of nodes belonging to the $k$-th community and it has at least one node (i.e., it is non-empty) for $k\in[K]$. Throughout this paper, the number of communities $K$ is assumed to be known. Estimating $K$ with theoretical guarantees is a challenging problem and will not be pursued in this paper. Alternatively, we will introduce an efficient method without theoretical guarantees for estimating K in Section \ref{secK}. Let $\ell$ be an $n\times 1$ vector such that $\ell(i)$ takes a value from $\{1,2,\ldots,K\}$ and represents the community label for node $i$ for all $i\in[n]$. Let $Z\in\{0,1\}^{n\times K}$ be the community membership matrix such that $Z(i,k)=1$ if $\ell(i)=k$ and 0 otherwise for $i\in[n],k\in[K]$. Since each node only belongs to a single community, only the $\ell(i)$-th entry of the $i$-th row of $Z$ is 1 while the other $(K-1)$ entries are zeros. Since each community has at least one node, $Z$'s rank is $K$. Set $n_{k}$ as the number of nodes in the $k$-th community, i.e., $n_{k}=\sum_{i\in[n]}Z(i,k)$ for $k\in[K]$. Set $n_{\mathrm{min}}=\mathrm{min}_{k\in[K]}n_{k}$ and $n_{\mathrm{max}}=\mathrm{max}_{k\in[K]}n_{k}$. Let $\theta$ be an $n\times 1$ vector shared by all layers such that $\theta(i)\in(0,1]$ is the degree heterogeneity of node $i$ for $i\in[n]$. Set $\theta_{\mathrm{min}}=\mathrm{min}_{i\in[n]}\theta(i)$ and $\theta_{\mathrm{max}}=\mathrm{max}_{i\in[n]}\theta(i)$. Let $\Theta$ be an $n\times n$ diagonal matrix with the $i$-th diagonal entry being $\theta(i)$ for $i\in[n]$. Let $B_{l}\in[0,1]^{K\times K}$ be a symmetric matrix and call it block connectivity matrix for the $l$-th layer network for $l\in[L]$. In this paper, we assume that the block connectivity matrices can be different for each layer. For any pair of nodes $i\in[n],j\in[n]$ and any layer $l\in[L]$, the multi-layer DCSBM (MLDCSBM) considered in this paper assumes that each $A_{l}(i,j)$ is generated independently according to
 \begin{align}\label{AlijMLDCSBM}
 A_{l}(i,j)\sim\mathrm{Bernoulli}(\theta(i)\theta(j)B_{l}(\ell(i),\ell(j))).
 \end{align}
Under the MLDCSBM, define $\Omega_{l}\equiv\mathbb{E}[A_{l}]$ for $l\in[L]$. By Equation (\ref{AlijMLDCSBM}), for $l\in[L]$, $\Omega_{l}$ can be written as a product of matrices,
\begin{align}\label{AlMLDCSBM}
\Omega_{l}=\Theta ZB_{l}Z'\Theta.
\end{align}
Set $\mathcal{B}=\{B_{1}, B_{2},\ldots,B_{L}\}$. By Equation (\ref{AlMLDCSBM}), we know that the MLDCSBM considered in this paper is parameterized by the three model parameters $Z,\Theta$, and $\mathcal{B}$, so we denote it by $\mathrm{MLDCSBM}(Z,\Theta,\mathcal{B})$. When all entries of $\theta$ are equal, MLDCSBM reduces to the MLSBM considered in \citep{han2015consistent,paul2016consistent,paul2020spectral,lei2020consistent,lei2023bias}. When $L=1$, MLDCSBM degenerates to the popular DCSBM model \citep{karrer2011stochastic}. Since $\theta(i)\leq\theta_{\mathrm{max}}$ for all $i\in[n]$, we have $\Omega_{l}(i,j)\leq\theta^{2}_{\mathrm{max}}$ for $i\in[n],j\in[n],l\in[L]$, suggesting that $\theta_{\mathrm{max}}$ controls the overall sparsity of the multi-layer network.

Given $Z,\Theta$, and  $\{B_{l}\}^{L}_{l=1}$, $A_{l}$ can be generated from MLDCSBM via Equation (\ref{AlijMLDCSBM}) for all $l\in[L]$. Given the observed adjacency matrices $A_{1},\ldots,A_{L}$, the goal of community detection is to estimate the community label vector $\ell$. Let $\hat{\ell}\in\{1,2,\ldots,K\}^{n\times 1}$ be an estimated community label vector. Set $\hat{\mathcal{C}}_{k}=\{i: \hat{\ell}(i)=k \mathrm{~for~}i\in[n]\}$ as the set of nodes in the $k$-th estimated community for $k\in[K]$. In this paper, to measure the performance of a community detection algorithm when the ground truth community partition $\{\mathcal{C}_{1},\mathcal{C}_{2},\ldots,\mathcal{C}_{K}\}$ is known, we consider the \emph{Clustering error} introduced in \cite{joseph2016impact},
\begin{align}\label{ErrorRate}
\hat{f}=\mathrm{min}_{\pi\in \mathcal{P}_{K}}\mathrm{max}_{k\in[K]}\frac{|\mathcal{C}_{k}\cap \hat{\mathcal{C}}^{c}_{\pi(k)}|+|\mathcal{C}^{c}_{k}\cap \hat{\mathcal{C}}_{\pi(k)}|}{n_{k}},
\end{align}
where $\mathcal{P}_{K}$ denotes the set of all permutations of $\{1,2,\ldots, K\}$ and the superscript $c$ means a complementary set. When the $L$ adjacency matrices $\{A_{l}\}^{L}_{l=1}$ are generated from our model $\mathrm{MLDCSBM}(Z,\Theta,\mathcal{B})$, our primary goal is to design efficient methods for detecting the latent communities of nodes from $\{A_{l}\}^{L}_{l=1}$. Suppose that $\mathcal{M}$ represents such a community detection method and $\hat{\ell}$ is the estimated node labels obtained by applying algorithm $\mathcal{M}$ to $\{A_{l}\}^{L}_{l=1}$ with $K$ communities. In this paper, our secondary goal is to demonstrate that the method $\mathcal{M}$ achieves consistent community detection. Specifically, we aim to show that its Clustering error $\hat{f}$ approaches zero as the number of nodes $n$ and/or the number of layers $L$ tend to infinity. To accomplish this goal, we establish theoretical guarantees for $\mathcal{M}$ by deriving a theoretical upper bound for its Clustering error $\hat{f}$.
\section{Algorithms}\label{sec3}
In this section, we present two spectral clustering algorithms to detect node community labels within the MLDCSBM model. Both algorithms collapse the multi-layer networks into an aggregate one. The first algorithm is constructed utilizing the aggregate matrix $A_{\mathrm{sum}}=\sum_{l\in[L]}A_{l}$, where $A_{\mathrm{sum}}$ represents the sum of adjacency matrices. This matrix has been previously employed in the design of spectral algorithms for community detection under the MLSBM framework \citep{han2015consistent,paul2020spectral}. The second algorithm is designed based on the aggregate matrix $S_{\mathrm{sum}}=\sum_{l\in[L]}(A^{2}_{l}-D_{l})$, where $D_{l}$ is an $n\times n$ diagonal matrix with the $i$-th diagonal entry being $\sum_{j\in[n]}A_{l}(i,j)$, the degree of node $i$ in layer $l$ for $i\in[n],l\in[L]$. $S_{\mathrm{sum}}$ is the bias-adjusted sum of squared adjacency matrices and it was initially introduced in \citep{lei2023bias} under MLSBM. In Sections \ref{SecNSoA} and \ref{SecNDSoSA}, we will introduce the two algorithms in detail. In Section \ref{SecSubsampling}, we will accelerate the two algorithms through subsampling to detect communities in large-scale multi-layer networks.
\subsection{Spectral clustering on the sum of adjacency matrices}\label{SecNSoA}
For the aggregate matrix $A_{\mathrm{sum}}$, we have $\mathbb{E}[A_{\mathrm{sum}}]=\sum_{l\in[L]}\Omega_{l}=\sum_{l\in[L]}\Theta ZB_{l}Z'\Theta=\Theta Z(\sum_{l\in[L]}B_{l})Z'\Theta$. Set $\Omega_{\mathrm{sum}}=\sum_{l\in[L]}\Omega_{l}$. For convenience, when we say ``leading eigenvalues'', we are comparing the magnitudes of the eigenvalues throughout this paper. To specify the spectral clustering algorithm designed based on $A_{\mathrm{sum}}$, we first analyze the algebraic properties of $\Omega_{\mathrm{sum}}$, the expectation matrix of $A_{\mathrm{sum}}$, through the following lemma.
\begin{lem}\label{EigOsum}
Under $\mathrm{MLDCSBM}(Z,\Theta,\mathcal{B})$, suppose $\mathrm{rank}(\sum_{l\in[L]}B_{l})=K$. Set $\Omega_{\mathrm{sum}}=U\Sigma U'$ as the compact eigen-decomposition of $\Omega_{\mathrm{sum}}$ such that $U$ is an $n\times K$ matrix satisfying $U'U=I_{K\times K}$ and $\Sigma$ is a $K\times K$ diagonal matrix with the $k$-th diagonal entry being the $k$-th leading eigenvalue of $\Omega_{\mathrm{sum}}$ for $k\in[K]$. Let $U_{*}$ be an $n\times K$ matrix with the $i$-th row being $U_{*}(i,:)=\frac{U(i,:)}{\|U(i,:)\|_{F}}$ for $i\in[n]$. Then $U_{*}=ZX$ with $X$ being a $K\times K$ full-rank matrix and
$\|X(k,:)-X(\tilde{k},:)\|_{F}=\sqrt{2}$ for all $1\leq k<\tilde{k}\leq K$.
\end{lem}

In Lemma \ref{EigOsum}, $U_{*}=ZX$ suggests that $U_{*}$ has $K$ distinct rows and $U_{*}(i,:)=U_{*}(j,:)$ if $\ell(i)=\ell(j)$ for $i,j\in[n]$, i.e., the two rows $U_{*}(i,:)$ and $U_{*}(j,:)$ are identical if nodes $i$ and $j$ belong to the same community. Hence, applying the K-means algorithm to all rows of $U_{*}$ with $K$ clusters exactly recovers $\ell$ up to a permutation of community labels. Summarizing the above analysis yields the following four-stage oracle algorithm, which we call Ideal NSoA. Input: $\{\Omega_{l}\}^{L}_{l=1}$ and $K$. Output: $\ell$.
\begin{itemize}
  \item Compute $\Omega_{\mathrm{sum}}=\sum_{l\in[L]}\Omega_{l}$.
  \item Obtain $U\Sigma U'$, the leading $K$ eigen-decomposition of $\Omega_{\mathrm{sum}}$.
  \item Normalize the rows of $U$ to get $U_{*}$.
  \item Run K-means algorithm to $U_{*}$'s rows with $K$ clusters to obtain $\ell$.
\end{itemize}

In practice, $\Omega_{\mathrm{sum}}$ is unknown but $A_{\mathrm{sum}}$ is given. Let $\hat{U}\hat{\Sigma}\hat{U}'$ be the leading $K$ eigen-decomposition of $A_{\mathrm{sum}}$ such that $\hat{U}$ is an $n\times K$ matrix satisfying $\hat{U}'\hat{U}=I_{K\times K}$ and the $K\times K$ diagonal matrix $\hat{\Sigma}$ contains $A_{\mathrm{sum}}$'s leading $K$ eigenvalues. Let $\hat{U}_{*}$ be an $n\times K$ matrix with the $i$-th row being $\hat{U}_{*}(i,:)=\frac{\hat{U}(i,:)}{\|\hat{U}(i,:)\|_{F}}$ for $i\in[n]$. Because $\Omega_{\mathrm{sum}}$ is the expectation of $A_{\mathrm{sum}}$ under MLDCSBM, $\hat{U}_{*}$ is a slightly perturbed version of $U_{*}$ and it should have roughly $K$ distinct rows. Thus, applying the K-means algorithm to $\hat{U}_{*}$  should return a good community partition. The above analysis leads to a spectral clustering algorithm applied to $A_{\mathrm{sum}}$, which is summarized in Algorithm \ref{alg:NSoA}. It is evident that Ideal NSoA represents the ideal form of NSoA, given that $\Omega_{l}$ corresponds to the expectation of $A_{l}$ for all $l \in [L]$. Note that in the naming of Algorithm \ref{alg:NSoA}, normalized means the normalization step that calculates $\hat{U}_{*}$.
\begin{algorithm*}
\caption{\underline{N}ormalized spectral clustering based on \underline{s}um \underline{o}f \underline{a}djacency matrices (NSoA)}
\label{alg:NSoA}
\begin{algorithmic}[1]
\Require Adjacency matrices $A_{1}, A_{2}, \ldots, A_{L}$, and number of communities $K$.
\Ensure Estimated node labels $\hat{\ell}$.
\State Calculate $A_{\mathrm{sum}}=\sum_{l\in[L]}A_{l}$.
\State Obtain $\hat{U}\hat{\Sigma}\hat{U}'$, the leading $K$ eigen-decomposition of $A_{\mathrm{sum}}$.
\State Calculate $\hat{U}_{*}$ by normalizing each row of $\hat{U}$ to have unit length.
\State Run K-means algorithm on all rows of $\hat{U}_{*}$ with $K$ clusters to obtain $\hat{\ell}$.
\end{algorithmic}
\end{algorithm*}

Here, we provide the space complexity and time complexity of NSoA:
\begin{itemize}
  \item For the space complexity, the storage requirements for $\{A_{l}\}^{L}_{l=1}$, $A_{\mathrm{sum}}$, $\hat{U}$, $\hat{\Sigma}$, $\hat{U}_{*}$, and $\hat{\ell}$ are $O(Ln^{2})$, $O(n^{2})$, $O(nK)$, $O(K^{2})$, $O(nK)$, and $O(n)$, respectively. Given that $K \ll n$, the overall space complexity of NSoA is $O(Ln^{2})$.
  \item For the time complexity, steps 1-4 exhibit time complexities of $O(Ln^{2})$, $O(n^{2}K)$, $O(nK)$, and $O(nK^{2}T_{\mathrm{iter}})$, where $T_{\mathrm{iter}}$ represents the number of iterations for the K-means algorithm (set to 100 in this paper). Since $K \ll n$, the aggregate time complexity of NSoA is $O(Ln^{2} + Kn^{2})$.
\end{itemize}
\subsection{Spectral clustering on the debiased sum of
squared adjacency matrices}\label{SecNDSoSA}
To explain the intuition of the design of the second algorithm using the aggregate matrix $S_{\mathrm{sum}}$, we provide the following lemma which provides the eigen-decomposition of $\tilde{S}_{\mathrm{sum}}\equiv\sum_{l\in[L]}\Omega^{2}_{l}$.
\begin{lem}\label{EigOsum2}
Under $\mathrm{MLDCSBM}(Z,\Theta,\mathcal{B})$, suppose $\mathrm{rank}(\sum_{l\in[L]}B^{2}_{l})=K$. Set $\tilde{S}_{\mathrm{sum}}=V\Lambda V'$ as the compact eigen-decomposition of $\Omega_{\mathrm{sum}}$ such that $V'V=I_{K\times K}$. Let $V_{*}$ be an $n\times K$ matrix with the $i$-th row being $V_{*}(i,:)=\frac{V(i,:)}{\|V(i,:)\|_{F}}$ for $i\in[n]$. Then $V_{*}=ZY$ with $Y$ being a $K\times K$ full-rank matrix and
$\|Y(k,:)-Y(\tilde{k},:)\|_{F}=\sqrt{2}$ for all $1\leq k<\tilde{k}\leq K$.
\end{lem}
\begin{figure*}
\centering
\subfigure[]{\includegraphics[width=0.245\textwidth]{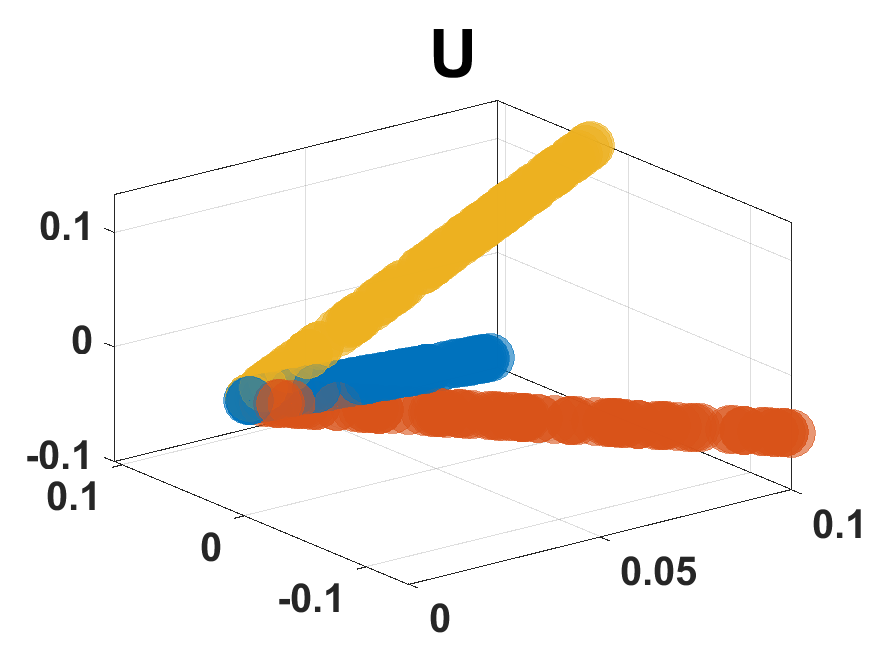}}
\subfigure[]{\includegraphics[width=0.245\textwidth]{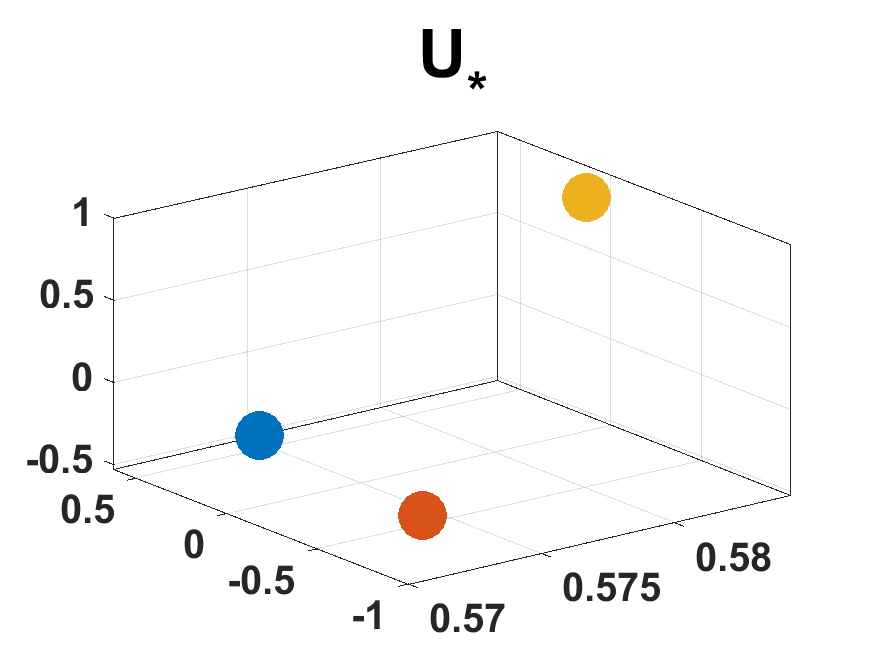}}
\subfigure[]{\includegraphics[width=0.245\textwidth]{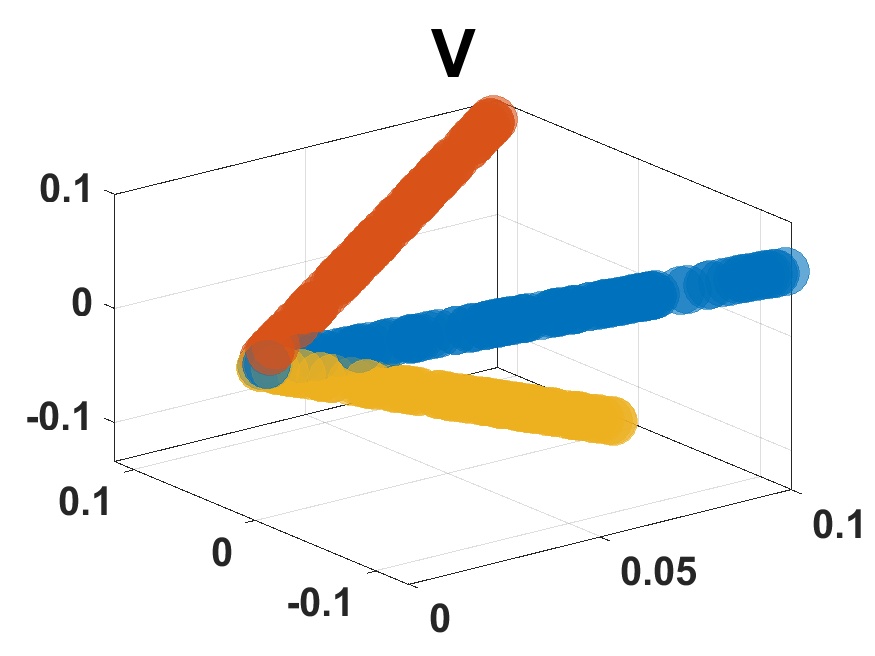}}
\subfigure[]{\includegraphics[width=0.245\textwidth]{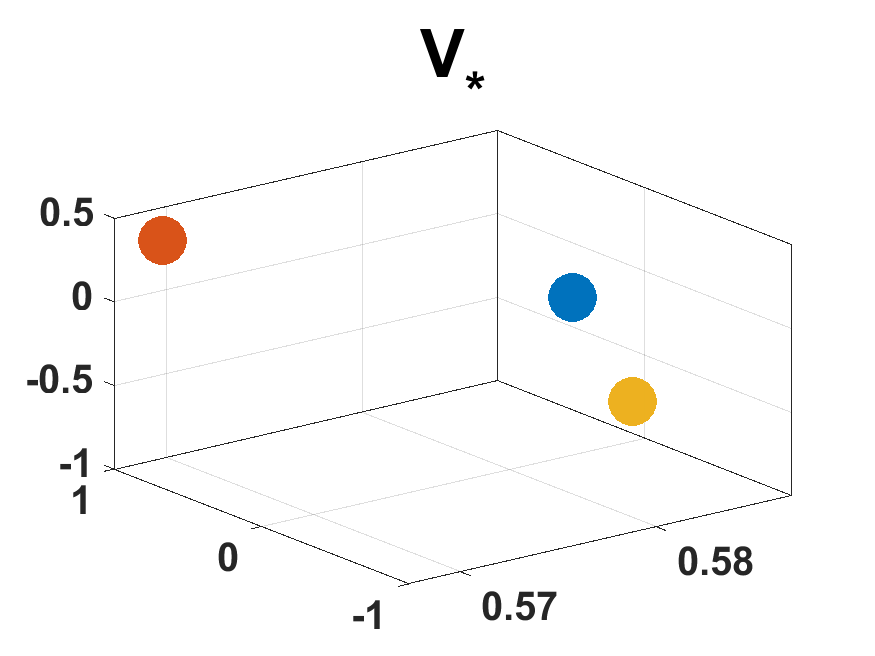}}
\caption{In this simulated example, we set $n=300$, $K=3$, and $L=20$, where each community comprises 100 nodes. Each entry of $B_{l}$ (as well as $\theta$) is a random value from a Uniform distribution on $[0,1]$ for $l\in [L]$. After configuring $Z$, $\mathcal{B}$, and $\Theta$, we calculate $\Omega_{l}$ using Equation (\ref{AlMLDCSBM}), and subsequently derive $\Omega_{\mathrm{sum}}$ and $\tilde{S}_{\mathrm{sum}}$. Then we compute $U$ and $V$, as well as their normalized counterparts $U_{*}$ and $V_{*}$, from $\Omega_{\mathrm{sum}}$ and $\tilde{S}_{\mathrm{sum}}$, respectively. Panels (a)-(d) illustrate points, where each point represents a row from $U$, $U_{*}$, $V$, and $V_{*}$, respectively. Across all panels, points of the same color indicate nodes belonging to the same community. The observations from Panels (a) and (c) reveal that, before normalization, nodes within the same community exhibit \texttt{identical directions} in the projected space. Conversely, Panels (b) and (d) demonstrate that, after normalization, nodes within the same community occupy \texttt{identical positions} in the projected space.}
\label{UUstarVVstar} 
\end{figure*}

Similar to $U_{*}$, applying the K-means algorithm to all rows of $V_{*}$ with $K$ clusters allows for the accurate detection of communities. Analogous to the Ideal NSoA algorithm, we present the Ideal NDSoSA algorithm based on Lemma \ref{EigOsum2}. Input: $\{\Omega_{l}\}^{L}_{l=1}$ and $K$. Output: $\ell$.
\begin{itemize}
  \item Set $\tilde{S}_{\mathrm{sum}}=\sum_{l\in[L]}\Omega^{2}_{l}$.
  \item Obtain $V\Lambda V'$, the leading $K$ eigen-decomposition of $\tilde{S}_{\mathrm{sum}}$.
  \item Normalize $V$'s rows to get $V_{*}$.
  \item Run K-means on $V_{*}$'s rows to obtain $\ell$.
\end{itemize}

To gain a deeper understanding of Lemmas \ref{EigOsum} and \ref{EigOsum2}, we visualize the geometric structure of $U, U_{*}, V$, and $V_{*}$ in Fig.~\ref{UUstarVVstar} using a simulated example with three communities. Panel (a) of Fig.~\ref{UUstarVVstar} indicates that $U$ contains more than $K$ unique rows, making it infeasible to use K-means clustering on $U$ to recover $\ell$ accurately. However, Panel (b) of Fig.~\ref{UUstarVVstar} demonstrates that $U_{*}$ has exactly $K$ distinct rows, and applying K-means clustering to $U_{*}$ successfully recovers $\ell$. Similar conclusions hold for Panels (c) and (d). By \citep{lei2023bias}, we know that $\sum_{l\in[L]}A^{2}_{l}$ is a biased estimate of $\tilde{S}_{\mathrm{sum}}$ while the debiased sum of squared adjacency matrices $S_{\mathrm{sum}}$ is a good estimate of $\tilde{S}_{\mathrm{sum}}$. Let $\hat{V}\hat{\Lambda}\hat{V}'$ be the leading $K$ eigen-decomposition of $S_{\mathrm{sum}}$ such that $\hat{V}'\hat{V}=I_{K\times K}$ and the $K\times K$ diagonal matrix $\hat{\Lambda}$ contains $S_{\mathrm{sum}}$'s leading $K$ eigenvalues as its diagonal elements. Let $\hat{V}_{*}$ be $V$'s row-normalization version such that $\hat{V}_{*}(i,:)=\frac{\hat{V}(i,:)}{\|\hat{V}(i,:)\|_{F}}$ for $i\in[n]$. The community partition can be obtained by running K-means on $\hat{V}_{*}$.  Algorithm \ref{alg:NDSoSA}
designed based on $S_{\mathrm{sum}}$ summarizes the above analysis. By comparing the four steps of Ideal NDSoSA with those of NDSoSA, it becomes apparent that Ideal NDSoSA represents the oracle scenario for NDSoSA. We'd emphasize that both NSoA and NDSoSA solely necessitate the $L$ adjacency matrices $\{A_{l}\}^{L}_{l=1}$ and the number of communities $K$ as inputs, and do not require any additional tuning parameters or regularizers.
\begin{algorithm*}
\caption{\underline{N}ormalized spectral clustering based on \underline{d}ebiased \underline{s}um \underline{o}f
\underline{s}quared \underline{a}djacency matrices (NDSoSA)}
\label{alg:NDSoSA}
\begin{algorithmic}[1]
\Require Adjacency matrices $A_{1}, A_{2}, \ldots, A_{L}$, and number of communities $K$.
\Ensure Estimated node labels $\hat{\ell}$.
\State Calculate $S_{\mathrm{sum}}=\sum_{l\in[L]}(A^{2}_{l}-D_{l})$.
\State Obtain $\hat{V}\hat{\Lambda}\hat{V}'$, the leading $K$ eigen-decomposition of $S_{\mathrm{sum}}$.
\State Calculate $\hat{V}_{*}$ by normalizing each row of $\hat{V}$ to have unit length.
\State Run K-means algorithm on all rows of $\hat{V}_{*}$ with $K$ clusters to obtain $\hat{\ell}$.
\end{algorithmic}
\end{algorithm*}

The space complexity of NDSoSA is identical to that of NSoA, specifically $O(Ln^{2})$. The primary contributor to the time cost of NDSoSA is its initial step, which exhibits a time complexity of $O(Ln^{3})$. Consequently, the overall time complexity of NDSoSA is $O(Ln^{3})$. Given that NSoA's time complexity is $O(Ln^{2}+Kn^{2})$, we see that NSoA yields lower computational cost compared to NDSoSA. Meanwhile, by comparing Algorithms \ref{alg:NSoA} and \ref{alg:NDSoSA}, we see that both NSoA and NDSoSA algorithms involve aggregating adjacency matrices from different layers, performing an eigen-decomposition to extract leading eigenvectors, normalizing these eigenvectors, and then applying K-means clustering to assign community labels to nodes. The key difference lies in the aggregation step, where NSoA simply sums the adjacency matrices, while NDSoSA sums the squared adjacency matrices after adjusting for degrees, which helps in debiasing and improving the performance of community detection as we will analyze further in Section \ref{CompareDA}. In summary, while both algorithms use the multi-layer structure for community detection in multi-layer networks, NDSoSA exhibits superior community detection performance at the cost of increased time complexity compared to NSoA.
\begin{rem}
In Section \ref{sec4}, we will provide theoretical upper bounds on the error rates for our two proposed methods, NSoA and NDSoSA. These theoretical results are based on Assumptions \ref{Assum1}-\ref{Assum22}. It is crucial to emphasize that these assumptions serve solely for theoretical purposes and do not restrict the practical implementations of our methods. This can be clearly seen in Algorithms \ref{alg:NSoA} and \ref{alg:NDSoSA}, as neither algorithm requires conditions related to these assumptions.
\end{rem}
\subsection{Subsampling spectral clustering in large-scale multi-layer networks}\label{SecSubsampling}
In this paper, we call a multi-layer network with $n$ nodes and $L$ layers as a large-scale multi-layer network if either $n$ or $L$ is large. Recall that the computational costs of NSoA and NDSoSA are $O(Ln^{2}+Kn^{2})$ and $O(Ln^{3})$ respectively, it becomes evident that both algorithms can be time-consuming, particularly when $n$ and/or $L$ are large. Recently, Deng et al. \citep{deng2024subsampling} proposed a novel subsampling spectral clustering method designed based on a randomly sampled subnetwork from the entire network to detect communities in single-layer networks. Inspired by the subsampling concept outlined in \citep{deng2024subsampling}, we develop the accelerated versions of NSoA and NDSoSA to identify communities in large-scale multi-layer networks here. Let $\mathcal{S}_{\mathrm{nodes}}$ be a subset of nodes that records $n_{\mathrm{sample}}$ distinct integers, chosen randomly from 1 to $n$ without repetition, where $n_{\mathrm{sample}}$ is the subsample size of nodes. Similarly, let $\mathcal{S}_{\mathrm{layers}}$ be a subset of layers that records $L_{\mathrm{sample}}$ distinct integers, chosen randomly from 1 to $L$ without repetition, where $L_{\mathrm{sample}}$ is the subsample size of layers. For simplicity, let $\mathcal{S}_{\mathrm{layers}}=\{l_{1}, l_{2}, \ldots, l_{L_{\mathrm{sample}}}\}$, where $l_{1}<l_{2}<\ldots<l_{L_{\mathrm{sample}}-1}<l_{L_{\mathrm{sample}}}\leq L$ and $l_{j}$ denotes the $j$-th selected layer for $j\in[L_{\mathrm{sample}}]$. Sure, $A_{l_{j}}$ denotes the adjacency matrix of the $j$-th selected layer for $j\in[L_{\mathrm{sample}}]$. Define an $n\times n_{\mathrm{sample}}$ matrix $A^{\mathrm{sub}}_{l_{j}}$ as a sub-adjacency matrix of $A_{l_{j}}$ such that $A^{\mathrm{sub}}_{l_{j}}=A_{l_{j}}(:,\mathcal{S}_{\mathrm{nodes}})$ for $j\in[L_{\mathrm{sample}}]$. Let $D^{\mathrm{sub}}_{l_{j}}$ be an $n\times n$ matrix with $i$-th diagonal element being the degree of node $i$ in $A^{\mathrm{sub}}_{l_{j}}$, i.e., $D^{\mathrm{sub}}_{l_{j}}(i,i)=\sum_{\bar{i}\in[n_{\mathrm{sample}}]}A^{\mathrm{sub}}_{l_{j}}(i,\bar{i})$ for $j\in[L_{\mathrm{sample}}]$. Define an $n\times n_{\mathrm{sample}}$ matrix $A^{\mathrm{sub}}_{\mathrm{sum}}$ as
\begin{align}\label{SubSum}
A^{\mathrm{sub}}_{\mathrm{sum}}=\sum_{j\in[L_{\mathrm{sample}}]}A^{\mathrm{sub}}_{l_{j}}.
\end{align}
Define an $n\times n$ matrix $\hat{S}_{\mathrm{sum}}$ as
\begin{align}\label{SubSoS1}
\hat{S}_{\mathrm{sum}}=\sum_{j\in[L_{\mathrm{sample}}]}(A^{\mathrm{sub}}_{l_{j}}\times(A^{\mathrm{sub}}_{l_{j}})'-D^{\mathrm{sub}}_{l_{j}}).
\end{align}
Define an $n\times n_{\mathrm{sample}}$ matrix $S^{\mathrm{sub}}_{\mathrm{sum}}$ as
\begin{align}\label{SubSoS2}
S^{\mathrm{sub}}_{\mathrm{sum}}=\hat{S}_{\mathrm{sum}}(:,\mathcal{S}_{\mathrm{nodes}}).
\end{align}
After defining the two $n\times n_{\mathrm{sample}}$ matrices $A^{\mathrm{sub}}_{\mathrm{sum}}$ and $S^{\mathrm{sub}}_{\mathrm{sum}}$, we now introduce the accelerated versions of NSoA and NDSoSA in Algorithms \ref{alg:SNSoA} and \ref{alg:SNDSoSA}, respectively. Below, we present the computational analysis of SNSoA and SNDSoSA:
\begin{itemize}
  \item For the time complexity of SNSoA, its steps 1-5 have time complexities of $O(n+L), O(L_{\mathrm{sample}}n_{\mathrm{sample}}n)$, $O(Kn_{\mathrm{sample}}n)$, $O(nK)$, and $O(nK^{2}T_{\mathrm{iter}})$, respectively. Since $K$ is generally a small integer, as a result, SNSoA's total complexity is $O(L_{\mathrm{sample}}n_{\mathrm{sample}}n)$.
  \item For the time complexity of SNDSoSA, its steps 1-5 have time complexities of $O(n+L), O(L_{\mathrm{sample}}n_{\mathrm{sample}}n^{2})$, $O(Kn_{\mathrm{sample}}n), O(nK)$, and $O(nK^{2}T_{\mathrm{iter}})$, respectively. As a result, SNDSoSA's total complexity is $O(L_{\mathrm{sample}}n_{\mathrm{sample}}n^{2})$.
  \item In this paper, we set $L_{\mathrm{sample}}=\mathrm{round}(\mathrm{log}^{2}(L))$ when $L\geq10$ and $L_{\mathrm{sample}}=L$ otherwise. Inspired by the setting of sample size of nodes $n_{\mathrm{sample}}$ in the numerical study of \citep{deng2024subsampling}, we set $n_{\mathrm{sample}}=\mathrm{round}(\varpi\mathrm{log}^{2}(n))$ if $n\geq500$ and $n_{\mathrm{sample}}=n$ if $n<500$, where $\varpi$ is an integer whose value depends on $n$. In this paper, we set $\varpi=5$ when $500\leq n<2000$, and $\varpi=15$ when $2000\leq n\leq20000$. For values of $n$ exceeding this range, $\varpi$ should be set to a larger value to ensure a larger number of sample nodes are selected. Under the above settings, for large values of $n$ and/or $L$, the computational cost of SNSoA is $O(\varpi \mathrm{log}^{2}(L) \mathrm{log}^{2}(n) n)$, which can be significantly lesser than NSoA's $O(Ln^{2} + Kn^{2})$. Similarly, SNDSoSA's computational cost is $O(\varpi \mathrm{log}^{2}(L) \mathrm{log}^{2}(n) n^{2})$, significantly lower than NDSoSA's $O(Ln^{3})$. Evidently, by considering fewer layers and sub-networks, we achieve acceleration in both NSoA and NDSoSA for large-scale multi-layer networks.
\end{itemize}

\begin{algorithm*}
\caption{\underline{S}ubsampling \underline{NSoA} (SNSoA)}
\label{alg:SNSoA}
\begin{algorithmic}[1]
\Require Adjacency matrices $\{A_{l}\}^{L}_{l=1}$, number of communities $K$, subsample size of nodes $n_{\mathrm{sample}}$, and subsample size of layers $L_{\mathrm{sample}}$.
\Ensure Estimated node labels $\hat{\ell}$.
\State Randomly select $n_{\mathrm{sample}}$ distinct nodes from the entire $n$ nodes and record the selected $n_{\mathrm{sample}}$ distinct nodes in the set $\mathcal{S}_{nodes}$. Randomly select $L_{\mathrm{sample}}$ distinct layers from the entire $L$ layers and record the selected $L_{\mathrm{sample}}$ distinct layers in the set $\mathcal{S}_{layers}$.
\State Calculate $A^{\mathrm{sub}}_{\mathrm{sum}}$ via Equation (\ref{SubSum}).
\State Let $\tilde{U}$ be an $n\times K$ matrix such that $\tilde{U}'\tilde{U}=I_{K\times K}$ and $\tilde{U}(:,k)$ is the left-singular vector of the $k$-th largest singular value of $A^{\mathrm{sub}}_{\mathrm{sum}}$ for $k\in[K]$.
\State Calculate $\tilde{U}_{*}$ by normalizing each row of $\tilde{U}$ to have unit length.
\State Run K-means algorithm on all rows of $\tilde{U}_{*}$ with $K$ clusters to obtain $\hat{\ell}$.
\end{algorithmic}
\end{algorithm*}

\begin{algorithm*}
\caption{\underline{S}ubsampling \underline{NDSoSA} (SNDSoSA)}
\label{alg:SNDSoSA}
\begin{algorithmic}[1]
\Require Adjacency matrices $\{A_{l}\}^{L}_{l=1}$, number of communities $K$, subsample sizes $n_{\mathrm{sample}}$ (nodes) and $L_{\mathrm{sample}}$ (layers).
\Ensure Estimated node labels $\hat{\ell}$.
\State Randomly select $n_{\mathrm{sample}}$ nodes and $L_{\mathrm{sample}}$ layers, recording them in sets $\mathcal{S}_{nodes}$ and $\mathcal{S}_{layers}$, respectively.
\State Calculate $S^{\mathrm{sub}}_{\mathrm{sum}}$ via Equation (\ref{SubSoS2}).
\State Compute $\tilde{V}$, the left-singular vectors of the top $K$ singular values of $S^{\mathrm{sub}}_{\mathrm{sum}}$.
\State Normalize each row of $\tilde{V}$ to obtain $\tilde{V}_{*}$.
\State Run K-means algorithm on all rows of $\tilde{V}_{*}$ with $K$ clusters to obtain $\hat{\ell}$.
\end{algorithmic}
\end{algorithm*}
\section{Consistency results}\label{sec4}
In this section, we show how multi-layer networks benefit community detection by investigating the asymptotic consistency of community detection using NSoA and NDSoSA. In our theoretical analysis, we let both $n$ and $L$ grow, where increasing $n$ means that the size of the multi-layer network grows and increasing $L$ represents that more layers of networks are recorded. Our theoretical results also keep track of other model parameters such as $K, \theta_{\mathrm{min}}, \theta_{\mathrm{max}}, n_{\mathrm{min}}$, and $n_{\mathrm{max}}$. After showing the consistency results, we also provide a detailed comparison of conditions and error rates between NSoA and NDSoSA.
\subsection{Consistency results for NSoA}
In this subsection, we establish the consistency results for NSoA. First, for NSoA's theoretical analysis, we need the following assumption to control the overall sparsity of the multi-layer network.
\begin{assum}\label{Assum1}
 $\theta_{\mathrm{max}}\|\theta\|_{1}L\geq\mathrm{log}(n+L)$.
\end{assum}
Since $\|\theta\|_{1}\leq\theta_{\mathrm{max}}n$, we have $\theta^{2}_{\mathrm{max}}\geq\frac{\mathrm{log}(n+L)}{nL}$ by Assumption \ref{Assum1}. Recall that $\theta_{\mathrm{max}}$ controls the sparsity of each network, $\theta^{2}_{\mathrm{max}}\geq\frac{\mathrm{log}(n+L)}{nL}$ means that a multi-layer network needs a lower requirement on sparsity than a single network since $\frac{\mathrm{log}(n+L)}{nL}$ is much smaller than $\frac{\mathrm{log}(n)}{n}$ for large $L$, and this suggests that a multi-layer network benefits community detection by a lower bound requirement on sparsity of a single network. Assumption \ref{Assum1} is required when we aim at bounding $\|A_{\mathrm{sum}}-\Omega_{\mathrm{sum}}\|$ using the Matrix Bernstein theorem in \cite{tropp2012user} in this paper. It is worth noting that such sparsity requirements are ubiquitous and crucial in establishing theoretical guarantees for community detection methods. For instance, Theorem 4.1 in \citep{qin2013regularized}, Theorem 3.1 in \citep{lei2015consistency}, Theorem 2.2 in \citep{SCORE}, Theorem 1 in \citep{wang2020spectral}, Theorem 3.2 in \citep{mao2021estimating}, Theorem 1 in \citep{su2023spectral}, Theorem 3 in \citep{jing2021community}, Theorem 1 in \citep{lei2023bias}, and Corollary 2 in \citep{xu2023covariate} all impose similar sparsity conditions on their respective community detection algorithms, whether applied to single-layer or multi-layer networks.
\begin{rem}
When $\Theta = \sqrt{\rho}I_{n\times n}$ for $\rho \in (0,1]$, MLDCSBM simplifies to MLSBM. In this scenario, Assumption \ref{Assum1} transforms into $\rho \geq \frac{\log(n+L)}{nL}$, implying that each layer can be exceedingly sparse and our Assumption \ref{Assum1} is mild. For community detection in single-layer networks, spectral methods require $\rho \geq \frac{\log(n)}{n}$ for consistent community detection, as stated in Theorem 3.1 of \citep{lei2015consistency}. Thus, we observe that each layer in the multi-layer network can be significantly sparser for community detection than single-layer networks. Furthermore, consider the special case where there is only one community ($K=1$), reducing the classical SBM to the well-known Erdős–Rényi (ER) random graph $G(n,p)$ \citep{erdos1960evolution}. According to Section 2.5 of \citep{abbe2018community}, the ER random graph is connected with high probability if and only if $\rho \geq \frac{\log(n)}{n}$. Hence, the inequality $\rho \geq \frac{\log(n+L)}{nL}$ implies that each layer may be disconnected and there may exist numerous isolated nodes within each layer. This phenomenon can be clearly observed in Fig.~\ref{Nrho} of Section \ref{sec5}.
\end{rem}
Based on Assumption \ref{Assum1}, the following lemma provides a bound on $\|A_{\mathrm{sum}}-\Omega_{\mathrm{sum}}\|$ under MLDCSBM. This lemma demonstrates that when $A_{l}$ is generated from $\mathrm{MLDCSBM}(Z,\Theta,\mathcal{B})$ for $l\in[L]$, the aggregation matrix $A_{\mathrm{sum}}$ is close to its expectation $\Omega_{\mathrm{sum}}$ in terms of the spectral norm. Using the bound provided by this lemma, we can establish a bound on the difference between $U$ and $\hat{U}$ in terms of the Frobenius norm by Lemma 5.1 of \citep{lei2015consistency}. Finally, we use the bound of the difference between $U$ and $\hat{U}$ to derive a bound on NSoA's Clustering error. To enhance the understanding of the theoretical analysis of the NSoA algorithm, we present a detailed technical roadmap of NSoA in Fig.~\ref{RoadmapNSoA}. As illustrated, to establish the theoretical guarantees of NSoA, we endeavor to bound the outputs between NSoA and its ideal version at each step.
\begin{lem}\label{boundAsum}
Under $\mathrm{MLDCSBM}(Z,\Theta,\mathcal{B})$, when Assumption \ref{Assum1} holds, with probability at least $1-O(\frac{1}{n+L})$, we have
\begin{align*}
\|A_{\mathrm{sum}}-\Omega_{\mathrm{sum}}\|=O(\sqrt{\theta_{\mathrm{max}}\|\theta\|_{1}L\mathrm{log}(n+L)}).
\end{align*}
\end{lem}
\begin{rem}
By the proof of Lemma \ref{boundAsum}, we see that if Assumption \ref{Assum1} becomes $\theta_{\mathrm{max}}\|\theta\|_{1}L\geq\mathrm{log}(n)$, then the probability in Lemma \ref{boundAsum} becomes $1-O(\frac{1}{n})$ and the upper bound of $\|A_{\mathrm{sum}}-\Omega_{\mathrm{sum}}\|$ becomes $O(\sqrt{\theta_{\mathrm{max}}\|\theta\|_{1}L\mathrm{log}(n)})$. We consider the case $\theta_{\mathrm{max}}\|\theta\|_{1}L\geq\mathrm{log}(n+L)$ instead of $\theta_{\mathrm{max}}\|\theta\|_{1}L\geq\mathrm{log}(n)$ because we are mainly interested in the asymptotic region that both $n$ and $L$ grow in this paper.
\end{rem}
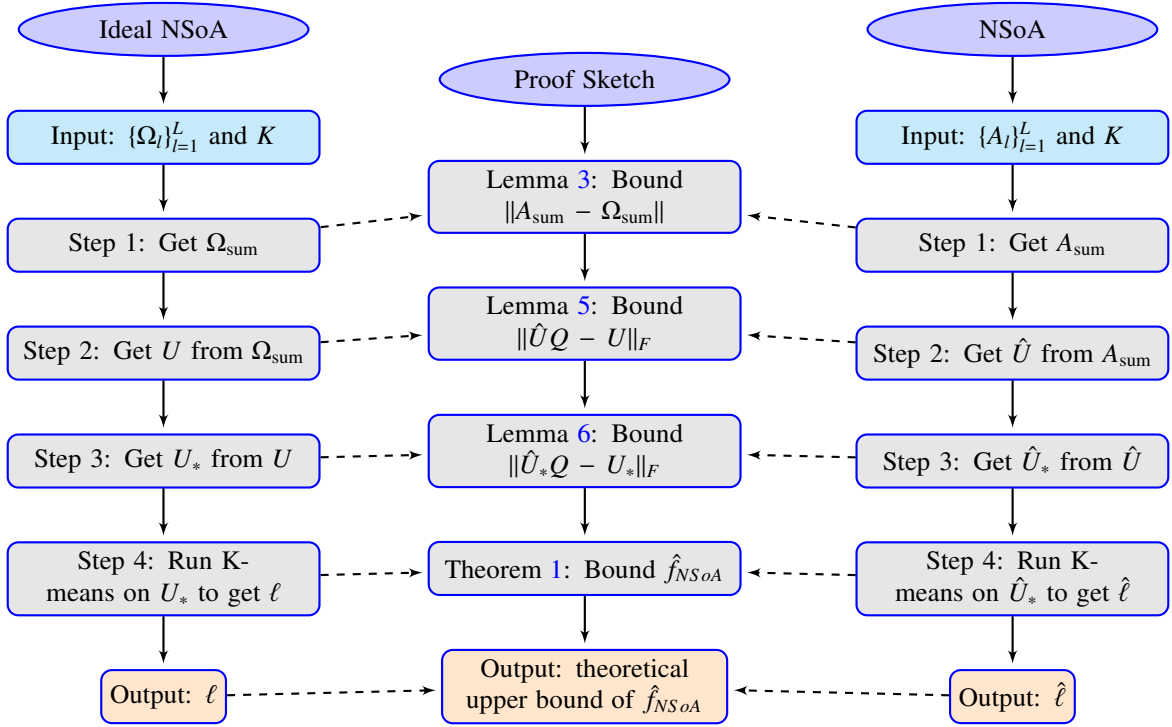
\begin{figure*}[htbp]
\centering
\begin{tikzpicture}[
    remember picture,
    auto,
    decision/.style={diamond, draw=blue, thick, fill=blue!20, text width=4.5em, align=flush center, inner sep=1pt},
    block0/.style={ellipse, draw=blue, thick, fill=blue!20, text width=7em, align=center, rounded corners, minimum height=2em},
    block1/.style={rectangle, draw=blue, thick, fill=cyan!20, text width=11em, align=center, rounded corners, minimum height=2em},
    block2/.style={rectangle, draw=blue, thick, fill=gray!20, text width=11em, align=center, rounded corners, minimum height=2em},
    block3/.style={rectangle, draw=blue, thick, fill=gray!20, text width=11em, align=center, rounded corners, minimum height=2em},
    block4/.style={rectangle, draw=blue, thick, fill=gray!20, text width=11em, align=center, rounded corners, minimum height=2em},
    block5/.style={rectangle, draw=blue, thick, fill=gray!20, text width=11em, align=center, rounded corners, minimum height=2em},
    block6/.style={rectangle, draw=blue, thick, fill=orange!20, text width=4em, align=center, rounded corners, minimum height=2em},
    line/.style={draw, thick, -latex', shorten >=2pt},
    cloud/.style={draw=red, thick, ellipse, fill=red!20, minimum height=2em}
]
\matrix [row sep=7mm, column sep=5mm] {
    \node [block0] (idealNSoA) {Ideal NSoA}; \\
    \node [block1] (ideal1) {Input: $\{\Omega_{l}\}^{L}_{l=1}$ and $K$}; \\
    \node [block2] (ideal2) {Step 1: Get $\Omega_{\mathrm{sum}}$}; \\
    \node [block3] (ideal3) {Step 2: Get $U$ from $\Omega_{\mathrm{sum}}$}; \\
    \node [block4] (ideal4) {Step 3: Get $U_{*}$ from $U$}; \\
    \node [block5] (ideal5) {Step 4: Run K-means on $U_{*}$ to get $\ell$}; \\
    \node [block6] (ideal6) {Output: $\ell$}; \\
};
\begin{scope}[every path/.style=line]
    \path (idealNSoA) -- (ideal1);
    \path (ideal1) -- (ideal2);
    \path (ideal2) -- (ideal3);
    \path (ideal3) -- (ideal4);
    \path (ideal4) -- (ideal5);
    \path (ideal5) -- (ideal6);
\end{scope}
\end{tikzpicture}
\hspace{1cm}
\begin{tikzpicture}[
    remember picture,
    auto,
    decision/.style={diamond, draw=blue, thick, fill=blue!20, text width=4.5em, align=flush center, inner sep=1pt},
    block1/.style={ellipse, draw=blue, thick, fill=blue!20, text width=7em, align=center, rounded corners, minimum height=2em},
    block2/.style={rectangle, draw=blue, thick, fill=gray!20, text width=11em, align=center, rounded corners, minimum height=2em},
    block3/.style={rectangle, draw=blue, thick, fill=gray!20, text width=11em, align=center, rounded corners, minimum height=2em},
    block4/.style={rectangle, draw=blue, thick, fill=gray!20, text width=11em, align=center, rounded corners, minimum height=2em},
    block5/.style={rectangle, draw=blue, thick, fill=gray!20, text width=11em, align=center, rounded corners, minimum height=2em},
    block6/.style={rectangle, draw=blue, thick, fill=orange!20, text width=10em, align=center, rounded corners, minimum height=2em},
    line/.style={draw, thick, -latex', shorten >=2pt},
    cloud/.style={draw=red, thick, ellipse, fill=red!20, minimum height=2em}
]
\matrix [row sep=7mm, column sep=5mm] {
    \node [block1] (proof1) {Proof Sketch}; \\
    \node [block2] (proof2) {Lemma \ref{boundAsum}: Bound $\|A_{\mathrm{sum}}-\Omega_{\mathrm{sum}}\|$}; \\
    \node [block3] (proof3) {Lemma \ref{BoundUhatU}: Bound $\|\hat{U}Q-U\|_{F}$}; \\
    \node [block4] (proof4) {Lemma \ref{BoundUstarhatUsatr}: Bound $\|\hat{U}_{*}Q-U_{*}\|_{F}$}; \\
    \node [block5] (proof5) {Theorem \ref{mainNSoA}: Bound $\hat{f}_{NSoA}$}; \\
    \node [block6] (proof6) {Output: theoretical upper bound of $\hat{f}_{NSoA}$}; \\
};
\begin{scope}[every path/.style=line]
    \path (proof1) -- (proof2);
    \path (proof2) -- (proof3);
    \path (proof3) -- (proof4);
    \path (proof4) -- (proof5);
    \path (proof5) -- (proof6);
\end{scope}
\end{tikzpicture}
\begin{tikzpicture}[remember picture, overlay]
    \draw [dashed, -latex', thick, shorten >=2pt] (ideal2) -- (proof2);
    \draw [dashed, -latex', thick, shorten >=2pt] (ideal3) -- (proof3);
    \draw [dashed, -latex', thick, shorten >=2pt] (ideal4) -- (proof4);
    \draw [dashed, -latex', thick, shorten >=2pt] (ideal5) -- (proof5);
    \draw [dashed, -latex', thick, shorten >=2pt] (ideal6) -- (proof6);
\end{tikzpicture}
\hspace{1cm}
\begin{tikzpicture}[
    remember picture,
    auto,
    decision/.style={diamond, draw=blue, thick, fill=blue!20, text width=4.5em, align=flush center, inner sep=1pt},
    block0/.style={ellipse, draw=blue, thick, fill=blue!20, text width=7em, align=center, rounded corners, minimum height=2em},
    block1/.style={rectangle, draw=blue, thick, fill=cyan!20, text width=11em, align=center, rounded corners, minimum height=2em},
    block2/.style={rectangle, draw=blue, thick, fill=gray!20, text width=11em, align=center, rounded corners, minimum height=2em},
    block3/.style={rectangle, draw=blue, thick, fill=gray!20, text width=11em, align=center, rounded corners, minimum height=2em},
    block4/.style={rectangle, draw=blue, thick, fill=gray!20, text width=11em, align=center, rounded corners, minimum height=2em},
    block5/.style={rectangle, draw=blue, thick, fill=gray!20, text width=11em, align=center, rounded corners, minimum height=2em},
    block6/.style={rectangle, draw=blue, thick, fill=orange!20, text width=4em, align=center, rounded corners, minimum height=2em},
    line/.style={draw, thick, -latex', shorten >=2pt},
    cloud/.style={draw=red, thick, ellipse, fill=red!20, minimum height=2em}
]
\matrix [row sep=7mm, column sep=5mm] {
    \node [block0] (realNSoA) {NSoA}; \\
    \node [block1] (real1) {Input: $\{A_{l}\}^{L}_{l=1}$ and $K$}; \\
    \node [block2] (real2) {Step 1: Get $A_{\mathrm{sum}}$}; \\
    \node [block3] (real3) {Step 2: Get $\hat{U}$ from $A_{\mathrm{sum}}$}; \\
    \node [block4] (real4) {Step 3: Get $\hat{U}_{*}$ from $\hat{U}$}; \\
    \node [block5] (real5) {Step 4: Run K-means on $\hat{U}_{*}$ to get $\hat{\ell}$}; \\
    \node [block6] (real6) {Output: $\hat{\ell}$}; \\
};
\begin{scope}[every path/.style=line]
    \path (realNSoA) -- (real1);
    \path (real1) -- (real2);
    \path (real2) -- (real3);
    \path (real3) -- (real4);
    \path (real4) -- (real5);
    \path (real5) -- (real6);
\end{scope}
\end{tikzpicture}
\begin{tikzpicture}[remember picture, overlay]
    \draw [dashed, -latex', thick, shorten >=2pt] (real2) -- (proof2);
    \draw [dashed, -latex', thick, shorten >=2pt] (real3) -- (proof3);
    \draw [dashed, -latex', thick, shorten >=2pt] (real4) -- (proof4);
    \draw [dashed, -latex', thick, shorten >=2pt] (real5) -- (proof5);
    \draw [dashed, -latex', thick, shorten >=2pt] (real6) -- (proof6);
\end{tikzpicture}
\caption{The overall technical roadmap of the NSoA algorithm.}
\label{RoadmapNSoA}
\end{figure*}
To build and simplify the theoretical guarantees on NSoA's consistency results, we need the following condition which requires a linear growth of the smallest singular value of $\sum_{l\in[L]}B_{l}$.
\begin{assum}\label{Assum11}
$|\lambda_{K}(\sum_{l\in[L]}B_{l})|\geq c_{1}L$ for some constant $c_{1}>0$.
\end{assum}
Here, we use one simple example to explain the rationality of Assumption \ref{Assum11}. Consider the special situation that we generate $L$ independent adjacency matrices from a single DCSBM model, i.e., the case $B_{1}=B_{2}=\ldots=B_{L}$. For this case, we have $|\lambda_{K}(\sum_{l\in[L]})B_{l}|=L|\lambda_{K}(B_{1})|$, which suggests the rationality of Assumption \ref{Assum11}. Assumption \ref{Assum11} is required when we bound the difference between $\hat{U}$ and $U$, where we need to provide a lower bound for $|\lambda_{K}(\Omega_{\mathrm{sum}})|$ to streamline our analysis. Furthermore, the eigenvalue growth condition is common in establishing estimation consistency for community detection methods in multi-layer networks, as exemplified by Assumption 1 in \citep{lei2023bias}, Assumption 2 in \citep{su2023spectral}, and Assumption A in \citep{xu2023covariate}.

When the estimated label vector $\hat{\ell}$ is returned by our NSoA algorithm, we let $\hat{f}_{NSoA}$ be the error rate of NSoA computed by Equation (\ref{ErrorRate}). The following theorem is the main result of NSoA and it provides the upper bound of NSoA's error rate.
\begin{thm}\label{mainNSoA}
Under $\mathrm{MLDCSBM}(Z,\Theta,\mathcal{B})$, when Assumptions \ref{Assum1} and \ref{Assum11} hold, with probability at least $1-O(\frac{1}{n+L})$, we have
\begin{align*}
\hat{f}_{NSoA}=O(\frac{K^{2}\theta^{3}_{\mathrm{max}}\|\theta\|_{1}n_{\mathrm{max}}\mathrm{log}(n+L)}{\theta^{6}_{\mathrm{min}}n^{3}_{\mathrm{min}}L}).
\end{align*}
\end{thm}
At first glance, Theorem \ref{mainNSoA} is complex since it keeps track of too many model parameters. In fact, we benefit from such a complex form since we can analyze the influence of each model parameter on NSoA's performance directly. For example, when we fix $Z,\Theta$, and $\mathbb{B}$, increasing $L$ decreases NSoA's error rate from Theorem \ref{mainNSoA}, and this also suggests the benefit of a multi-layer network since more layers of networks help us better estimate communities. In particular, when $L\rightarrow+\infty$, NSoA's error rate goes to zero, which shows that NSoA consistently detect nodes' communities respective to the number of layers $L$ and suggests that if possible, we can always increase the number of layers for a better community detection. We also see that decreasing $n_{\mathrm{min}}$ increases the error rate because it is well-known that community recovery is hard when the size of one community is too small.

By considering some conditions on model parameters, we can further simplify the bound in Theorem \ref{mainNSoA}.
\begin{cor}\label{CorAsum}
Suppose that the conditions in Theorem \ref{mainNSoA} hold, if we further assume that $K=O(1), \frac{n_{\mathrm{min}}}{n_{\mathrm{max}}}=O(1)$, and $\theta_{\mathrm{min}}=O(\sqrt{\rho}), \theta_{\mathrm{max}}=O(\sqrt{\rho})$ for $\rho\in(0,1]$, we have
\begin{align*}
\hat{f}_{NSoA}=O(\frac{\mathrm{log}(n+L)}{\rho nL}).
\end{align*}
\end{cor}
In Corollary \ref{CorAsum}, $K=O(1)$ means that the number of communities $K$ is fixed, $n_{\mathrm{min}}/n_{\mathrm{max}}=O(1)$ means that community sizes are balanced, and $\theta_{\mathrm{min}}/\theta_{\mathrm{max}}=O(1)$ means that the degree heterogeneity parameters are balanced. Both the condition of a fixed K and the condition of balanced community sizes are frequently encountered in theoretical analyses of community detection algorithms. Examples include the conditions outlined in Corollary 3.1 of \citep{mao2021estimating}, Assumption 1 of \citep{lei2023bias}, Condition 1 in \citep{qing2024bipartite}, and Assumption 1 of \citep{su2023spectral}. From Corollary \ref{CorAsum}, it is clear that NSoA enjoys consistent community recovery because NSoA's error rate decreases to zero as $n$ (or $L$) goes to infinity and NSoA performs better in estimating nodes' community labels with the increase of $n$ (or $L$ or $\rho$). For consistency, we need $\frac{\mathrm{log}(n+L)}{\rho nL}\ll1$, which matches the sparsity requirement in Assumption \ref{Assum1}. $\frac{\mathrm{log}(n+L)}{\rho nL}\ll1$ also gives that $\rho\gg\frac{\mathrm{log}(n+L)}{nL}$, suggesting that $\rho$ should shrink slower than $\frac{\mathrm{log}(n+L)}{nL}$. Meanwhile, if we set $\theta(i)=\sqrt{\rho}>0$ for all $i\in[n]$ such that $\Omega_{l}=\Theta ZB_{l}Z'\Theta=\rho ZB_{l}Z'$ for $l\in[L]$,  MLDCSBM degenerates to MLSBM and Corollary \ref{CorAsum} provides an upper bound on NSoA's error rate under MLSBM. It is evident that as $\rho$ increases, the number of edges in each layer also rises (see Fig.~\ref{Nrho} of Section \ref{sec5}). Hence, we refer to $\rho$ as the sparsity parameter, as it governs the overall sparsity of the multi-layer network.

The provided upper bound for the error rate of the NSoA algorithm, $\hat{f}_{NSoA} = O\left(\frac{\mathrm{log}(n+L)}{\rho nL}\right)$, offers valuable insights into how the performance of the algorithm is influenced by various parameters of the multi-layer network. Here are the details:
\begin{itemize}
  \item The bound suggests that as the number of nodes $n$ increases, the error rate decreases. Intuitively, more nodes provide more information and potential connections, which can lead to better community detection.
  \item Similarly, an increase in the number of layers $L$ reduces the error rate. Multiple layers allow the algorithm to capture more complex relationships and interactions between nodes.
  \item The error rate is inversely proportional to the sparsity parameter $\rho$. As $\rho$ increases (indicating fewer zero edges), the network becomes denser, and the algorithm performs better because there is more connectivity information available.
\end{itemize}
\begin{rem}\label{NoAssumptionsNSoA}
By referring to the proof of Theorem \ref{mainNSoA}, we know that $\hat{f}_{NSoA}=O(\frac{K^{2}\theta^{2}_{\mathrm{max}}n_{\mathrm{max}}\|A_{\mathrm{sum}}-\Omega_{\mathrm{sum}}\|^{2}}{\theta^{6}_{\mathrm{min}}n^{3}_{\mathrm{min}}\lambda^{2}_{K}(\sum_{l\in[L]}B_{l})})$ if we do not consider Assumptions \ref{Assum1} and \ref{Assum11}, as these two assumptions are essential for bounding $\|A_{\mathrm{sum}}-\Omega_{\mathrm{sum}}\|$ and $\lambda_{K}(\sum_{l\in[L]}B_{l})$, respectively. This bound underscores the versatility of our NSoA algorithm, as it does not impose any constraints on the structure of multi-layer networks. Furthermore, under the conditions stated in Corollary \ref{CorAsum}, this bound simplifies to $O(\frac{\|A_{\mathrm{sum}}-\Omega_{\mathrm{sum}}\|^{2}}{\rho^{2}n^{2}\lambda^{2}_{K}(\sum_{l\in[L]}B_{l})})$. However, it is evident that neither of these bounds explicitly reveal the individual impacts of the number of nodes $n$, the number of layers $L$, or the sparsity parameter $\rho$ on NSoA's performance, as these factors may be inherently embedded within $\|A_{\mathrm{sum}}-\Omega_{\mathrm{sum}}\|$ and $\lambda_{K}(\sum_{l\in[L]}B_{l})$. Consequently, to further study the influences of $n$, $L$, and $\rho$ (and, more crucially, to assess the estimation consistency of NSoA), we must invoke Assumption \ref{Assum1} to establish an upper bound for $\|A_{\mathrm{sum}}-\Omega_{\mathrm{sum}}\|$ and Assumption \ref{Assum11} to establish a lower bound for $|\lambda_{K}(\sum_{l\in[L]}B_{l})|$. Indeed, under Assumption \ref{Assum1}, $n$, $L$, and $\rho$ are explicitly incorporated into $\|A_{\mathrm{sum}}-\Omega_{\mathrm{sum}}\|$ by Lemma \ref{boundAsum}; similarly, under Assumption \ref{Assum11}, $L$ is also incorporated into $|\lambda_{K}(\sum_{l\in[L]}B_{l})|$.
\end{rem}
\subsection{Consistency results for NDSoSA}
In this subsection, we provide an upper bound on the error rate of NDSoSA in terms of MLDCSBM model parameters and this ensures the estimation consistency of NDSoSA in the scenario where either the number of nodes $n$ or layers $L$ approaches infinity. For the theoretical bound of NDSoSA, we need the following assumption which controls the overall sparsity.
\begin{assum}\label{Assum2}
 $\theta_{\mathrm{max}}\|\theta\|_{1}\|\theta\|^{2}_{F}L\geq\mathrm{log}(n+L)$.
\end{assum}
Lemma \ref{boundSsum} functions similar to Lemma \ref{boundAsum} and it bounds $\|S_{\mathrm{sum}}-\tilde{S}_{\mathrm{sum}}\|$ under MLDCSBM, i.e., $S_{\mathrm{sum}}$ is close to $\tilde{S}_{\mathrm{sum}}$ in terms of spectral norm.
\begin{lem}\label{boundSsum}
Under $\mathrm{MLDCSBM}(Z,\Theta,\mathcal{B})$, when Assumption \ref{Assum2} holds, with probability at least $1-O(\frac{1}{n+L})$, we have
\begin{align*}
\|S_{\mathrm{sum}}-\tilde{S}_{\mathrm{sum}}\|=O(\sqrt{\theta_{\mathrm{max}}\|\theta\|_{1}\|\theta\|^{2}_{F}L\mathrm{log}(n+L)})+O(\theta^{2}_{\mathrm{max}}\|\theta\|^{2}_{F}L).
\end{align*}
\end{lem}
\begin{rem}
By the proof of Lemma \ref{boundSsum}, if Assumption \ref{Assum2} becomes $\theta_{\mathrm{max}}\|\theta\|_{1}\|\theta\|^{2}_{F}L\geq\mathrm{log}(n)$, Lemma \ref{boundSsum} becomes: with probability at least $1-O(\frac{1}{n})$, we have  $\|S_{\mathrm{sum}}-\tilde{S}_{\mathrm{sum}}\|=O(\sqrt{\theta_{\mathrm{max}}\|\theta\|_{1}\|\theta\|^{2}_{F}L\mathrm{log}(n})+O(\theta^{2}_{\mathrm{max}}\|\theta\|^{2}_{F}L)$.
\end{rem}
To establish NDSoSA's theoretical guarantees on consistency, we need the following assumption which requires a linear growth of the smallest singular value of $\sum_{l\in[L]}B^{2}_{l}$ respective to the number of layers $L$, where this assumption aligns with Assumption 1 (b) of \citep{lei2023bias}, as both our NDSoSA and the SoS-Debias method presented in \citep{lei2023bias} are constructed utilizing $S_{\mathrm{sum}}$.
\begin{assum}\label{Assum22}
$|\lambda_{K}(\sum_{l\in[L]}B^{2}_{l})|\geq c_{2}L$ for some constant $c_{2}>0$.
\end{assum}
Similar to Assumption \ref{Assum11}, the rationality of Assumption \ref{Assum22} can be understood for the case $B_{1}=B_{2}=\ldots=B_{L}$. Let $\hat{f}_{NDSoSA}$ be the error rate of NDSoSA calculated via Equation (\ref{ErrorRate}). The following theorem presents the main result of NDSoSA.
\begin{thm}\label{mainNDSoSA}
Under $\mathrm{MLDCSBM}(Z,\Theta,\mathcal{B})$, when Assumptions \ref{Assum2} and \ref{Assum22} hold, with probability at least $1-O(\frac{1}{n+L})$, we have
\begin{align*}
\hat{f}_{NDSoSA}=O(\frac{K^{2}\theta^{3}_{\mathrm{max}}\|\theta\|_{1}\|\theta\|^{2}_{F}n_{\mathrm{max}}\mathrm{log}(n+L)}{\theta^{10}_{\mathrm{min}}n^{5}_{\mathrm{min}}L})+O(\frac{K^{2}\theta^{6}_{\mathrm{max}}\|\theta\|^{4}_{F}n_{\mathrm{max}}}{\theta^{10}_{\mathrm{min}}n^{5}_{\mathrm{min}}}).
\end{align*}
\end{thm}
The upper bound of NDSoSA's theoretical error rate in Theorem \ref{mainNDSoSA} is written in terms of several MLDCSBM model parameters. Similar to Theorem \ref{mainNSoA}, we can analyze the influence of these model parameters on NDSoSA's performance and we omit it here. Fig.~\ref{RoadmapNDSoSA} functions similar to Fig.~\ref{RoadmapNSoA} and it provides NDSoSA's overall technical roadmap.
\begin{figure*}[htbp]
\centering
\begin{tikzpicture}[
    remember picture,
    auto,
    decision/.style={diamond, draw=blue, thick, fill=blue!20, text width=4.5em, align=flush center, inner sep=1pt},
    block0/.style={ellipse, draw=blue, thick, fill=blue!20, text width=7em, align=center, rounded corners, minimum height=2em},
    block1/.style={rectangle, draw=blue, thick, fill=cyan!20, text width=11em, align=center, rounded corners, minimum height=2em},
    block2/.style={rectangle, draw=blue, thick, fill=gray!20, text width=11em, align=center, rounded corners, minimum height=2em},
    block3/.style={rectangle, draw=blue, thick, fill=gray!20, text width=11em, align=center, rounded corners, minimum height=2em},
    block4/.style={rectangle, draw=blue, thick, fill=gray!20, text width=11em, align=center, rounded corners, minimum height=2em},
    block5/.style={rectangle, draw=blue, thick, fill=gray!20, text width=11em, align=center, rounded corners, minimum height=2em},
    block6/.style={rectangle, draw=blue, thick, fill=orange!20, text width=4em, align=center, rounded corners, minimum height=2em},
    line/.style={draw, thick, -latex', shorten >=2pt},
    cloud/.style={draw=red, thick, ellipse, fill=red!20, minimum height=2em}
]
\matrix [row sep=7mm, column sep=5mm] {
    \node [block0] (idealNDSoSA) {Ideal NDSoSA}; \\
    \node [block1] (ideal1) {Input: $\{\Omega_{l}\}^{L}_{l=1}$ and $K$}; \\
    \node [block2] (ideal2) {Step 1: Get $\tilde{S}_{\mathrm{sum}}$}; \\
    \node [block3] (ideal3) {Step 2: Get $V$ from $\tilde{S}_{\mathrm{sum}}$}; \\
    \node [block4] (ideal4) {Step 3: Get $V_{*}$ from $V$}; \\
    \node [block5] (ideal5) {Step 4: Run K-means on $V_{*}$ to get $\ell$}; \\
    \node [block6] (ideal6) {Output: $\ell$}; \\
};
\begin{scope}[every path/.style=line]
    \path (idealNDSoSA) -- (ideal1);
    \path (ideal1) -- (ideal2);
    \path (ideal2) -- (ideal3);
    \path (ideal3) -- (ideal4);
    \path (ideal4) -- (ideal5);
    \path (ideal5) -- (ideal6);
\end{scope}
\end{tikzpicture}
\hspace{1cm}
\begin{tikzpicture}[
    remember picture,
    auto,
    decision/.style={diamond, draw=blue, thick, fill=blue!20, text width=4.5em, align=flush center, inner sep=1pt},
    block1/.style={ellipse, draw=blue, thick, fill=blue!20, text width=7em, align=center, rounded corners, minimum height=2em},
    block2/.style={rectangle, draw=blue, thick, fill=gray!20, text width=11em, align=center, rounded corners, minimum height=2em},
    block3/.style={rectangle, draw=blue, thick, fill=gray!20, text width=11em, align=center, rounded corners, minimum height=2em},
    block4/.style={rectangle, draw=blue, thick, fill=gray!20, text width=11em, align=center, rounded corners, minimum height=2em},
    block5/.style={rectangle, draw=blue, thick, fill=gray!20, text width=11em, align=center, rounded corners, minimum height=2em},
    block6/.style={rectangle, draw=blue, thick, fill=orange!20, text width=10em, align=center, rounded corners, minimum height=2em},
    line/.style={draw, thick, -latex', shorten >=2pt},
    cloud/.style={draw=red, thick, ellipse, fill=red!20, minimum height=2em}
]
\matrix [row sep=7mm, column sep=5mm] {
    \node [block1] (proof1) {Proof Sketch}; \\
    \node [block2] (proof2) {Lemma \ref{boundSsum}: Bound $\|S_{\mathrm{sum}}-\tilde{S}_{\mathrm{sum}}\|$}; \\
    \node [block3] (proof3) {Lemma \ref{BoundVhatV}: Bound $\|\hat{V}Q_{s}-V\|_{F}$}; \\
    \node [block4] (proof4) {Lemma \ref{BoundVstarhatVsatr}: Bound $\|\hat{V}_{*}Q_{s}-V_{*}\|_{F}$}; \\
    \node [block5] (proof5) {Theorem \ref{mainNDSoSA}: Bound $\hat{f}_{NDSoSA}$}; \\
    \node [block6] (proof6) {Output: theoretical upper bound of $\hat{f}_{NDSoSA}$}; \\
};
\begin{scope}[every path/.style=line]
    \path (proof1) -- (proof2);
    \path (proof2) -- (proof3);
    \path (proof3) -- (proof4);
    \path (proof4) -- (proof5);
    \path (proof5) -- (proof6);
\end{scope}
\end{tikzpicture}
\begin{tikzpicture}[remember picture, overlay]
    \draw [dashed, -latex', thick, shorten >=2pt] (ideal2) -- (proof2);
    \draw [dashed, -latex', thick, shorten >=2pt] (ideal3) -- (proof3);
    \draw [dashed, -latex', thick, shorten >=2pt] (ideal4) -- (proof4);
    \draw [dashed, -latex', thick, shorten >=2pt] (ideal5) -- (proof5);
    \draw [dashed, -latex', thick, shorten >=2pt] (ideal6) -- (proof6);
\end{tikzpicture}
\hspace{1cm}
\begin{tikzpicture}[
    remember picture,
    auto,
    decision/.style={diamond, draw=blue, thick, fill=blue!20, text width=4.5em, align=flush center, inner sep=1pt},
    block0/.style={ellipse, draw=blue, thick, fill=blue!20, text width=7em, align=center, rounded corners, minimum height=2em},
    block1/.style={rectangle, draw=blue, thick, fill=cyan!20, text width=11em, align=center, rounded corners, minimum height=2em},
    block2/.style={rectangle, draw=blue, thick, fill=gray!20, text width=11em, align=center, rounded corners, minimum height=2em},
    block3/.style={rectangle, draw=blue, thick, fill=gray!20, text width=11em, align=center, rounded corners, minimum height=2em},
    block4/.style={rectangle, draw=blue, thick, fill=gray!20, text width=11em, align=center, rounded corners, minimum height=2em},
    block5/.style={rectangle, draw=blue, thick, fill=gray!20, text width=11em, align=center, rounded corners, minimum height=2em},
    block6/.style={rectangle, draw=blue, thick, fill=orange!20, text width=4em, align=center, rounded corners, minimum height=2em},
    line/.style={draw, thick, -latex', shorten >=2pt},
    cloud/.style={draw=red, thick, ellipse, fill=red!20, minimum height=2em}
]
\matrix [row sep=7mm, column sep=5mm] {
    \node [block0] (realNDSoSA) {NDSoSA}; \\
    \node [block1] (real1) {Input: $\{A_{l}\}^{L}_{l=1}$ and $K$}; \\
    \node [block2] (real2) {Step 1: Get $S_{\mathrm{sum}}$}; \\
    \node [block3] (real3) {Step 2: Get $\hat{V}$ from $S_{\mathrm{sum}}$}; \\
    \node [block4] (real4) {Step 3: Get $\hat{V}_{*}$ from $\hat{V}$}; \\
    \node [block5] (real5) {Step 4: Run K-means on $\hat{V}_{*}$ to get $\hat{\ell}$}; \\
    \node [block6] (real6) {Output: $\hat{\ell}$}; \\
};
\begin{scope}[every path/.style=line]
    \path (realNDSoSA) -- (real1);
    \path (real1) -- (real2);
    \path (real2) -- (real3);
    \path (real3) -- (real4);
    \path (real4) -- (real5);
    \path (real5) -- (real6);
\end{scope}
\end{tikzpicture}
\begin{tikzpicture}[remember picture, overlay]
    \draw [dashed, -latex', thick, shorten >=2pt] (real2) -- (proof2);
    \draw [dashed, -latex', thick, shorten >=2pt] (real3) -- (proof3);
    \draw [dashed, -latex', thick, shorten >=2pt] (real4) -- (proof4);
    \draw [dashed, -latex', thick, shorten >=2pt] (real5) -- (proof5);
    \draw [dashed, -latex', thick, shorten >=2pt] (real6) -- (proof6);
\end{tikzpicture}
\caption{The overall technical roadmap of the NDSoSA algorithm.}
\label{RoadmapNDSoSA}
\end{figure*}
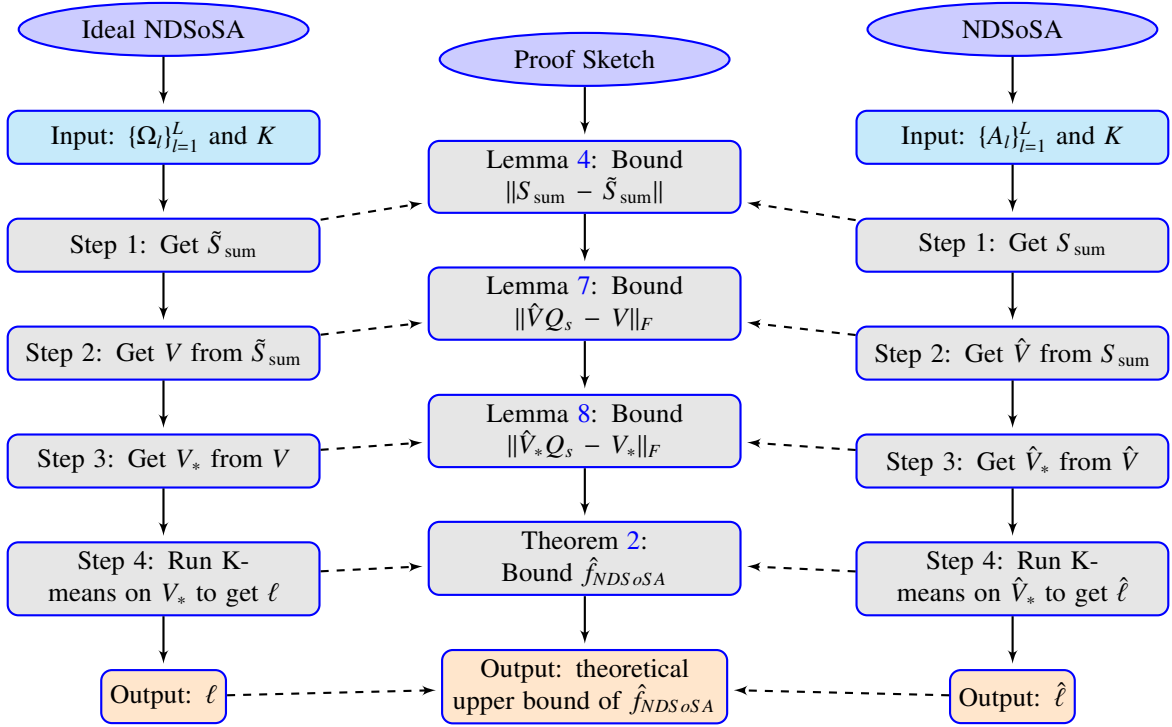
The following corollary further simplifies Theorem \ref{mainNDSoSA} by adding more conditions to these model parameters.
\begin{cor}\label{CorSsum}
Under the conditions of Theorem \ref{mainNDSoSA}, suppose $K=O(1), \frac{n_{\mathrm{min}}}{n_{\mathrm{max}}}=O(1)$, and $\theta_{\mathrm{min}}=O(\sqrt{\rho}), \theta_{\mathrm{max}}=O(\sqrt{\rho})$ for $\rho\in(0,1]$, we have
\begin{align*}
\hat{f}_{NDSoSA}=O(\frac{\mathrm{log}(n+L)}{\rho^{2}n^{2}L})+O(\frac{1}{n^{2}}).
\end{align*}
\end{cor}
Corollary \ref{CorSsum} guarantees that NDSoSA enjoys community estimation consistency since NDSoSA's error rate goes to zero as $n\rightarrow+\infty$ (or $L\rightarrow+\infty$). From this corollary, we also observe that NDSoSA yields more precise community detection results as $n$ (or $L$ or $\rho$) increases.
\begin{rem}
(Comparison to \citep{lei2023bias}) When $\theta(i)=\sqrt{\rho}$ for all $i\in[n]$ such that MLDCSBM reduces to MLSBM, the bound in Corollary \ref{CorSsum} is consistent with that of Theorem 1 in \citep{lei2023bias} and Assumption \ref{Assum2} becomes $\rho^{2}n^{2}L\geq\mathrm{log}(n+L)$ which also matches the sparsity requirement in Theorem 1 in \citep{lei2023bias}. The above analysis suggests the optimality of our theoretical results. Meanwhile, unlike \citep{lei2023bias} whose theoretical results rely on some complex tail probabilities bounds for matrix linear combinations with matrix-valued coefficients and matrix-valued quadratic form, we obtain the same theoretical results by a simple trick: we simply decompose $S_{\mathrm{sum}}-\tilde{S}_{\mathrm{sum}}$ into two parts such that the expectation of one part is a zero matrix while the other part is deterministic. For the part with zero mean, we can apply the matrix Bernstein inequality to bound its spectral norm directly. For the deterministic part, we can directly calculate its upper bound of spectral norm. In this way, we can bound the spectral norm of $S_{\mathrm{sum}}-\tilde{S}_{\mathrm{sum}}$ and this greatly reduces the complexity of our theoretical analysis. For details, please refer to the proof of Lemma \ref{boundSsum}.
\end{rem}
\begin{rem}
Similar to Remark \ref{NoAssumptionsNSoA}, when disregarding Assumptions \ref{Assum2} and \ref{Assum22}, we obtain $\hat{f}_{NDSoSA}=O(\frac{\theta^{2}_{\mathrm{max}}K^{2}n_{\mathrm{max}}\|S_{\mathrm{sum}}-\tilde{S}_{\mathrm{sum}}\|^{2}}{\theta^{10}_{\mathrm{min}}n^{5}_{\mathrm{min}}\lambda^{2}_{K}(\sum_{l\in[L]}B^{2}_{l})})$ from the proof of Theorem \ref{mainNDSoSA}. This bound further simplifies to $O(\frac{\|S_{\mathrm{sum}}-\tilde{S}_{\mathrm{sum}}\|^{2}}{\rho^{4}n^{4}\lambda^{2}_{K}(\sum_{l\in[L]}B^{2}_{l})})$ under conditions stated in Corollary \ref{CorSsum}. Both forms underscore NDSoSA's generality as they do not limit the structure of multi-layer networks. Again, we cannot study the influences of the model parameters $n, L$, and $\rho$ on NDSoSA's performance and estimation consistency without considering Assumptions \ref{Assum2} and \ref{Assum22}. These assumptions are necessary to establish an upper bound for $\|S_{\mathrm{sum}}-\tilde{S}_{\mathrm{sum}}\|$ and a lower bound $|\lambda_{K}(\sum_{l\in[L]}B^{2}_{l})|$, respectively.
\end{rem}
\subsection{Comparison between the consistency results of NSoA and NDSoSA}\label{CompareDA}
In this subsection, we provide a detailed comparison between the consistency results of NSoA and NDSoSA. Since Theorems \ref{mainNSoA} and \ref{mainNDSoSA} are written in terms of MLDCSBM model parameters, it is not easy to directly compare their upper bounds. Instead, to simplify the comparison, we primarily focus on the settings in Corollaries \ref{CorAsum} and \ref{CorSsum} when MLDCSBM reduces to MLSBM, i.e., the case   $\theta(i)=\sqrt{\rho}$ for all $i\in[n]$. For this case, Assumption \ref{Assum1} turns to be $\rho n\geq\frac{\mathrm{log}(n+L)}{L}$ and Assumption \ref{Assum2} turns to be $\rho n\geq\sqrt{\frac{\mathrm{log}(n+L)}{L}}$. Table \ref{CompareTwoMethods} summarizes the assumptions and error rates of NSoA and NDSoSA for this case. By analyzing Table \ref{CompareTwoMethods}, we have the following conclusions:
\begin{table*}[h!]
	\centering
	\caption{Consistency results of NSoA and NDSoSA when $K=O(1), \frac{n_{\mathrm{min}}}{n_{\mathrm{max}}}=O(1)$, and $\theta(i)=\sqrt{\rho}$ for all $i\in[n]$.}
	\label{CompareTwoMethods}
\begin{tabular}{cccccccccc}
\hline
Methods&Requirement on sparsity&Requirement on $\{B_{l}\}^{L}_{l=1}$&Error rate\\
\hline
NSoA&$\rho n\geq\frac{\mathrm{log}(n+L)}{L}$&$|\lambda_{K}(\sum_{l\in[L]}B_{l})|\geq c_{1}L$&$\frac{\mathrm{log}(n+L)}{\rho nL}$\\
NDSoSA&$\rho n\geq\sqrt{\frac{\mathrm{log}(n+L)}{L}}$&$|\lambda_{K}(\sum_{l\in[L]}B^{2}_{l})|\geq c_{2}L$&$\frac{\mathrm{log}(n+L)}{\rho^{2} n^{2}L}+\frac{1}{n^{2}}$\\
\hline
\end{tabular}
\end{table*}

\begin{itemize}
  \item NSoA and NDSoSA have no significant difference for the requirement on the $L$ block connectivity matrices $B_{1}, B_{2},\ldots,B_{L}$.
  \item When $L\leq\mathrm{log}(n+L)$, NSoA's requirement on the sparsity (i.e., $\rho$) of a multi-layer network is stronger than that of NDSoSA since $\frac{\mathrm{log}(n+L)}{L}\geq\sqrt{\frac{\mathrm{log}(n+L)}{L}}\Leftrightarrow\mathrm{log}(n+L)\geq L$. Similarly, when $L\geq\mathrm{log}(n+L)$, NSoA's requirement on the sparsity of a multi-layer network is weaker than that of NDSoSA. For a better understanding of the sparsity requirements of NSoA and NDSoSA, Fig.~\ref{z12d} displays the sparsity regions of the two methods as well as the region of $\frac{\mathrm{log}(n+L)}{nL}-\sqrt{\frac{\mathrm{log}(n+L)}{n^{2}L}}$. Panels (a) and (b) of Fig.~\ref{z12d} suggest that the sparsity requirements for both NSoA and NDSoSA concerning the sparsity parameter $\rho$ are mild, as $\rho$ can be extremely small when $n$ and $L$ are large. From Panel (c) of this figure, it is evident that NSoA's sparsity requirement is stronger than NDSoSA's only within a narrow region where $L\leq\mathrm{log}(n+L)$ (i.e., $\rho_{d}\geq0$), as this inequality holds primarily for small values of $L$. Conversely, in the wide area where $L\geq\mathrm{log}(n+L)$ (i.e., $\rho_{d}\leq0$), NSoA's sparsity requirement is weaker than that of NDSoSA.

\begin{figure*}
\centering
\subfigure[]{\includegraphics[width=0.33\textwidth]{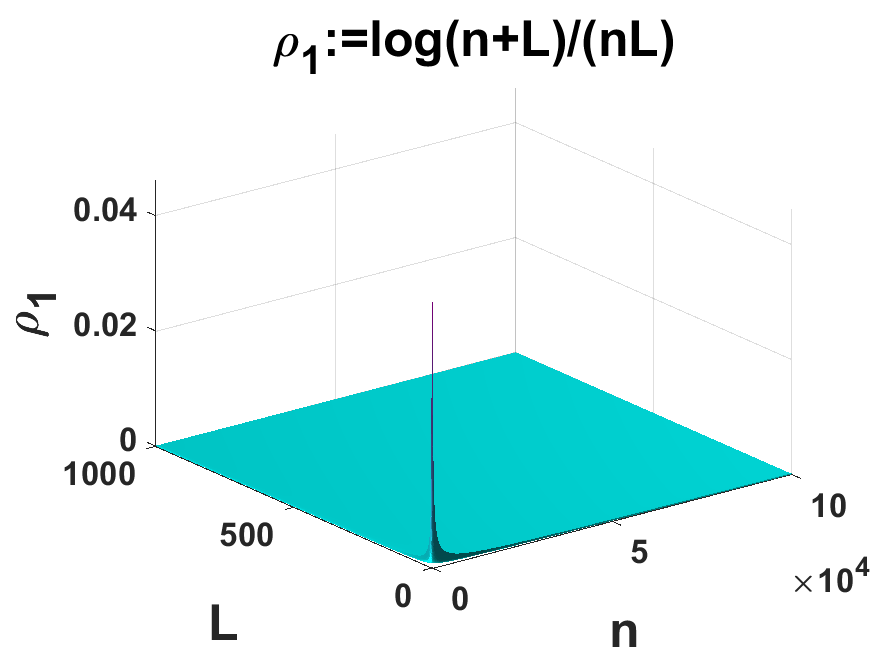}}
\subfigure[]{\includegraphics[width=0.33\textwidth]{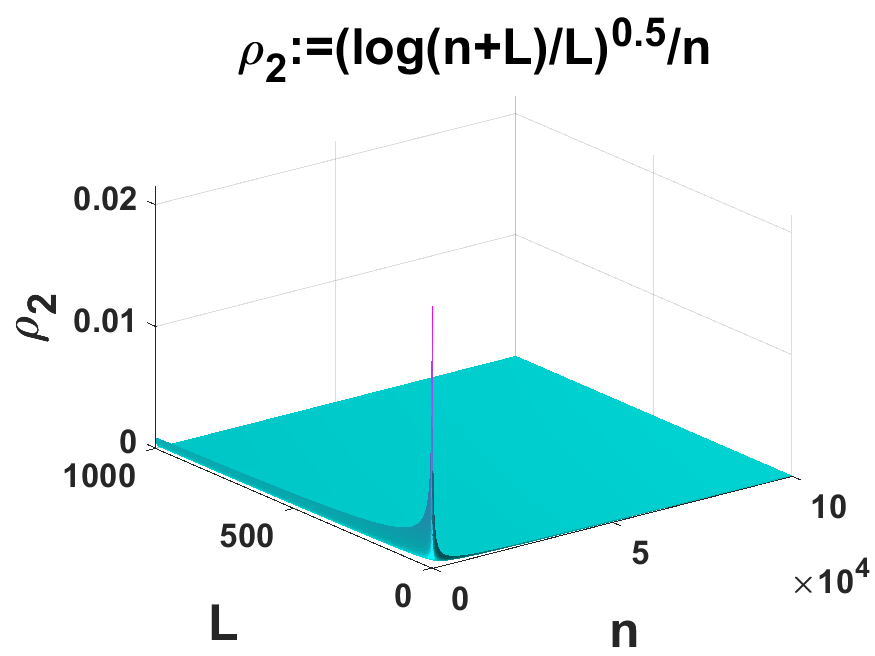}}
\subfigure[]{\includegraphics[width=0.33\textwidth]{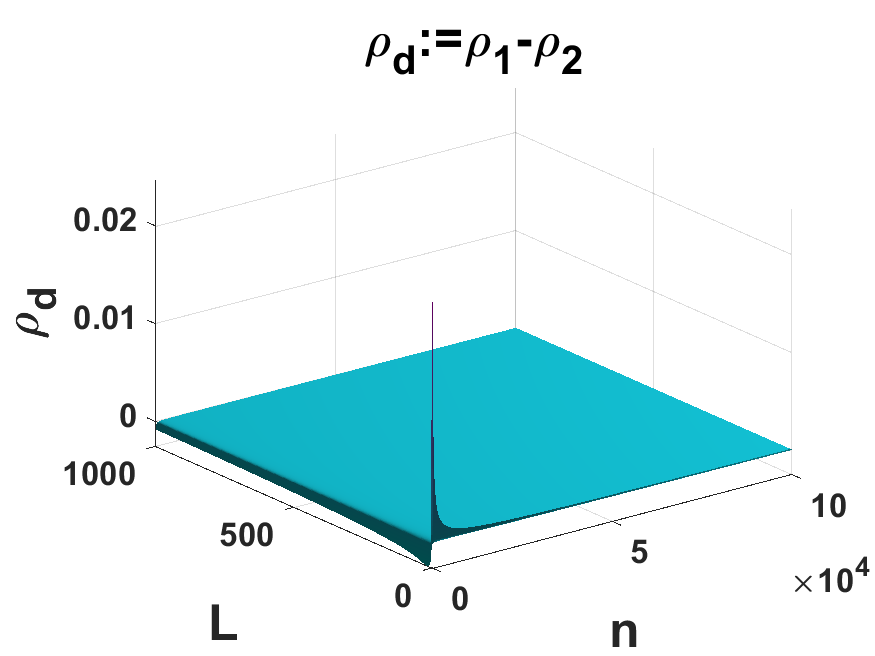}}
\caption{Panel (a): the region (located above the surface of $\rho_{1}$) that satisfies NSoA's requirement on the sparsity parameter $\rho$. Panel (b): the region (located above the surface of $\rho_{2}$) that satisfies NDSoSA's requirement on the sparsity parameter $\rho$. Panel (c): the region of $\rho_{d}:=\rho_{1}-\rho_{2}$. For all three panels, $n$ ranges in $\{100,200,300,\ldots,100000\}$ and $L$ ranges in $\{1,2,3,\ldots,1000\}$.}
\label{z12d} 
\end{figure*}
   \item Since  $\frac{\mathrm{log}(n+L)}{\rho^{2}n^{2}L}\geq\frac{1}{n^{2}}\Leftrightarrow\rho\leq\sqrt{\frac{\mathrm{log}(n+L)}{L}}$, when $\rho\leq\sqrt{\frac{\mathrm{log}(n+L)}{L}}$, NDSoSA's error rate is at the order $\frac{\mathrm{log}(n+L)}{\rho^{2}n^{2}L}$. Next, we compare NDSoSA's error rate $\frac{\mathrm{log}(n+L)}{\rho^{2}n^{2}L}$ with NSoA's error rate $\frac{\mathrm{log}(n+L)}{\rho nL}$:
       \begin{itemize}
         \item Case 1: NDSoSA's error rate is much smaller than that of NSoA. For this case ,we have $\frac{\mathrm{log}(n+L)}{\rho^{2}n^{2}L}\ll\frac{\mathrm{log}(n+L)}{\rho nL}\Leftrightarrow \rho \gg\frac{1}{n}$. Since NSoA's error rate is no larger than 1, we have $\frac{\mathrm{log}(n+L)}{\rho^{2}n^{2}L}\ll1\Leftrightarrow\rho\gg\frac{1}{n}\sqrt{\frac{\mathrm{log}(n+L)}{L}}$. So, when NDSoSA's error rate is much smaller than that of NSoA, $\rho$ should satisfy $\frac{1}{n}\mathrm{max}(1,\sqrt{\frac{\mathrm{log}(n+L)}{L}})\ll\rho\leq\sqrt{\frac{\mathrm{log}(n+L)}{L}}$, which gives that $\frac{1}{n}\ll\rho\leq\sqrt{\frac{\mathrm{log}(n+L)}{L}}$ when $L\geq\mathrm{log}(n+L)$ and $\frac{1}{n}\sqrt{\frac{\mathrm{log}(n+L)}{L}}\ll\rho\leq1$ when $L\leq\mathrm{log}(n+L)$.
         \item Case 2: NSoA's error rate is much smaller than that of NDSoSA. For this case, we have $\frac{\mathrm{log}(n+L)}{\rho^{2}n^{2}L}\gg\frac{\mathrm{log}(n+L)}{\rho nL}\Leftrightarrow \rho \ll\frac{1}{n}$. Because NDSoSA's error rate is less than 1, we have $\frac{\mathrm{log}(n+L)}{\rho nL}\ll1\Leftrightarrow\rho\gg\frac{1}{n}\frac{\mathrm{log}(n+L)}{L}$. Combine $\rho\ll\frac{1}{n}$ with $\rho\gg\frac{1}{n}\frac{\mathrm{log}(n+L)}{L}$, we get $\mathrm{log}(n+L)\ll L\Leftrightarrow\mathrm{log}(n)\ll L$. Therefore, when NSoA significantly outperforms NDoSA, $L$ and $\rho$ should satisfy $L\gg\mathrm{log}(n)$ and $\frac{1}{n}\frac{\mathrm{log}(n+L)}{L}\ll \rho\ll\frac{1}{n}$, respectively. In fact, the requirement $\rho$ should satisfy in Case 2 is much stricter than that of Case 1 since the number of layers $L$ should be set much larger than $\mathrm{log}(n)$ when $\frac{1}{n}\frac{\mathrm{log}(n+L)}{L}\ll \rho\ll\frac{1}{n}$ holds. For example, if we set $\rho=\frac{1}{\beta n}$ for $\beta\gg1$ such that $\rho\ll\frac{1}{n}$ holds. Since $\rho\gg\frac{1}{n}\frac{\mathrm{log}(n+L)}{L}$, we have $L\gg\beta\mathrm{log}(n+L)$. Meanwhile, when $\rho=\frac{1}{\beta n}$, NDSoSA's error rate is at the order $\frac{\beta^{2}\mathrm{log}(n+L)}{L}$ which should be less than 1, and it gives that $L\geq\beta^{2}\mathrm{log}(n+L)$. Therefore, when $\rho=\frac{1}{\beta n}$ for $\beta\gg1$, to make NSoA significantly outperform NDSoSA, the number of layers $L$ should satisfy $L\gg\beta\mathrm{log}(n+L)$ and $L\geq\beta^{2}\mathrm{log}(n+L)$. If we set $\beta=100$, $L$ should be larger than $10000\mathrm{log}(n+L)$; if we set $\beta=1000$, $L$ should be larger than $1000000\mathrm{log}(n+L)$. However, such a requirement on the number of layers $L$ is impractical because $L$ is always a small or median number in real-world multi-layer networks. Furthermore, even when $L$ is set sufficiently large to guarantee that NSoA significantly outperforms NDSoSA, Corollary \ref{CorSsum} suggests that for large $L$, NDSoSA's error rate becomes negligible. This indicates that increasing $L$ to a large value results in only minimal and even negligible advantages for NSoA over NDSoSA.
       \end{itemize}
   \item Since  $\frac{\mathrm{log}(n+L)}{\rho^{2}n^{2}L}\leq\frac{1}{n^{2}}\Leftrightarrow\rho\geq\sqrt{\frac{\mathrm{log}(n+L)}{L}}$, when $\rho\geq\sqrt{\frac{\mathrm{log}(n+L)}{L}}$ (note that since $\rho\leq1$, $L$ should satisfy $L\geq\mathrm{log}(n+L)$), NDSoSA's error rate is at the order $\frac{1}{n^{2}}$. Next, we compare NDSoSA's error rate $\frac{1}{n^{2}}$ with NSoA's error rate $\frac{\mathrm{log}(n+L)}{\rho nL}$:
       \begin{itemize}
         \item Case 3: NDSoSA's error rate is much smaller than that of NSoA. For this case ,we have $\frac{1}{n^{2}}\ll\frac{\mathrm{log}(n+L)}{\rho nL}\Leftrightarrow \rho \ll\frac{n\mathrm{log}(n+L)}{L}$. Since $\rho\geq\sqrt{\frac{\mathrm{log}(n+L)}{L}}$, we have $\sqrt{\frac{\mathrm{log}(n+L)}{L}}\ll\frac{n\mathrm{log}(n+L)}{L}\Leftrightarrow L\ll n^{2}\mathrm{log}(n+L)$. So, when NDSoSA performs betters than NSoA, $\rho$ and $L$ should satisfy $\sqrt{\frac{\mathrm{log}(n+L)}{L}}\leq\rho\ll\frac{n\mathrm{log}(n+L)}{L}$ and $\mathrm{log}(n+L)\leq L\ll n^{2}\mathrm{log}(n+L)$, respectively.
         \item Case 4: NSoA's error rate is much smaller than that of NDSoSA. This gives $\frac{1}{n^{2}}\gg\frac{\mathrm{log}(n+L)}{\rho n L}\Leftrightarrow\rho\gg\frac{n\mathrm{log}(n+L)}{L}$. Since $\rho\leq1$, we have $L\gg n\mathrm{log}(n+L)$. For this case, though NSoA performs much better than NDSoSA, NDSoSA still performs satisfactorily for large $n$ because its error rate is at the order $\frac{1}{n^{2}}$.
       \end{itemize}
\end{itemize}

In conclusion, upon a thorough analysis of the preceding discussion, it becomes apparent that although the bias-adjusted spectral clustering algorithm NDSoSA may not consistently demonstrate superior performance compared to the spectral clustering algorithm NSoA in terms of network sparsity requirements and error rates, NDSoSA remains a preferable choice. This preference is primarily attributed to the fact that, in the context of real multi-layer data, the circumstances where NSoA's error rate significantly exceeds that of NDSoSA often necessitate an unreasonably high value for the number of layers $L$, which is impractical in most scenarios. The aforementioned analysis leads us to confidently assert that NDSoSA nearly always outperforms NSoA in performance.
\section{Estimation of the number of communities $K$}\label{secK}
Recall that the inputs for our NSoA and NDSoSA are the $L$ adjacency matrices $\{A_{l}\}^{L}_{l=1}$ and the number of communities $K$. In real-world multi-layer networks, $K$ is typically unknown and must be predetermined. In this section, we present a method for estimating $K$ in multi-layer networks. Our approach hinges on the averaged modularity metric introduced in \cite{paul2021null}, which assesses the quality of community detection in multi-layer networks when ground-truth community labels are unavailable. This metric is defined as follows:
\begin{align*}
Q_{MNavrg}=\frac{1}{L}\sum_{l\in[L]}\sum_{i\in[n]}\sum_{j\in[n]}\frac{1}{2m_{l}}(A_{l}(i,j)-\frac{D_{l}(i,i)D_{l}(j,j)}{2m_{l}})\delta(\hat{\ell}(i),\hat{\ell}(j)),
\end{align*}
where $m_{l}=\frac{1}{2}\sum_{i\in[n]}D_{l}(i,i)$ for $l\in[L]$, and $\delta(\hat{\ell}(i),\hat{\ell}(j))$ is an indicator function that takes a value of 1 when $\hat{\ell}(i)$ equals $\hat{\ell}(j)$ and 0 otherwise. Notably, when $L$ is equal to 1, $Q_{MNavrg}$ simplifies to the well-known Newman-Girvan modularity \citep{newman2004finding}. A higher value of the averaged modularity signifies a better quality of community partition, and therefore, we always favor larger $Q_{MNavrg}$ scores. Similar to the Newman-Girvan modularity, the averaged modularity is only applicable to assortative multi-layer networks \citep{newman2002assortative,newman2003mixing,paul2021null}, where nodes within the same community exhibit a higher number of connections compared to nodes across different communities. Conversely, in disassortative multi-layer networks, nodes within the same community have fewer connections than those across communities. When dealing with real data where the number of communities, $K$, is unknown, we adopt the strategy introduced in \cite{nepusz2008fuzzy} to determine an appropriate value for $K$ in assortative multi-layer networks. This approach involves iteratively incrementing the number of communities and selecting the one that maximizes $Q_{MNavrg}$. Given that most real-world social networks are assortative \citep{newman2002assortative,newman2003mixing,radicchi2004defining}, our methodology for estimating $K$ through maximizing $Q_{MNavrg}$ proves to be a practical and easy-to-implement tool.
\begin{rem}
We should emphasize that our MLDCSBM model is capable of generating a wide range of multi-layer networks, including both assortative and dis-assortative types, due to its lack of restrictions on the elements of the connectivity matrices $\{B_{l}\}^{L}_{l=1}$. Specifically, when the minimum diagonal entry of $B_{l}$ exceeds its maximum off-diagonal entry for all $l\in[L]$, the MLDCSBM model generates assortative multi-layer networks. Conversely, when the maximum diagonal entry of $B_{l}$ is less than its minimum off-diagonal entry for all $l\in[L]$, the model generates disassortative multi-layer networks. Furthermore, our two proposed algorithms, NSoA and NDSoSA, are applicable to any kind of multi-layer network, as they do not impose any limitations on the latent community structure of real-world multi-layer networks.
\end{rem}
\section{Simulations}\label{sec5}
In this section, we evaluate the performance of NSoA and NDSoSA by comparing them with some benchmark methods using several metrics for simulated multi-layer networks generated from MLDCSBM. We compare our NSoA and NDSoSA with the following methods:
\begin{itemize}
  \item NSoSA: \underline{n}ormalized spectral clustering based on \underline{s}um \underline{o}f \underline{s}quared \underline{a}djacency matrices. NSoSA takes $\sum_{l\in[L]}A^{2}_{l}$ in Algorithm \ref{alg:NDSoSA} as input and it does not adjust the bias in $\sum_{l\in[L]}A^{2}_{l}$. Meanwhile, in \ref{MainNSoSA}, we present the consistency results for NSoSA and demonstrate, through theoretical analysis, that NDSoSA consistently outperforms NSoSA.
  \item Sum: spectral clustering based on \underline{sum} of adjacency matrices without normalization. This algorithm estimates communities by applying the K-means algorithm on the eigenvectors of $A_{\mathrm{sum}}$. Its consistency in community detection is studied in \citep{paul2020spectral} under MLSBM.
  \item SoS-Debias: this algorithm is introduced in \citep{lei2023bias} and it detects communities by running the K-means algorithm on the eigenvectors of $S_{\mathrm{sum}}$ under MLSBM.
  \item MASE: the method called \underline{m}ultiple \underline{a}djacency \underline{s}pectral \underline{e}mbedding \citep{arroyo2021inference} detects communities by applying K-means algorithm to the left singular vectors of the matrix obtained from the adjacency spectral embedding of each adjacency matrix.
  \item KMAM: aggregate spectral \underline{k}ernel on \underline{m}odule \underline{a}llegiance \underline{m}atrix \citep{paul2020spectral}. This method obtains the community assignments by applying the K-means algorithm to $\sum_{l\in[L]}\hat{U}_{l}\hat{U}'_{l}$ where $\hat{U}_{l}$ is the leading $K$ eigenvectors of $A_{l}$.
  \item CAMSBM: \underline{c}ovariate-\underline{a}ssisted \underline{m}ulti-layer \underline{s}tochastic \underline{b}lock \underline{m}odel \citep{xu2023covariate}. CAMSBM is a tensor-based community detection approach in multi-layer networks.
  \item In addition to the aforementioned methods originally designed for community detection in multi-layer networks, we also consider \underline{a}daptively \underline{w}eighted \underline{p}rocrustes (AWP) \citep{nie2018multiview} originally developed for multi-view clustering for comparative analysis.
\end{itemize}

To measure the performance of community detection for different methods, we consider four metrics: Clustering error $\hat{f}$ computed by Equation (\ref{ErrorRate}), Hamming error \citep{SCORE}, normalized mutual
information (NMI) \citep{strehl2002cluster,danon2005comparing,bagrow2008evaluating}, and adjusted rand index (ARI) \citep{hubert1985comparing,vinh2009information}. Hamming error is defined as $\mathrm{min}_{\pi}\sum_{i\in[n]}\mathbbm{1}(\ell(i)\neq\pi(\hat{\ell}(i)))$, where $\mathbbm{1}$ is the indicator function and the minimum is taken over all label permutations $\pi$. Hamming error ranges in $[0,1]$, and it is the smaller the better. NMI ranges in [0,1], ARI ranges in [-1,1], and both measures are the larger the better. To measure the accuracy of different methods in estimating $K$ through maximizing $Q_{MNavrg}$, we use the Accuracy rate defined in \citep{qing2024finding}. This metric ranges in $[0,1]$, with a higher value indicating higher accuracy in determining $K$.

For all experiments below, unless specified, we consider $K=3$ communities and generate $Z$ by letting each node belong to one of the $K$ communities with equal probability. We let $\theta(i)=\sqrt{\rho}\cdot\mathrm{rand}(1)$ for all $i\in[n]$, where $\rho\in(0,1]$ is the sparsity parameter which is set independently for each experiment, and $\mathrm{rand}(1)$ is a random value generated from a Uniform distribution on $[0,1]$. For all $l\in[L]$, we generate the $l$-th block connectivity matrix $B_{l}$ in the following way: set $\tilde{B}_{l}(k,j)=\mathrm{rand}(1)$ for $k\in[K], j\in[K]$ and let $B_{l}=\frac{\tilde{B}+\tilde{B}'}{2}$ so that $B_{l}$ is a symmetric matrix with elements ranging in $[0,1]$. Specifically, for the purpose of estimating the number of communities, we set the diagonal elements of $B_{l}$ to be 1 for all $l \in [L]$, thereby generating assortative multi-layer networks. $n, L$, and $\rho$ are set independently for each experiment. For each parameter setting, we report every metric of each method averaged over 100 independent trials. In this paper, we carry out all experimental studies with MATLAB R2021b on a standard personal computer, specifically the Thinkpad X1 Carbon Gen 8.
\begin{figure*}
\centering
\subfigure[]{\includegraphics[width=0.33\textwidth]{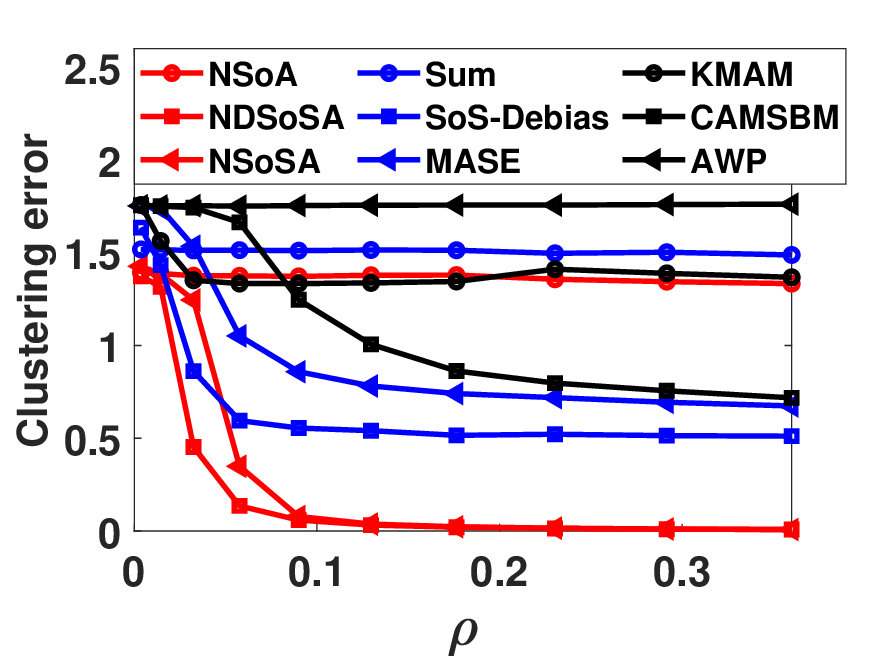}}
\subfigure[]{\includegraphics[width=0.33\textwidth]{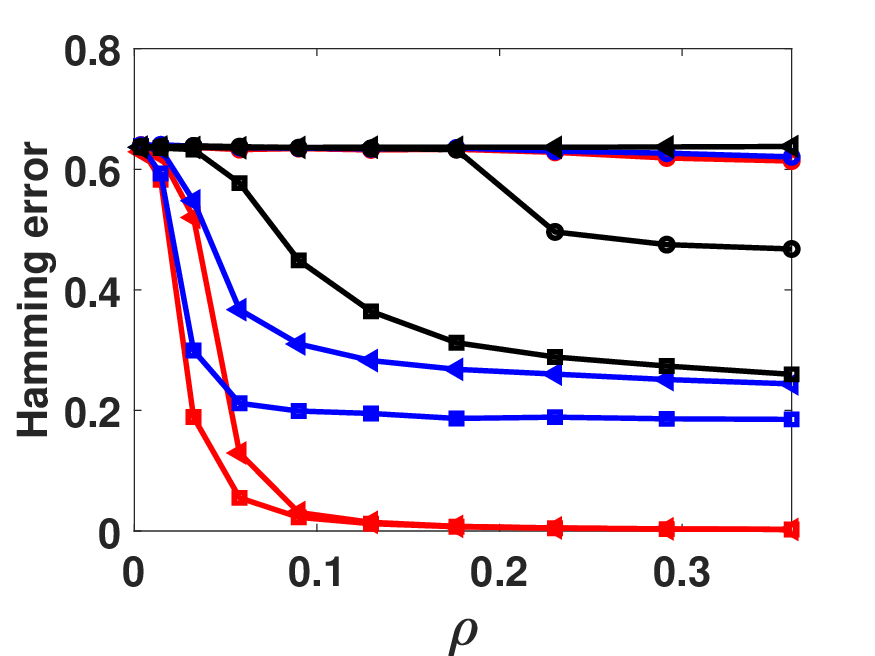}}
\subfigure[]{\includegraphics[width=0.33\textwidth]{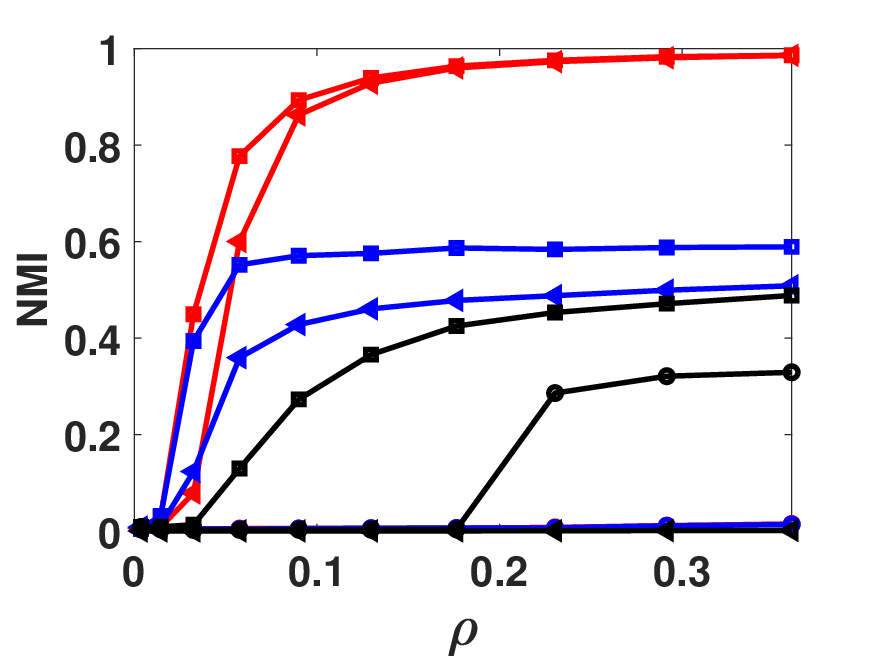}}
\subfigure[]{\includegraphics[width=0.33\textwidth]{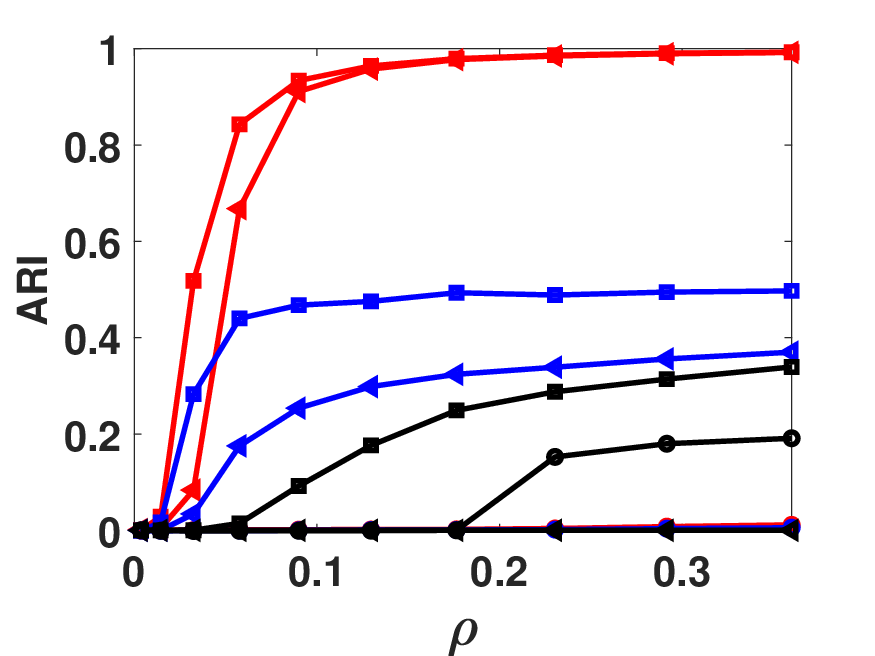}}
\subfigure[]{\includegraphics[width=0.33\textwidth]{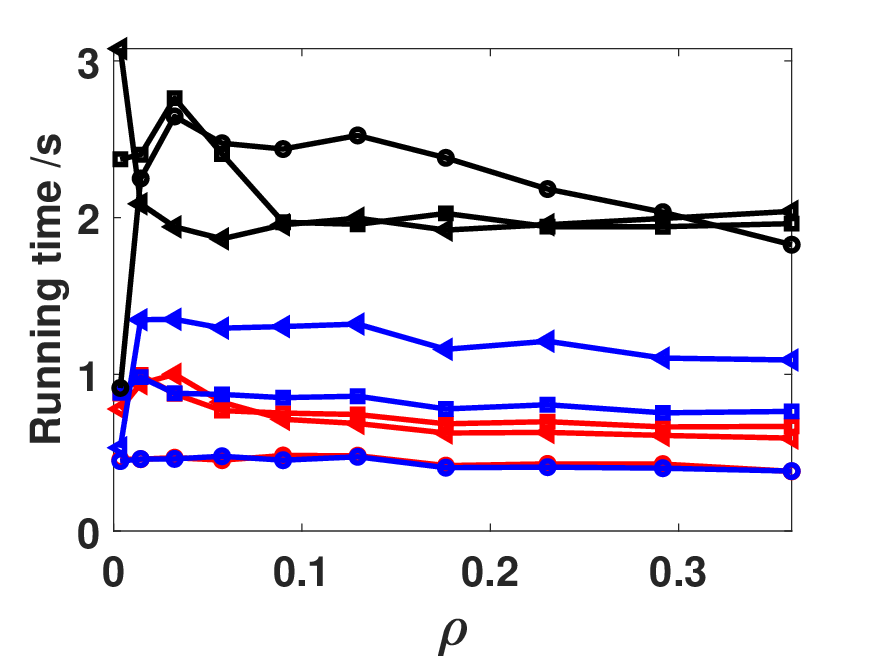}}
\subfigure[]{\includegraphics[width=0.33\textwidth]{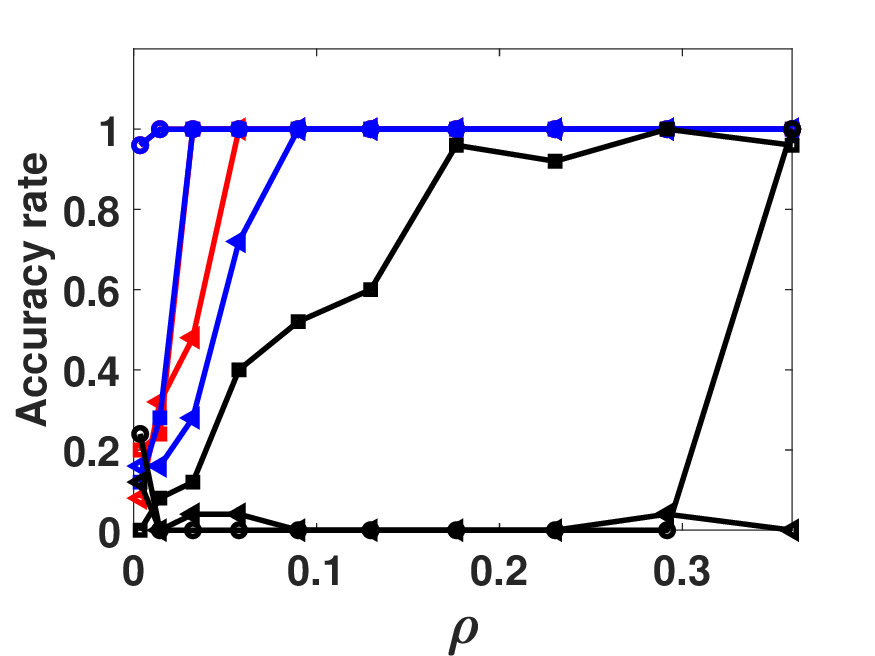}}
\caption{Numerical results of Experiment 1.}\label{Ex1}
\end{figure*}

\begin{figure*}
\centering
\subfigure[$\rho=0.0036$]{\includegraphics[width=0.245\textwidth]{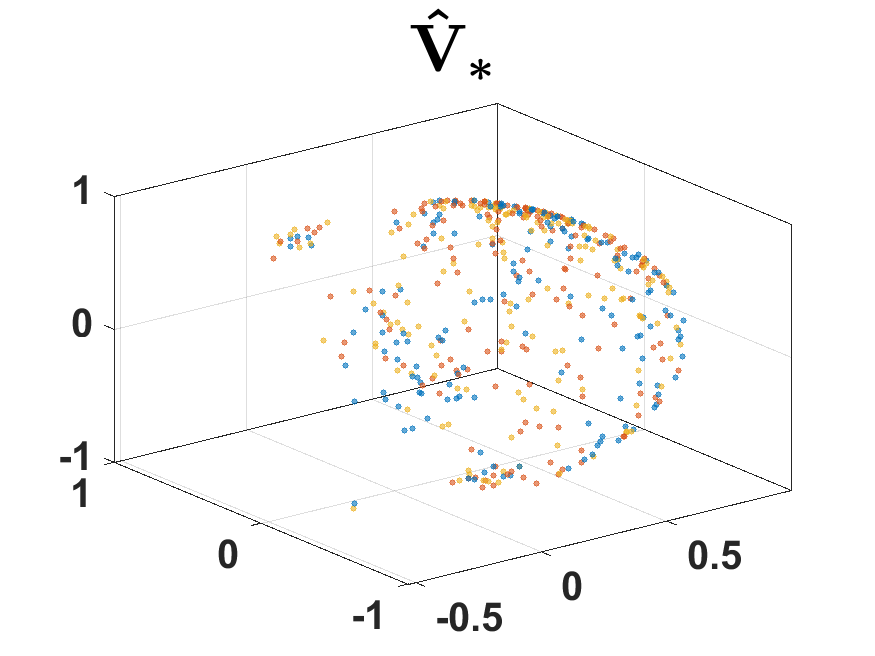}}
\subfigure[$\rho=0.0576$]{\includegraphics[width=0.245\textwidth]{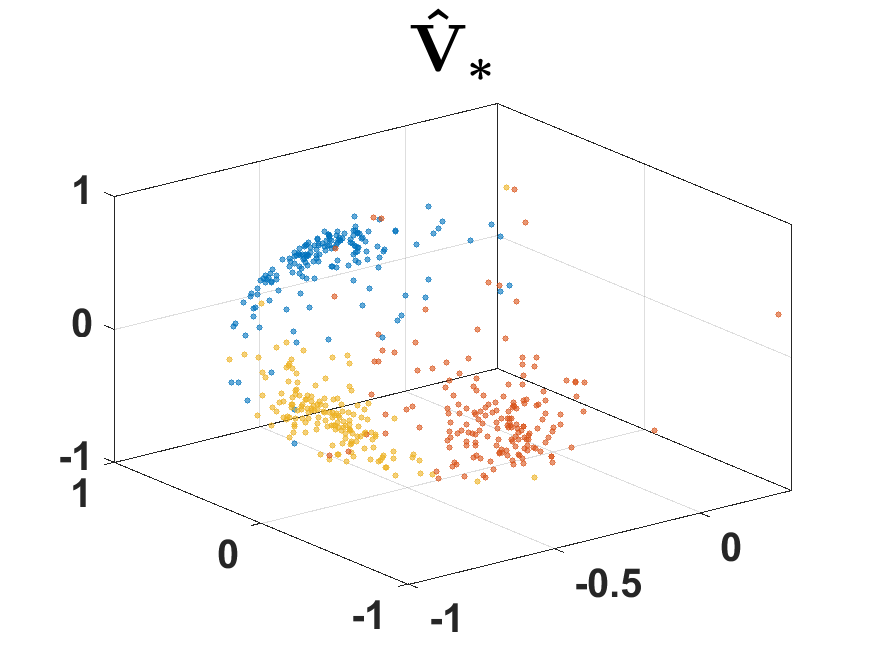}}
\subfigure[$\rho=0.1764$]{\includegraphics[width=0.245\textwidth]{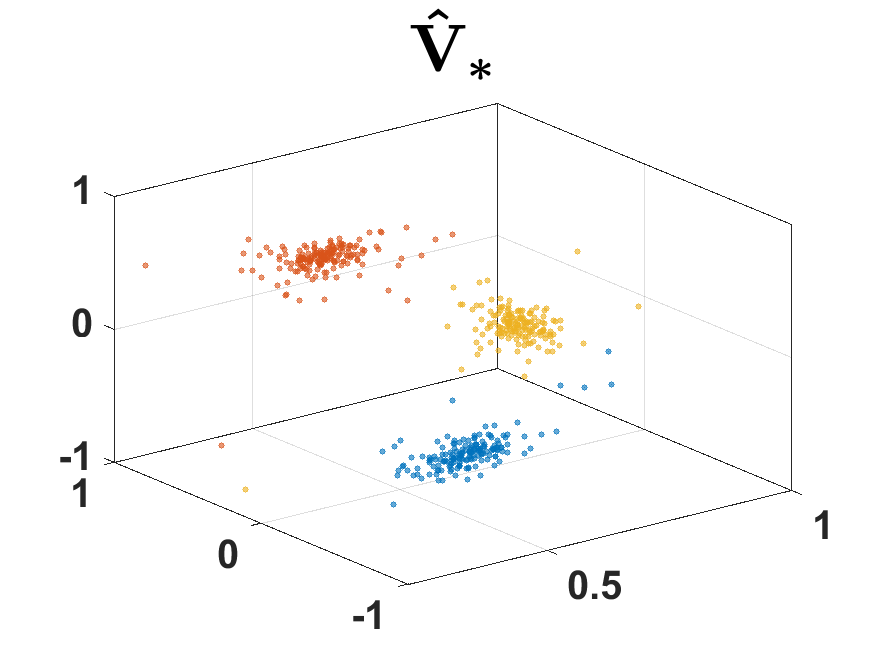}}
\subfigure[$\rho=0.36$]{\includegraphics[width=0.245\textwidth]{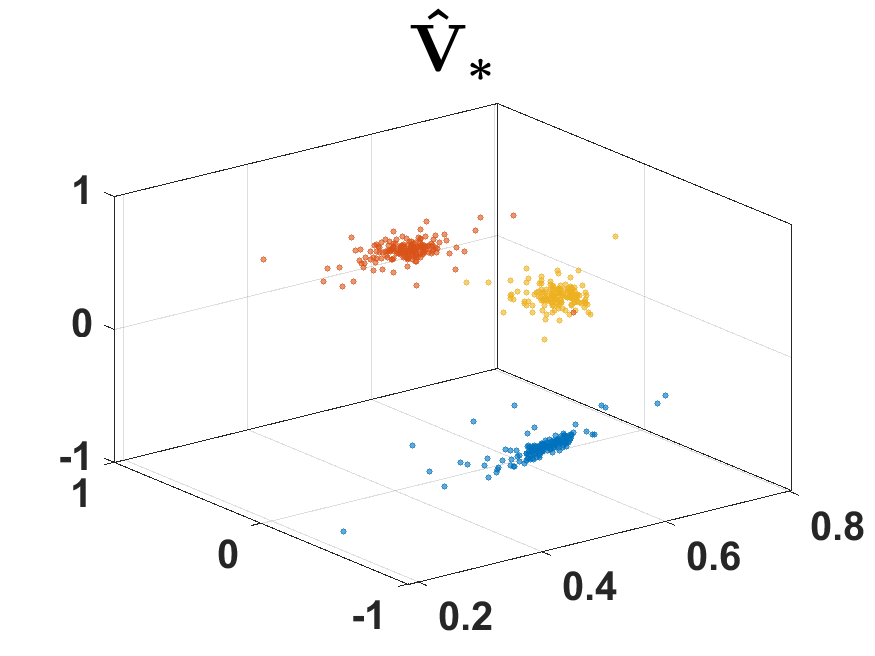}}
\caption{Panels (a)-(d) illustrate points, where each point corresponds to a row of $\hat{V}_{*}$ at varying levels of the sparsity parameter $\rho$. Across all panels, colors denote different communities.}
\label{Vhatstar} 
\end{figure*}

\emph{Experiment 1: Varying the Sparsity Parameter $\rho$}. In this experiment, we assess the performance of all methods by gradually increasing the sparsity parameter $\rho$ for a fixed $n=500$ and $L=100$. We vary $\rho$ across the range $\{0.06^{2}, 0.012^{2}, \ldots, 0.6^{2}\}$ to observe its impact. The findings are presented in Fig.~\ref{Ex1}. For the task of detecting communities, the results reveal that as the multi-layer network becomes denser, all methods except AMP exhibit improved performance. Notably, our NDSoSA algorithm surpasses its competitors in all scenarios. In particular, NDSoSA outperforms NSoA, and NSoSA, in turn, demonstrates better performance than NSoA. These observations align with our theoretical analysis presented in Sections \ref{CompareDA} and \ref{MainNSoSA}. In terms of running time, KMAM, CAMSBM, and AWP run slower than the other six methods. For the estimation of $K$, with the exception of KMAM, CAMSBM, and AWP, the remaining six methods exhibit satisfactory performance and demonstrate higher accuracy in identifying $K$ as the sparsity parameter $\rho$ increases. Fig.~\ref{Vhatstar} provides a visual illustration of how the geometric structure of $\hat{V}_{*}$ changes with varying the sparsity parameter $\rho$ in this experiment. As $\rho$ increases from $0.06^{2}$ to $0.6^{2}$, a clear trend emerges: clusters representing nodes within the same community become increasingly distinct from clusters of nodes in other communities. This observation underscores the impact of network sparsity on the performance of the NDSoSA algorithm for community detection. Specifically, as the network becomes denser with higher $\rho$ values, the separation between clusters improves, making it easier for NDSoSA to distinguish nodes belonging to different communities. This visualization aligns well with the quantitative results presented in Fig.~\ref{Ex1} and supports the theoretical understanding that denser networks with more connectivity information facilitate more accurate community detection. Fig.~\ref{Nrho} visualizes the network of one layer at varying levels of the sparsity parameter $\rho$. At lower values of $\rho$ (e.g., $\rho=0.0036$), the network is sparse, with only a few edges connecting the nodes. As $\rho$ increases (e.g., $\rho=0.36$), the network becomes denser, with more edges appearing between the nodes. Thus, Fig.~\ref{Nrho} clearly shows the impact of the sparsity parameter $\rho$ on the network structure, illustrating the transition from a sparse to a dense network as $\rho$ increases.
\begin{figure*}
\centering
\subfigure[$\rho=0.0036$]{\includegraphics[width=0.245\textwidth]{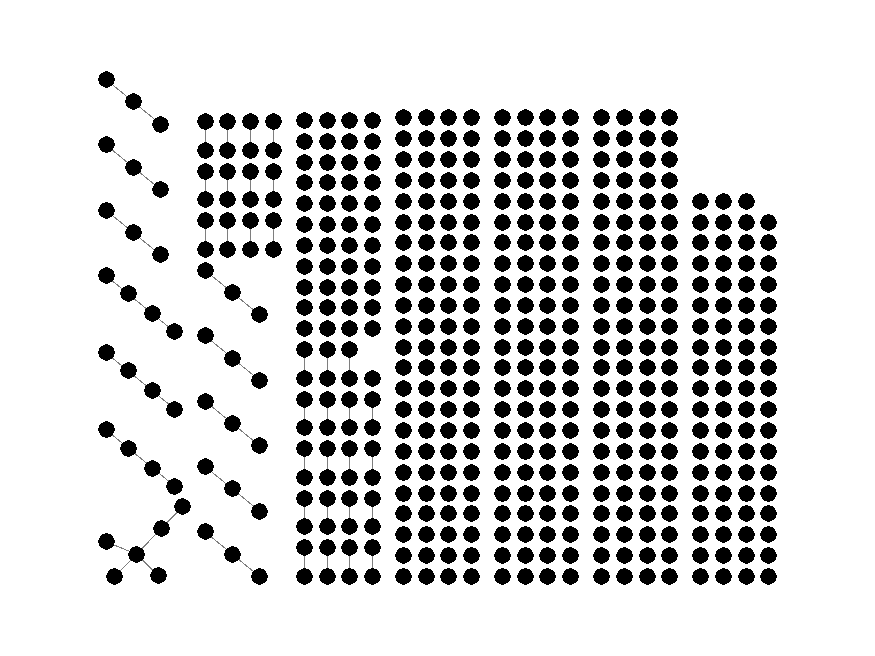}}
\subfigure[$\rho=0.0576$]{\includegraphics[width=0.245\textwidth]{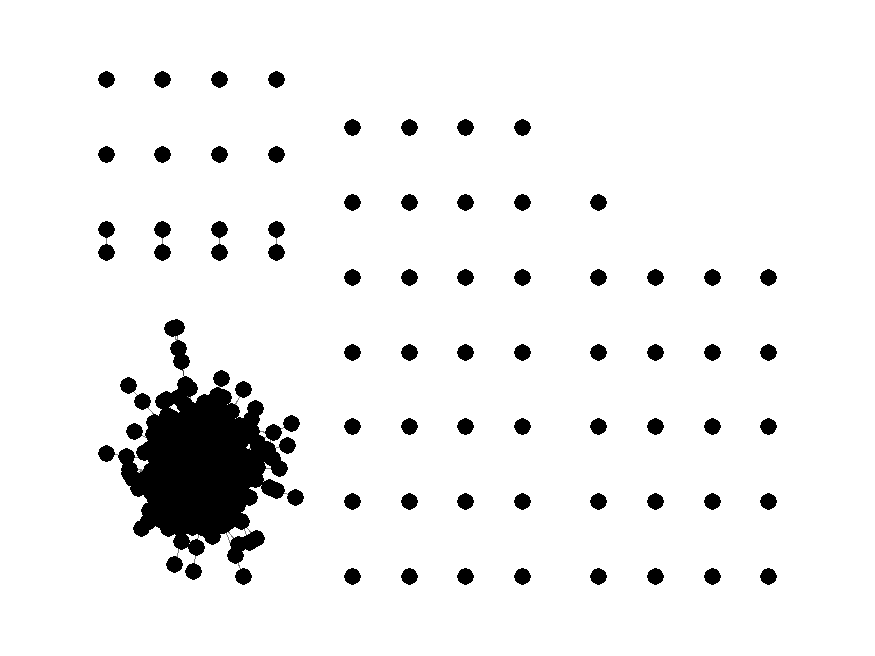}}
\subfigure[$\rho=0.1764$]{\includegraphics[width=0.245\textwidth]{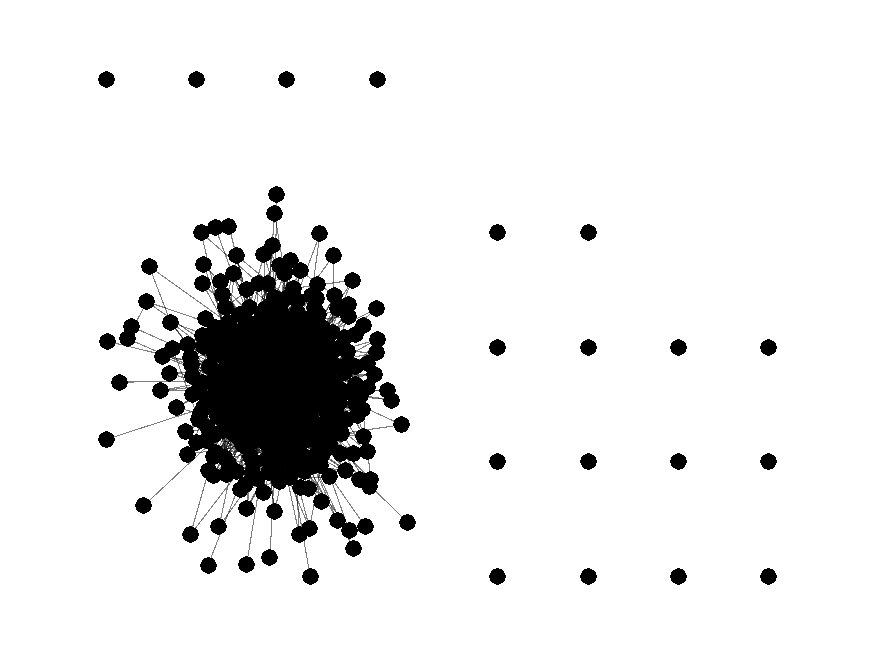}}
\subfigure[$\rho=0.36$]{\includegraphics[width=0.245\textwidth]{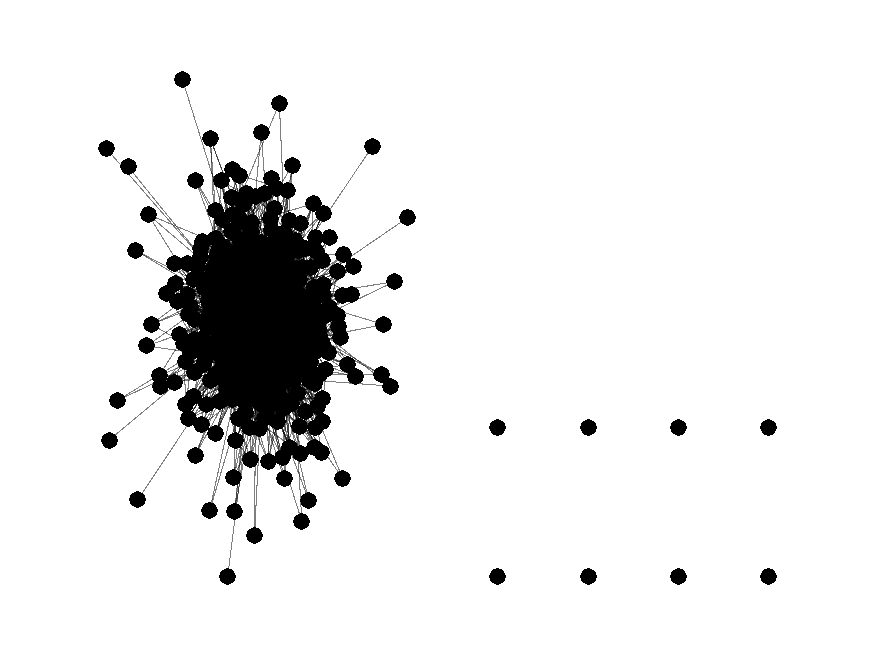}}
\caption{Panels (a)-(d) display the network of one layer at different levels of the sparsity parameter $\rho$.}
\label{Nrho} 
\end{figure*}

\emph{Experiment 2: Varying the Number of Nodes $n$}. In this experiment, we explore the impact of varying the number of nodes $n$ on the performance of all methods. We fix the number of layers at $L=20$ and the sparsity parameter at $\rho=0.16$. We then vary $n$ across the range $\{100, 200, \ldots, 1000\}$ to observe its effect. The results are presented in Fig.~\ref{Ex2}. For the detection of communities, our NDSoSA algorithm emerges as the top performer among all methods, slightly surpassing NSoSA. Additionally, all methods except for AWP exhibit improved performance as the number of nodes $n$ increases. Meanwhile, regarding running time, all methods require more time for community detection as the number of nodes increases. Specifically, the six methods, NSoA, NDSoSA, NSoSA, Sum, SoS-Debias, and MASE, exhibit faster performance compared to KMAM, CAMSBM, and AWP. The subpar performance of NSoA in this experiment can be attributed to the relatively small number of layers, $L=20$. In scenarios where $L$ is set to be a higher value, such as 100 or more, NSoA performs satisfactorily. However, for larger values of $L$, NDSoSA achieves error rates close to zero. Therefore, to effectively investigate the trend of NDSoSA's superior performance as $n$ increases, we chose a median value for $L$ in this experiment. For the determination of $K$, our methods NSoA, NDSoSA, and NSoSA provide highly accurate estimates, and their accuracy improves as $n$ increases.

\emph{Experiment 3: Varying the Number of Layers $L$}. This experiment focuses on assessing the impact of varying the number of layers $L$ on the performance of all methods. We fix the number of nodes at $n=500$ and the sparsity parameter at $\rho=0.16$. Subsequently, we vary $L$ within the range $\{2, 4, \ldots, 40\}$ to observe its influence. The results are graphically represented in Fig.~\ref{Ex3}. Notably, for the task of detecting communities, NDSoSA significantly outperforms all other methods across the entire parameter range. Furthermore, as the number of layers $L$ increases, the error rate of NDSoSA decreases, indicating its superior performance in handling complex multi-layer networks. Regarding running time, KMAM, CAMSBM, and AWP run slower than their competitors. For the task of identifying $K$, all methods except KMAM and AWP almost correctly determine $K$ when $L$ is no smaller than 6. Fig.~\ref{VhatstarL} illustrates the geometric structure of $\hat{V}^*$ at different $L$ in this experiment. We observe that initially, at $L=10$, there is some overlap between clusters. When $L$ increases to 20, 30, and 40, the clusters become more distinct and well-separated. This observation suggests that more layers provide more connections and interactions between nodes, enhancing cluster separability and refining the geometric structure of $\hat{V}_{*}$, thus yielding more accurate community detection.
\begin{figure*}
\centering
\subfigure[]{\includegraphics[width=0.33\textwidth]{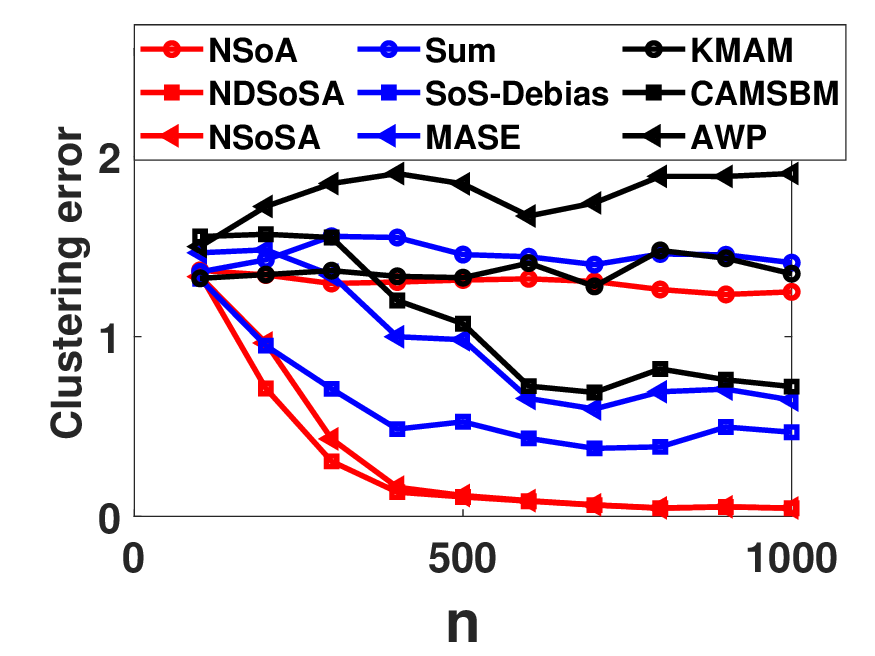}}
\subfigure[]{\includegraphics[width=0.33\textwidth]{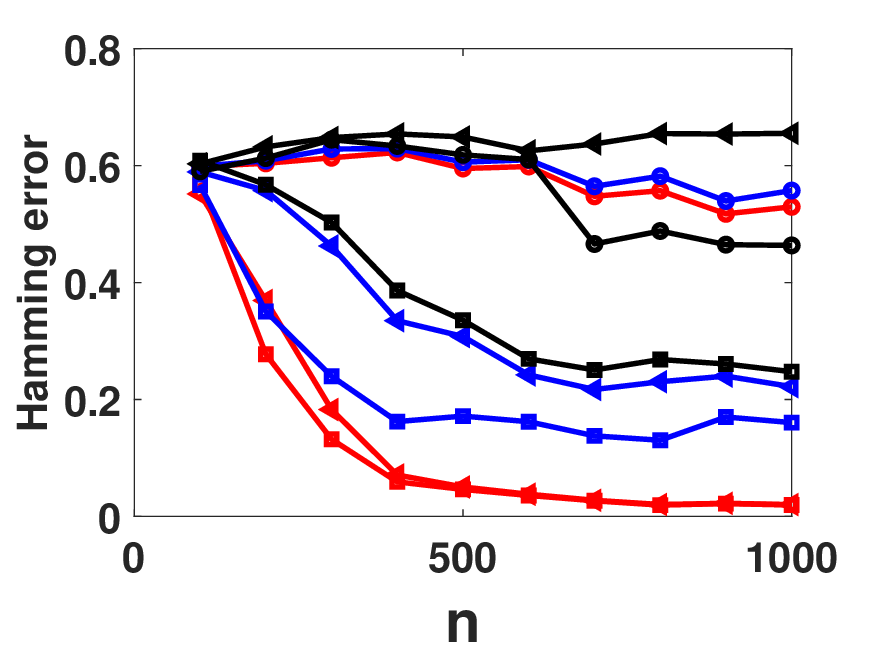}}
\subfigure[]{\includegraphics[width=0.33\textwidth]{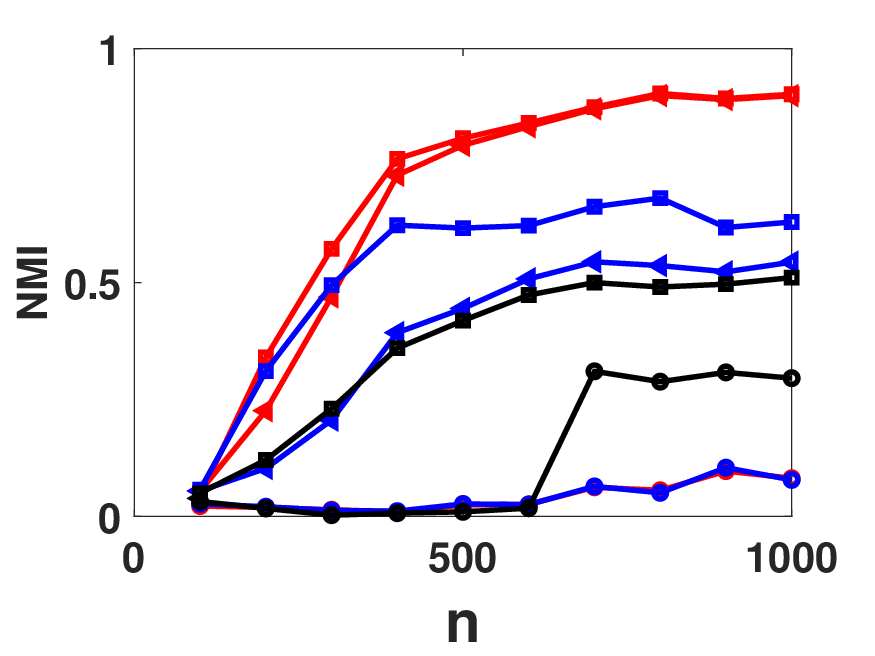}}
\subfigure[]{\includegraphics[width=0.33\textwidth]{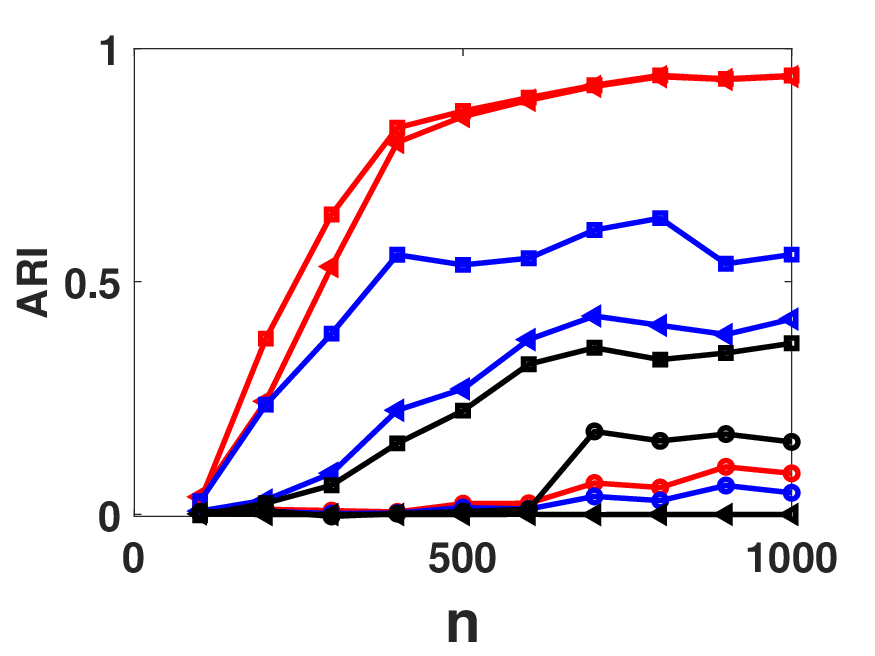}}
\subfigure[]{\includegraphics[width=0.33\textwidth]{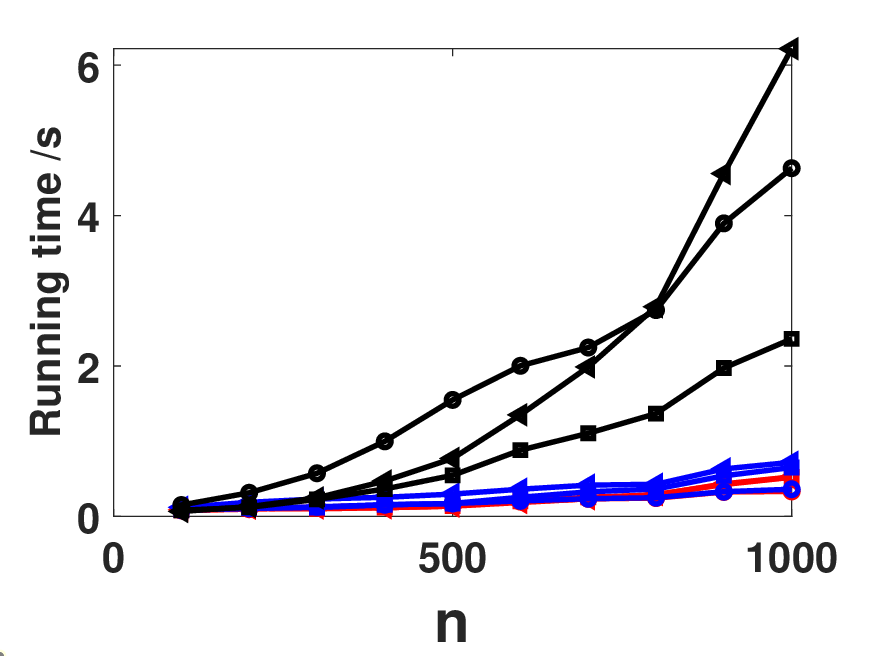}}
\subfigure[]{\includegraphics[width=0.33\textwidth]{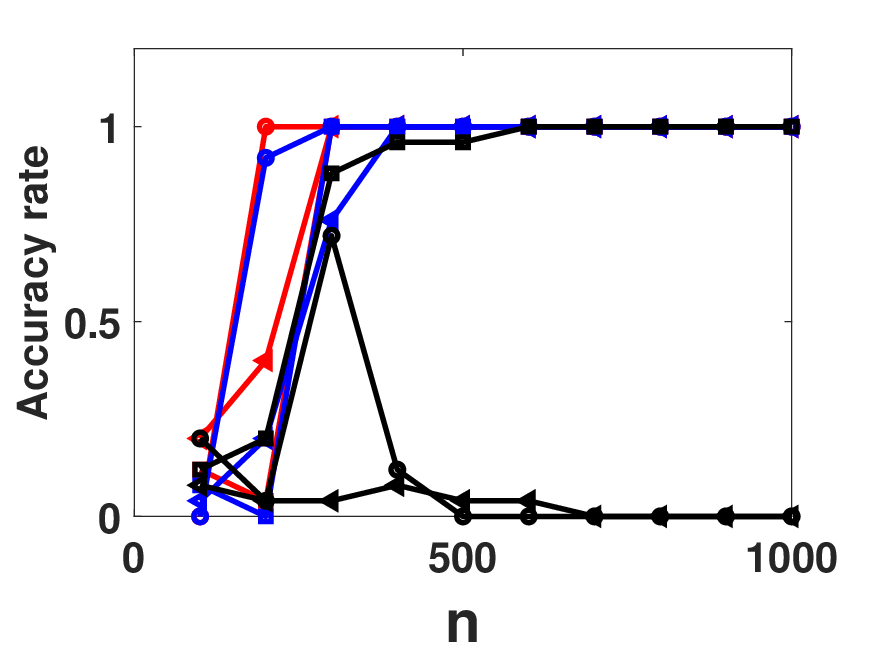}}
\caption{Numerical results of Experiment 2.}
\label{Ex2} 
\end{figure*}
\begin{figure*}
\centering
\subfigure[]{\includegraphics[width=0.33\textwidth]{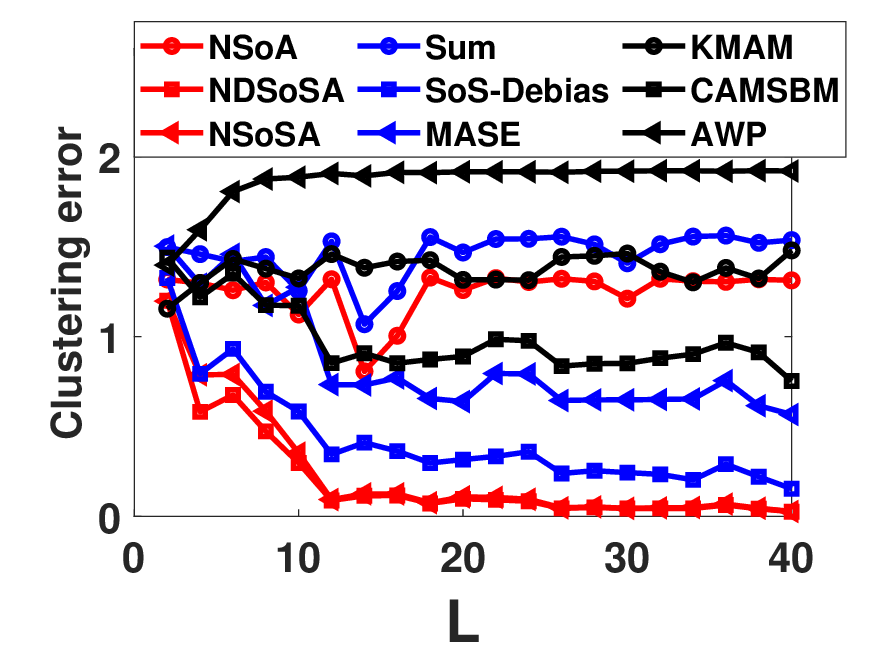}}
\subfigure[]{\includegraphics[width=0.33\textwidth]{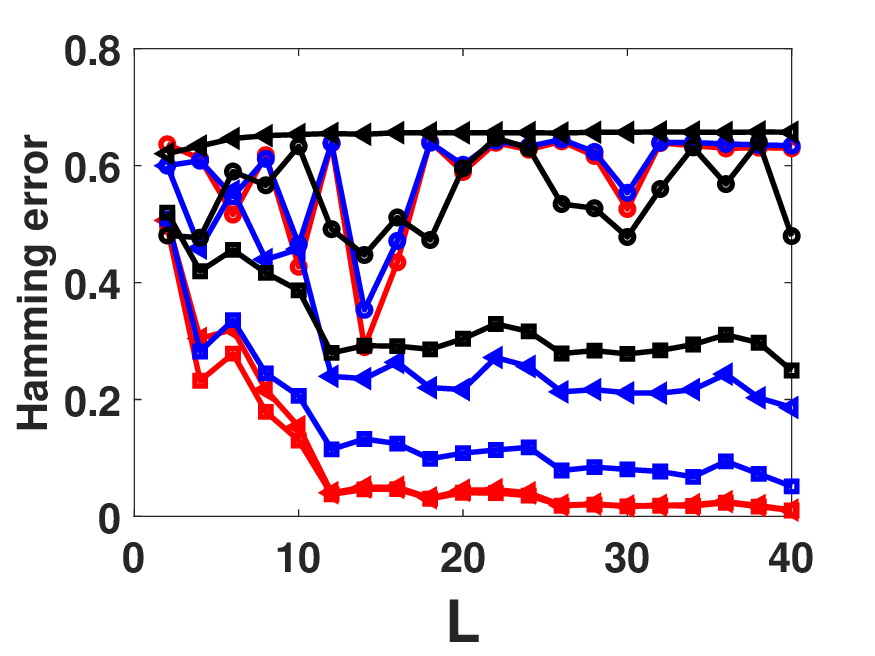}}
\subfigure[]{\includegraphics[width=0.33\textwidth]{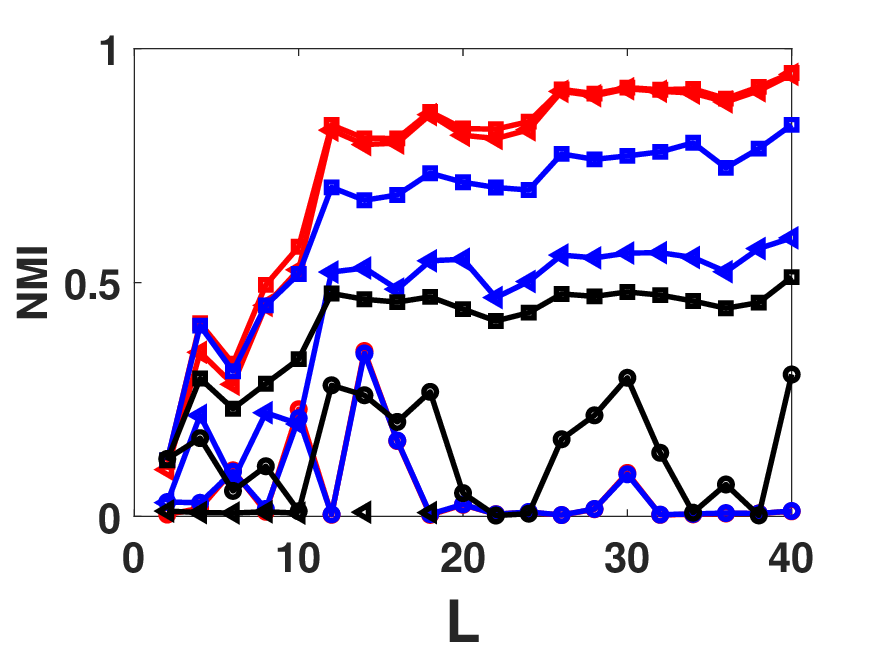}}
\subfigure[]{\includegraphics[width=0.33\textwidth]{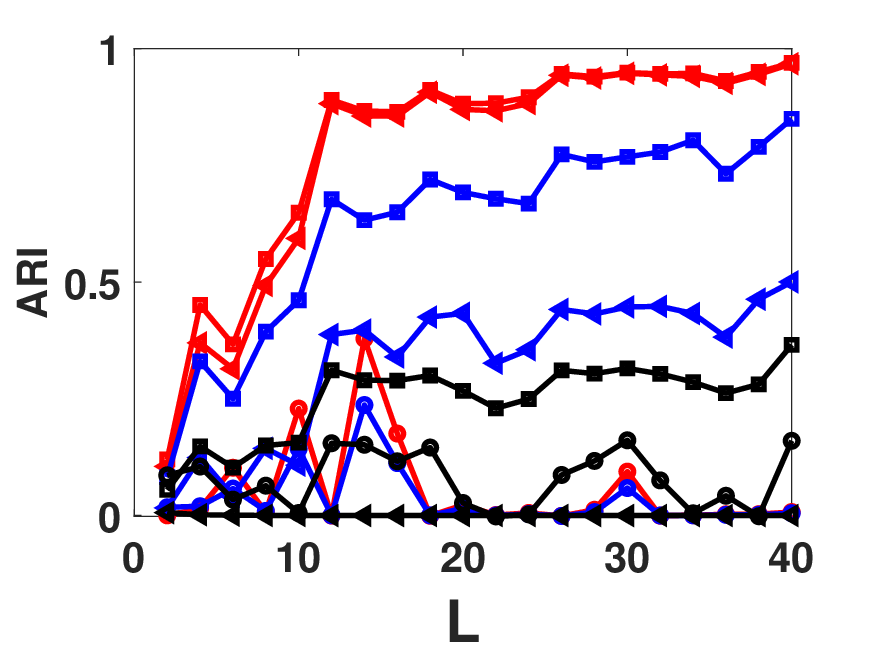}}
\subfigure[]{\includegraphics[width=0.33\textwidth]{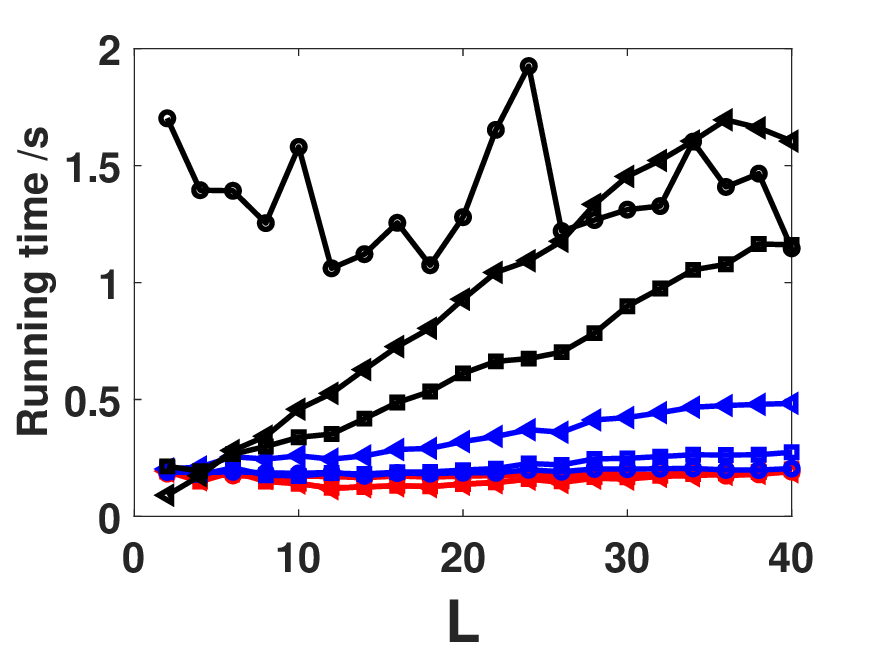}}
\subfigure[]{\includegraphics[width=0.33\textwidth]{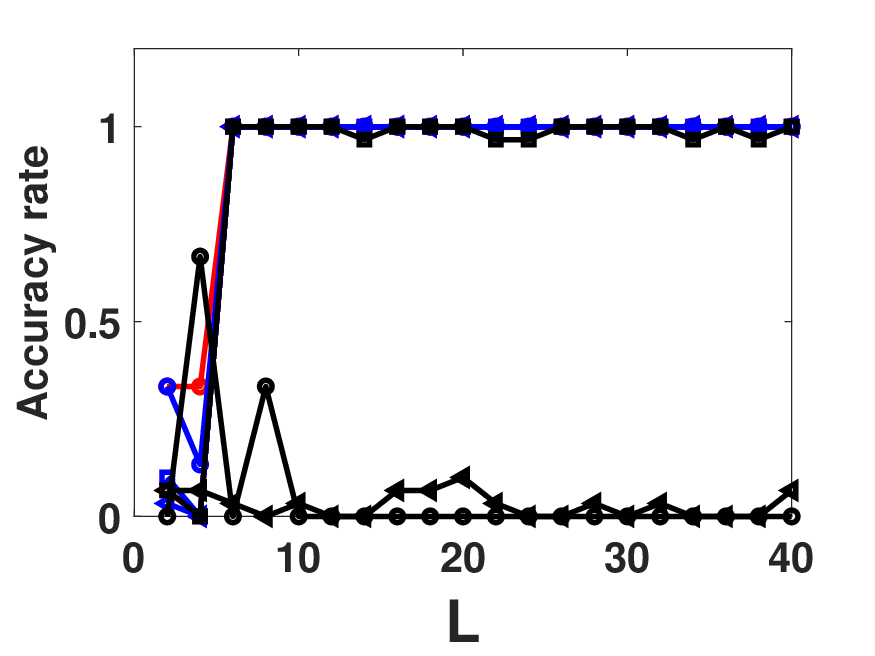}}
\caption{Numerical results of Experiment 3.}
\label{Ex3} 
\end{figure*}

\begin{figure*}
\centering
\subfigure[$L=10$]{\includegraphics[width=0.245\textwidth]{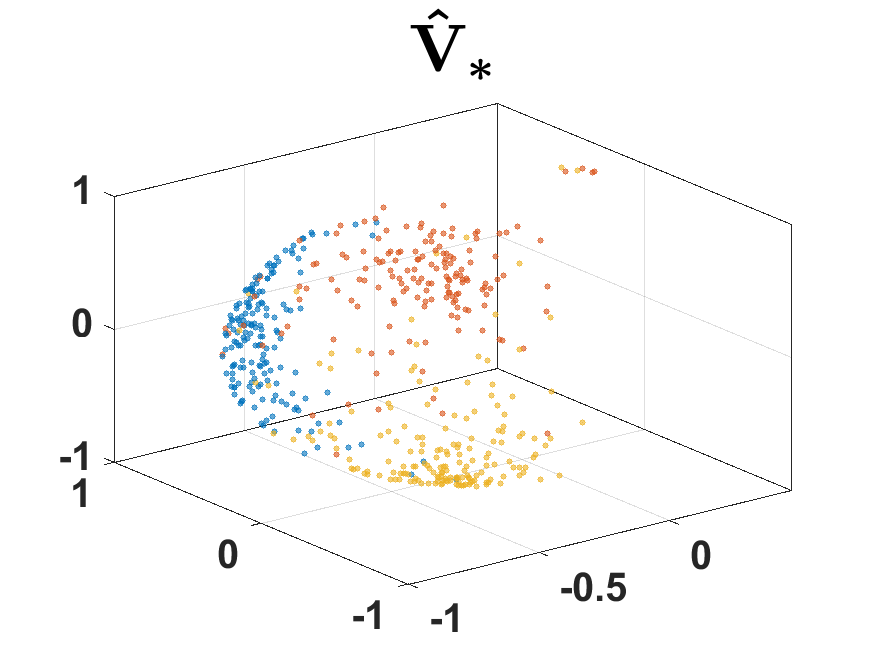}}
\subfigure[$L=20$]{\includegraphics[width=0.245\textwidth]{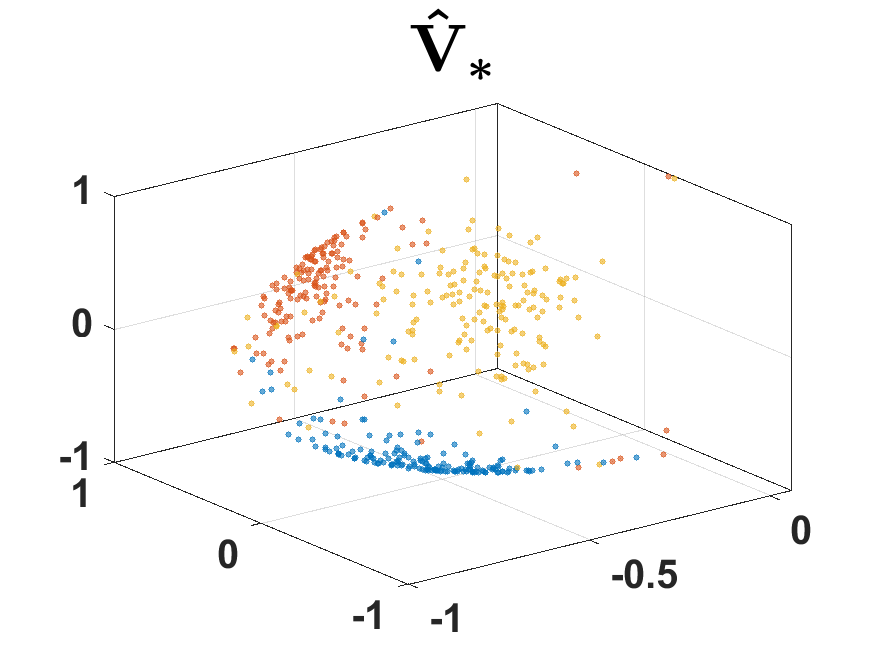}}
\subfigure[$L=30$]{\includegraphics[width=0.245\textwidth]{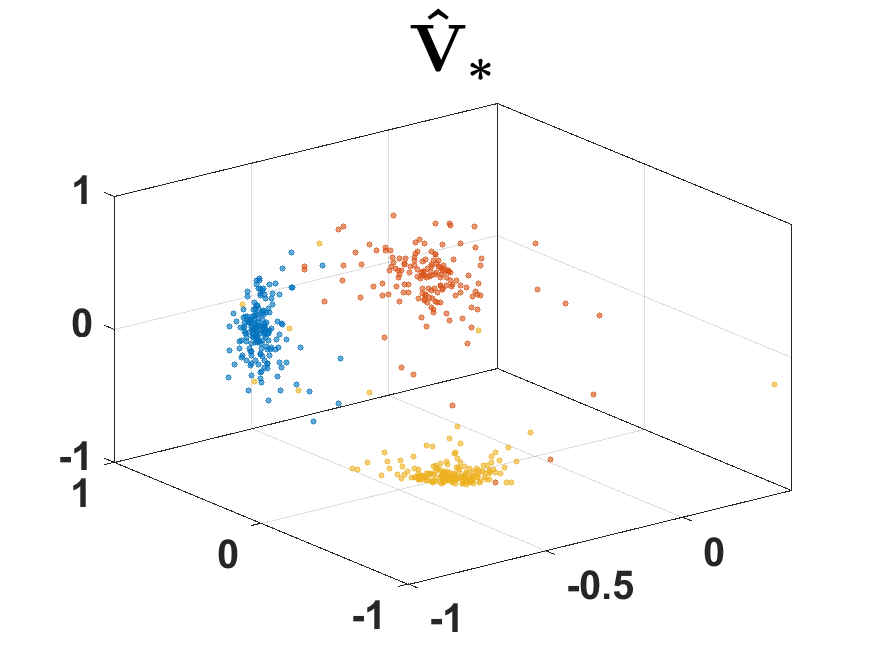}}
\subfigure[$L=40$]{\includegraphics[width=0.245\textwidth]{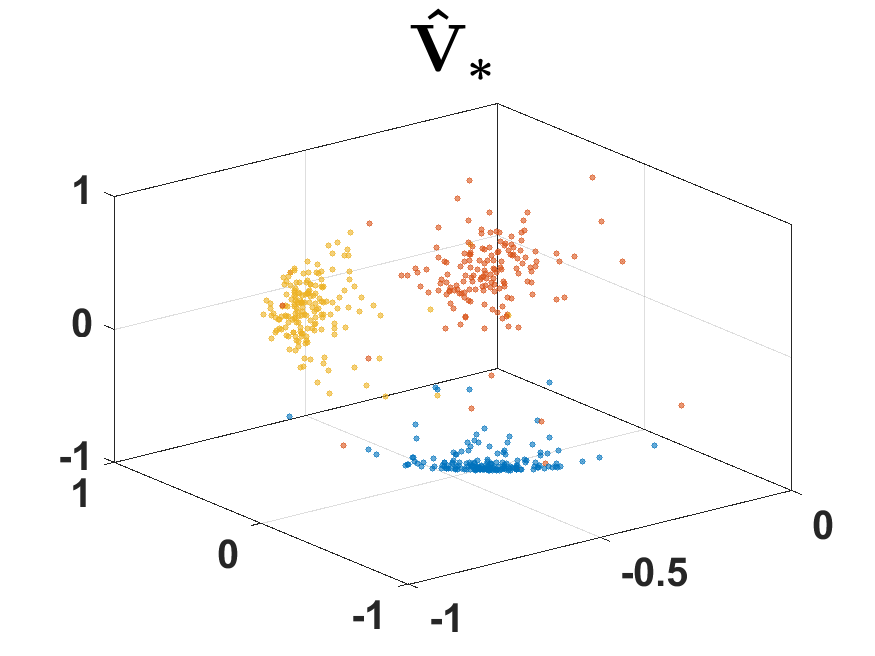}}
\caption{Panels (a)-(d) illustrate $\hat{V}_{*}$ at different number of layers $L$. In all panels, colors indicate communities.}
\label{VhatstarL} 
\end{figure*}

\emph{Experiment 4: Varying the Number of Communities $K$}. This experiment assesses the efficacy of all methods as the number of communities $K$ grows. We fix $L=20$ and $\rho=0.16$, while allowing $K$ to vary within the set $\{1,2,\ldots,6\}$. The value of $n$ is set to $200K$, ensuring that each community comprises roughly 200 nodes. The results of this experiment are shown in Fig.~\ref{Ex4}. Our methods, NSoA, NDSoSA, and NSoSA, exhibit satisfactory performance with low Clustering errors, low Hamming errors, high NMI, high ARI, minimal running time, and high Accuracy rates. In contrast, their competing methods in this experiment perform poorly. We also observe that our NSoA, NDSoSA, and NSoSA even accurately estimate $K$ when the true $K$ is 1.
\begin{figure*}
\centering
\subfigure[]{\includegraphics[width=0.33\textwidth]{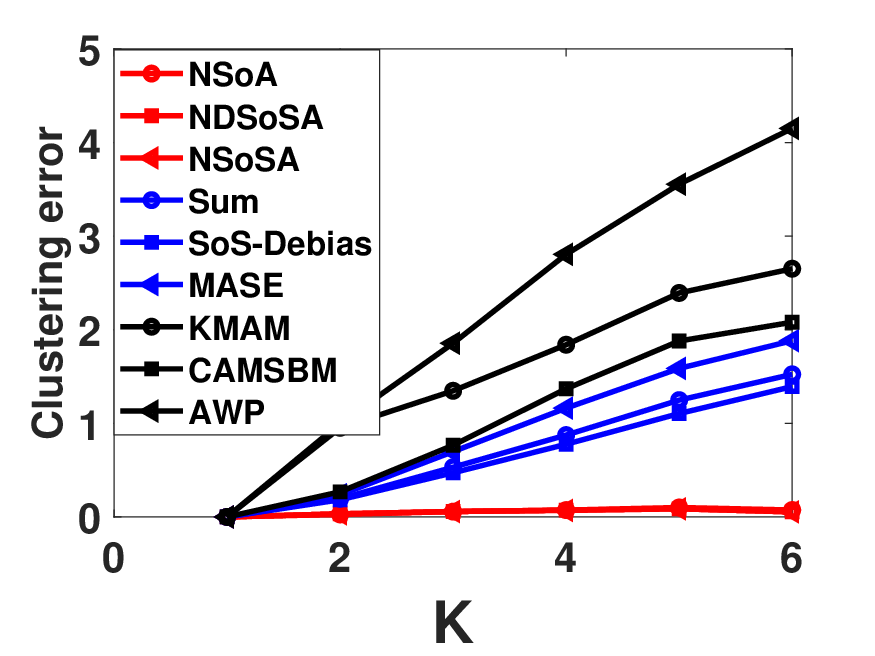}}
\subfigure[]{\includegraphics[width=0.33\textwidth]{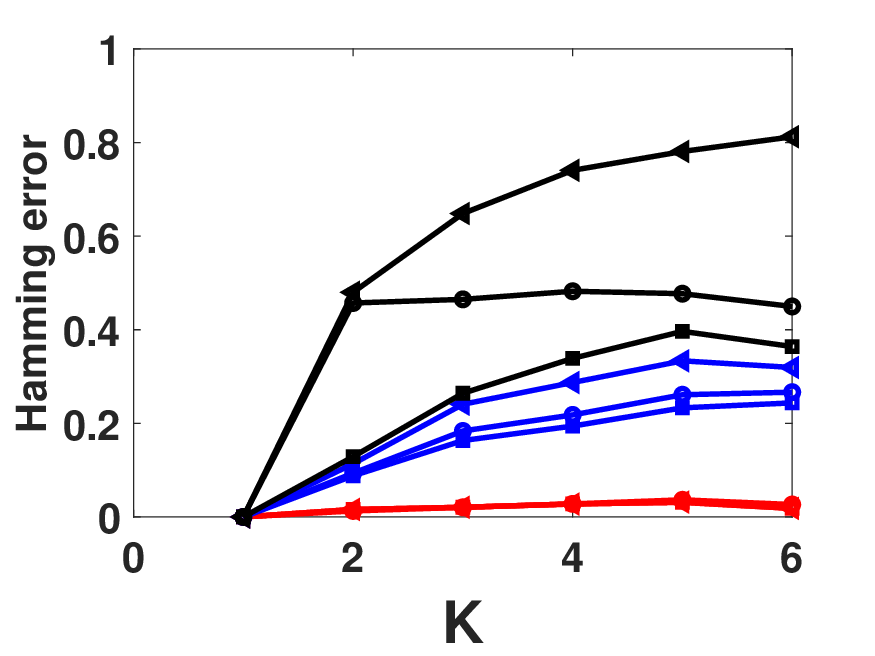}}
\subfigure[]{\includegraphics[width=0.33\textwidth]{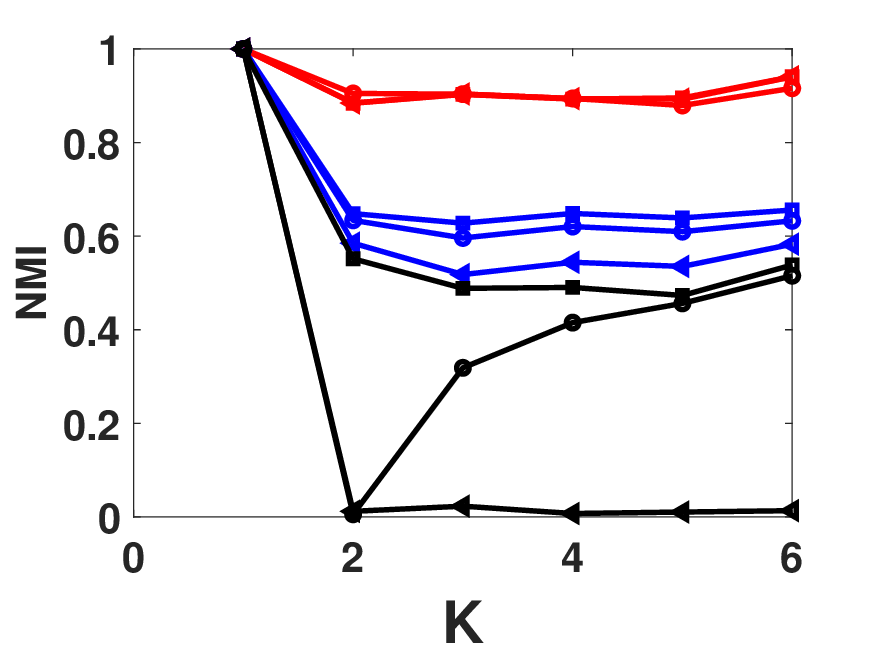}}
\subfigure[]{\includegraphics[width=0.33\textwidth]{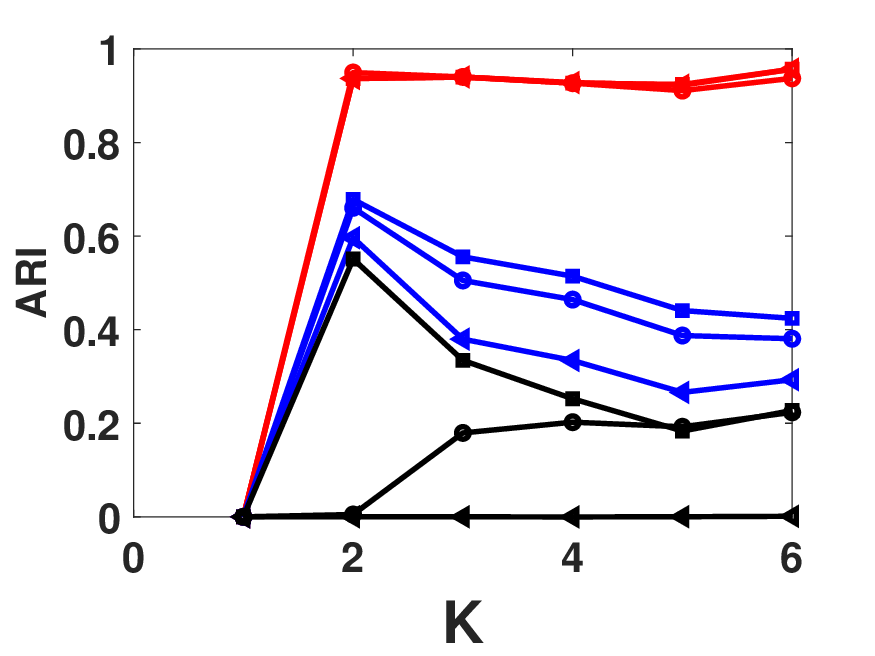}}
\subfigure[]{\includegraphics[width=0.33\textwidth]{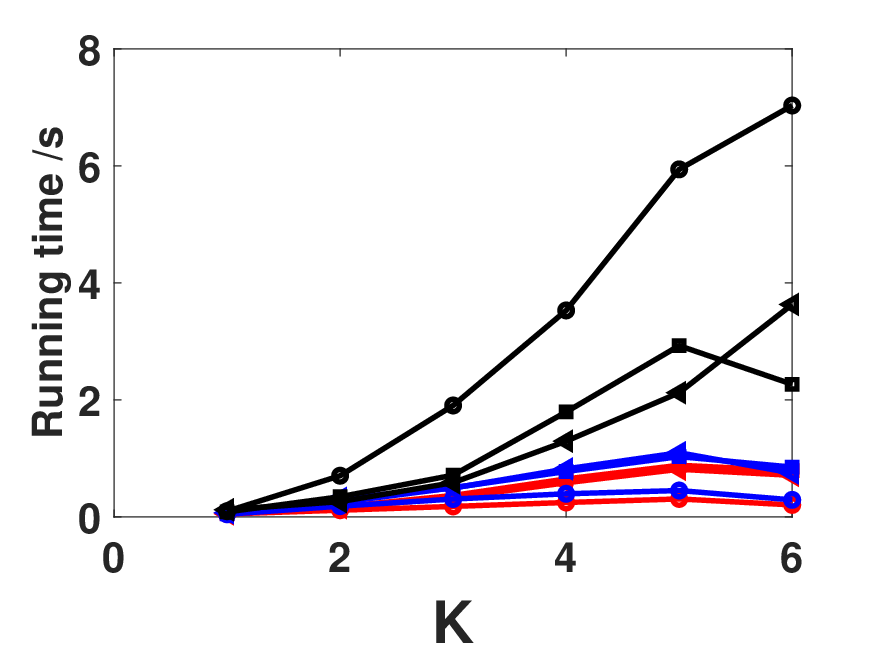}}
\subfigure[]{\includegraphics[width=0.33\textwidth]{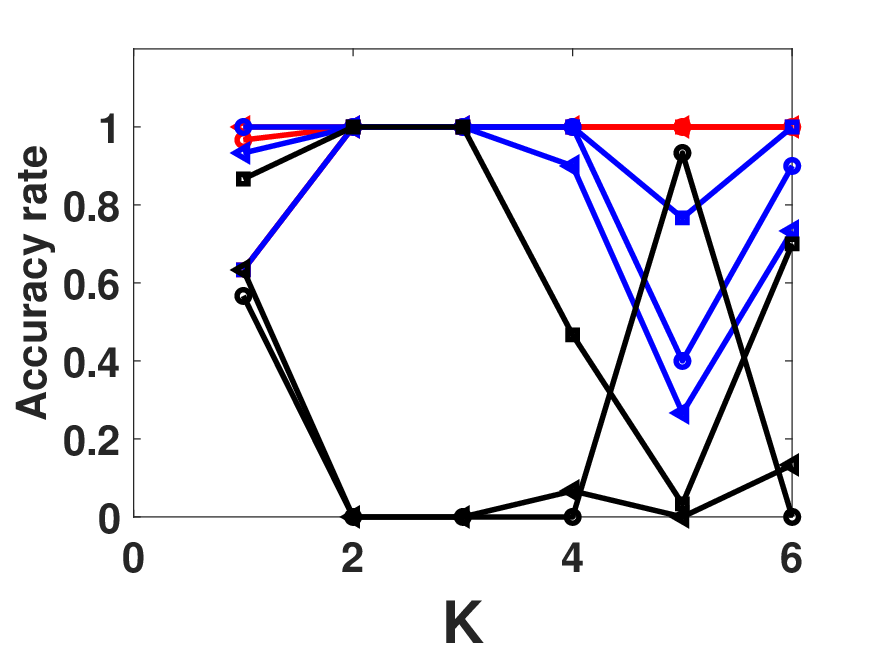}}
\caption{Numerical results of Experiment 4.}
\label{Ex4} 
\end{figure*}

We then conduct two experiments to assess the efficacy of the two accelerated algorithms, SNSoA and SNDSoSA, by comparing them with NSoA and NDSoSA in large-scale multi-layer networks. The sample sizes $n_{\mathrm{sample}}$ and $L_{\mathrm{sample}}$ remain the same as those in Section \ref{SecSubsampling}. We omit the inclusion of other methods studied in previous experiments for comparative analysis, as prior numerical results have shown that our NSoA and NDSoSA algorithms typically outperform their competitors in terms of speed.

\emph{Experiment 5: Effect of subsampling for large $n$.} In this experiment, we set $\rho=0.25, L=3$, and let $n$ range in $\{2000,4000,\ldots,24000\}$. The numerical results are presented in Fig.~\ref{Ex5}. Our observations are as follows: (a) As the number of nodes increases, the performance of all methods improves, with NDSoSA and SNDSoSA consistently showing the best results; (b) all methods provide precise estimates of $K$; (c) NSoA and NDSoSA typically outperform their accelerated versions SNSoA and SNDSoSA in the task of detecting communities, respectively; (d) SNDSoSA typically performs better than SNSoA in detecting communities, although it does so at the expense of slower execution time; (e) SNDSoSA runs significantly faster than NDSoSA, particularly for large values of $n$. Additionally, SNSoA demonstrates a running time comparable to NSoA for small $n$ and runs faster than NSoA for large $n$.

\emph{Experiment 6: Effect of subsampling for large $L$.} In this experiment, we set $\rho=0.25, n=500$, and let $L$ range in $\{20,40, \ldots,200\}$. The results presented in Fig.~\ref{Ex6} illustrate that: (a) increasing the number of layers $L$ enhances the precision of community detection for all methodologies; (b) NSoA and NDSoSA consistently surpass their accelerated counterparts in the task of community detection; (c) all methods accurately determine $K$; (d) SNDSoSA runs faster than NDSoSA, particularly for larger values of $L$. Additionally, SNSoA exhibits comparable running time to NSoA when $L\leq60$ and outperforms NSoA in speed when $L>60$.
\begin{figure*}
\centering
\subfigure[]{\includegraphics[width=0.33\textwidth]{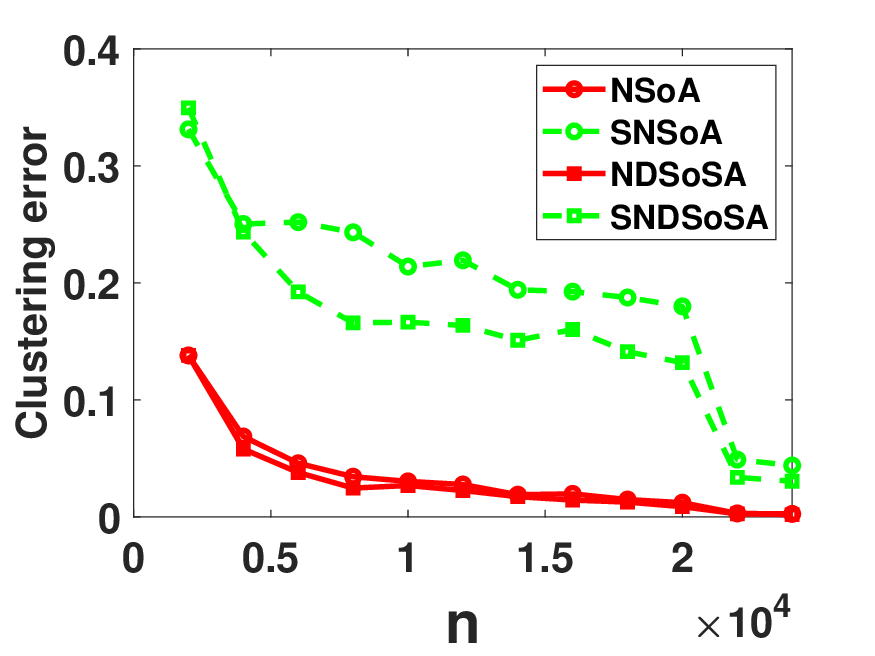}}
\subfigure[]{\includegraphics[width=0.33\textwidth]{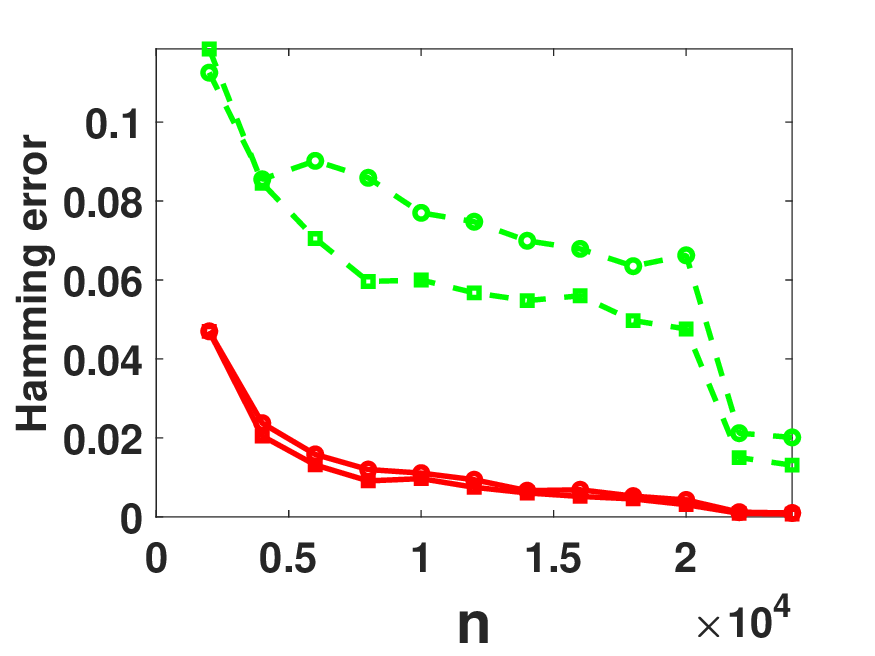}}
\subfigure[]{\includegraphics[width=0.33\textwidth]{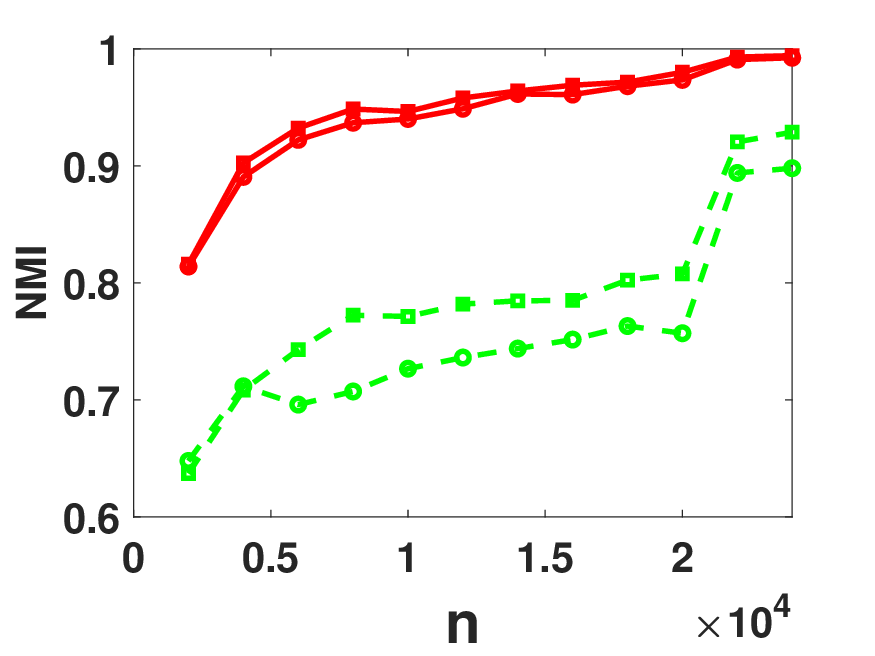}}
\subfigure[]{\includegraphics[width=0.33\textwidth]{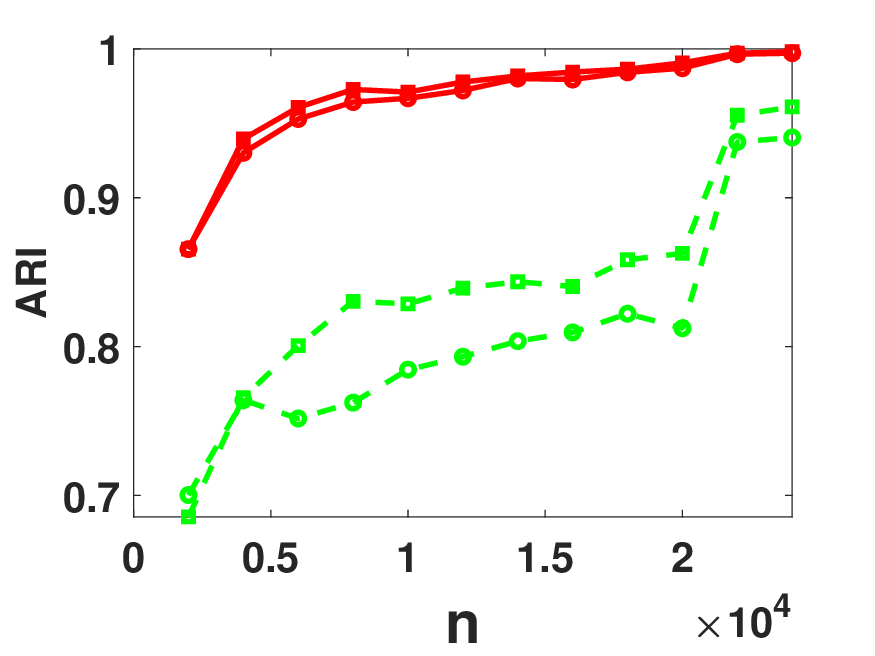}}
\subfigure[]{\includegraphics[width=0.33\textwidth]{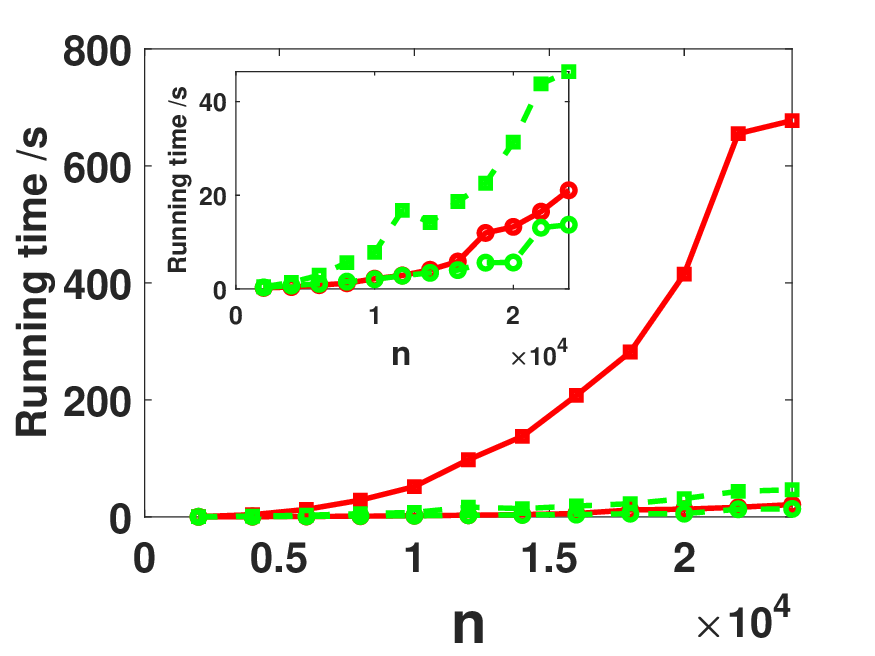}}
\subfigure[]{\includegraphics[width=0.33\textwidth]{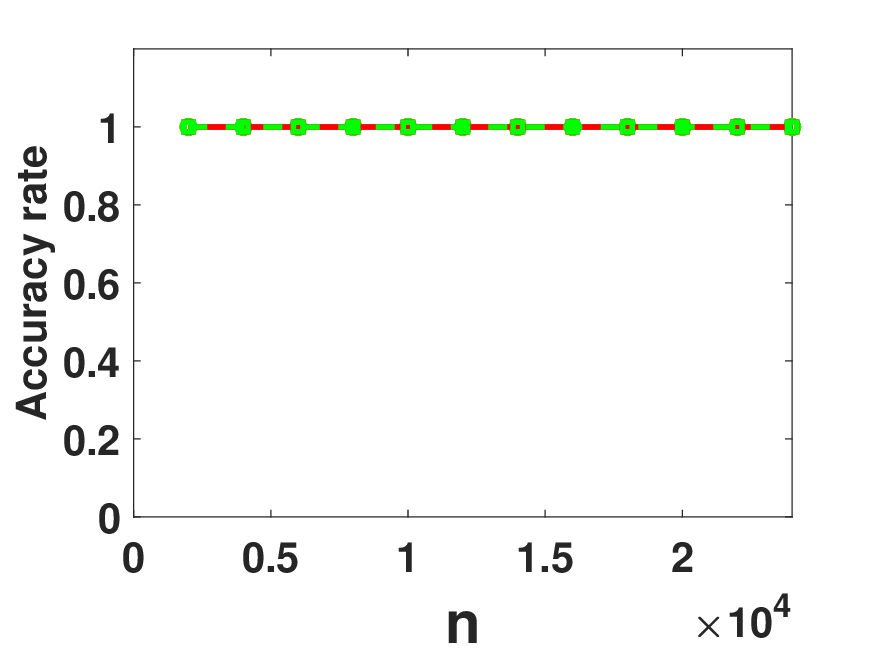}}
\caption{Numerical results of Experiment 5.}
\label{Ex5} 
\end{figure*}
\begin{figure*}
\centering
\subfigure[]{\includegraphics[width=0.33\textwidth]{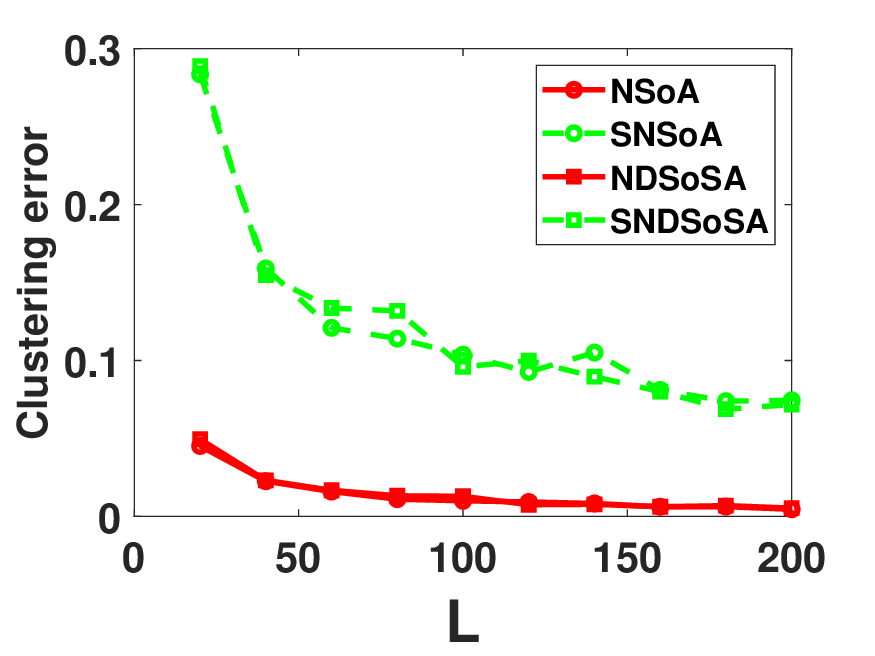}}
\subfigure[]{\includegraphics[width=0.33\textwidth]{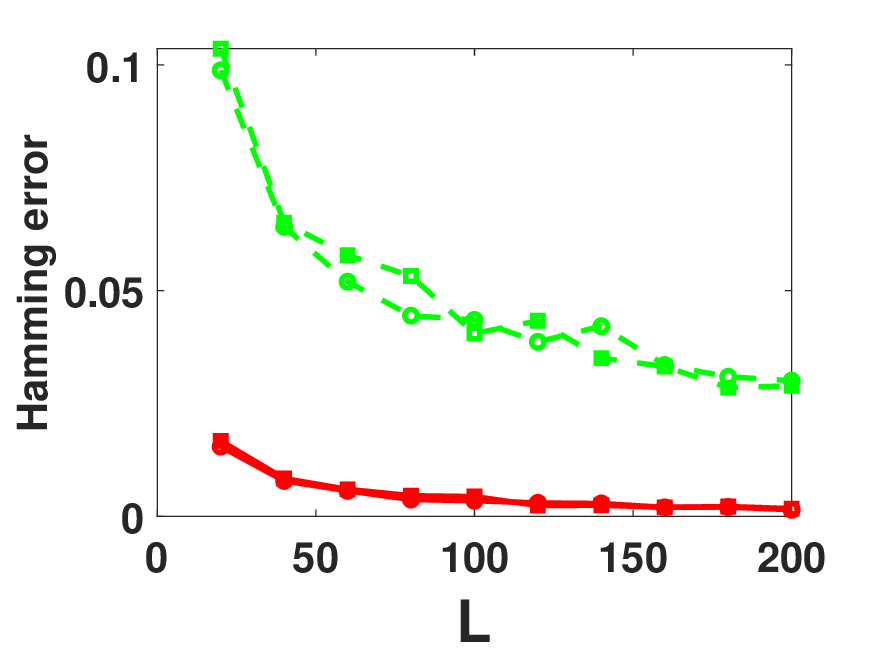}}
\subfigure[]{\includegraphics[width=0.33\textwidth]{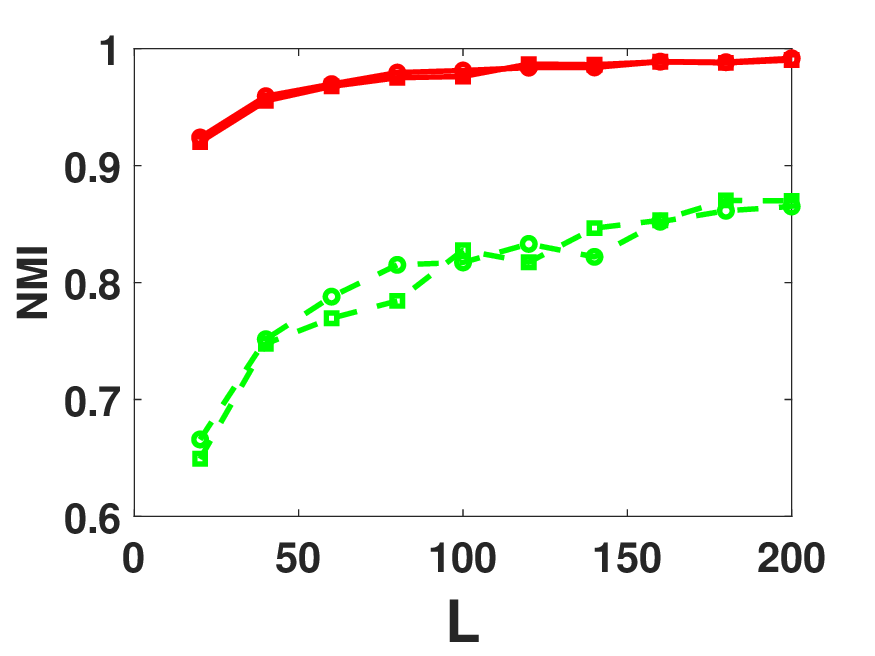}}
\subfigure[]{\includegraphics[width=0.33\textwidth]{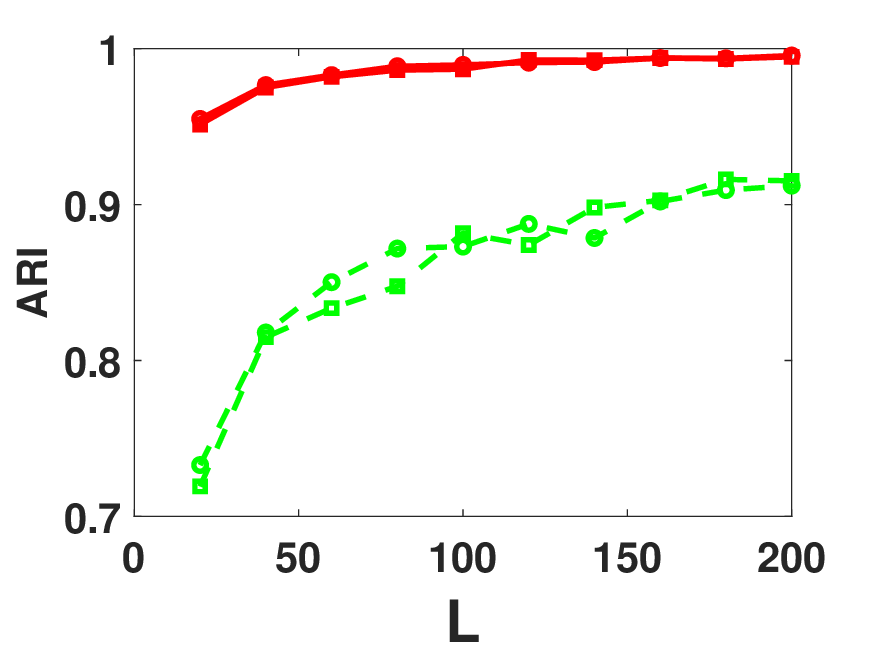}}
\subfigure[]{\includegraphics[width=0.33\textwidth]{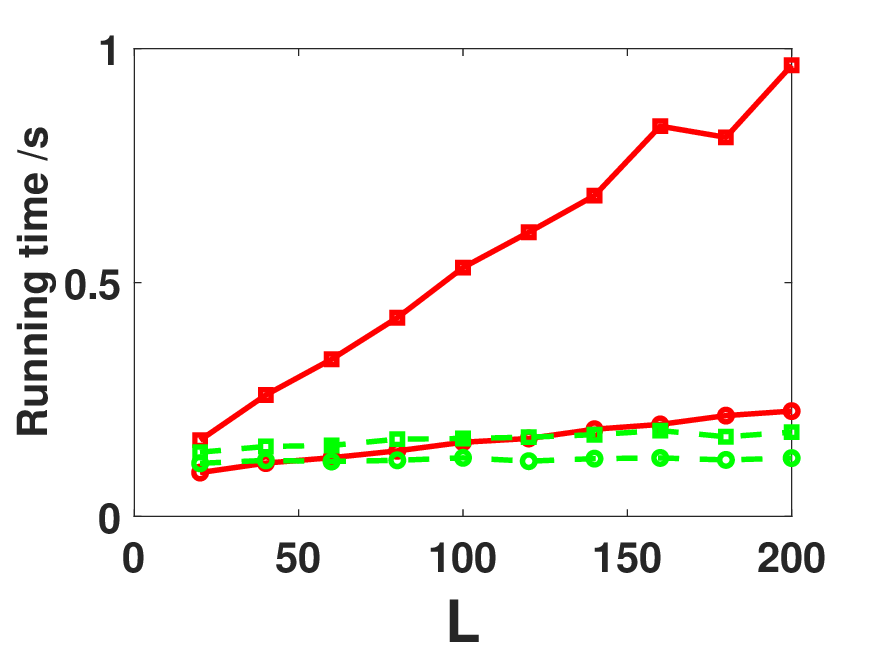}}
\subfigure[]{\includegraphics[width=0.33\textwidth]{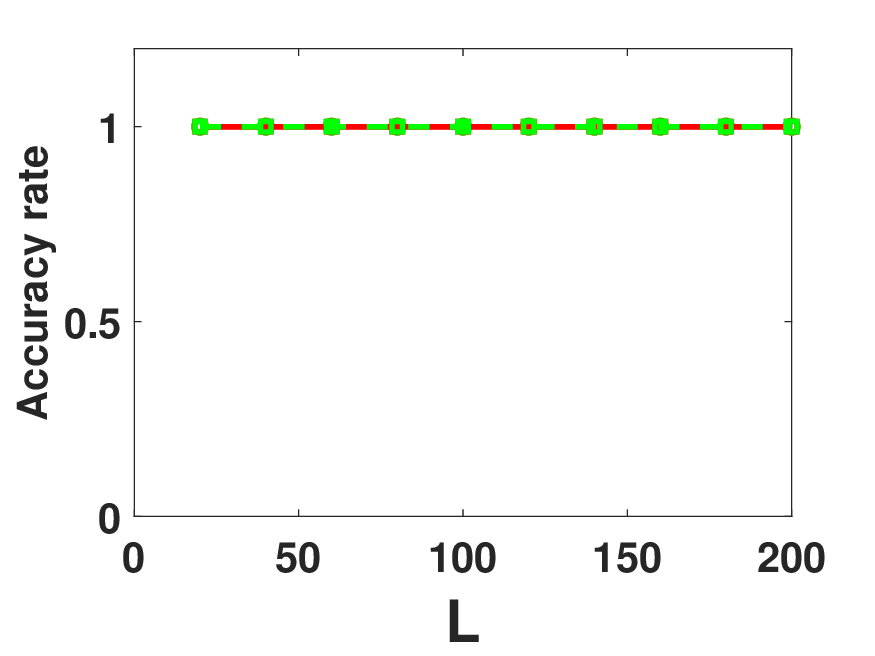}}
\caption{Numerical results of Experiment 6.}
\label{Ex6} 
\end{figure*}
\begin{rem}
In the simulation study, the choice of different parameters is designed to provide a comprehensive comparison of the performance of different community detection algorithms within the MLDCSBM framework. While the specific values chosen may appear arbitrary at first glance, they are selected with careful consideration to align with common real-world multi-layer network characteristics and satisfy the assumptions required for theoretical analysis.

Firstly, the number of nodes $ n $ and layers $ L $ are selected to represent varying sizes of multi-layer networks. In real-world scenarios, the number of nodes and layers can differ significantly across networks, ranging from small-scale networks with tens of nodes and layers to large-scale networks with thousands or millions of nodes and layers. The selected values, such as $n=500$ (or $n\in\{100, 200, \ldots, 1000\}$) and $L=100$ (or $L\in\{2, 4, \ldots,40\}$), are arbitrary yet reasonable to test the algorithms under different scales.

Secondly, the block connectivity matrices $\{B_l\}^{L}_{l=1} $ are generated with elements drawn from a Uniform distribution on [0,1], reflecting the uncertainty in edge probabilities between communities in real-world multi-layer networks. This setup ensures that the generated multi-layer networks exhibit diverse connectivity patterns across layers, consistent with real-world scenarios where relationships within and between communities may vary significantly.

Thirdly, the number of communities $ K $ is also chosen to represent scenarios commonly observed in real-world networks. In many real-world contexts, nodes can be naturally grouped into a moderate number of communities based on shared characteristics or interactions.

Fourthly, the sparsity parameter $ \rho $ is set to values that comply with the sparsity requirements outlined in Assumptions \ref{Assum1} and \ref{Assum2}. This ensures that the generated networks have a realistic level of sparsity, reflecting the fact that real-world networks are often sparse with relatively few edges compared to the total number of potential edges.

Finally, the heterogeneity parameter $ \theta $ is drawn from a Uniform distribution on [0,1] to capture the variability in node degrees observed in real-world multi-layer networks. This setup reflects the heterogeneous nature of many real-world systems because nodes' degrees are often heterogeneous in real data.

In summary, the choice of parameters in the simulation study aligns with common real-world multi-layer network characteristics and satisfies the theoretical assumptions required for our analysis. The flexibility of the MLDCSBM framework allows for these parameters to be varied widely, providing a robust comparison for community detection algorithms under a range of settings that mimic real-world scenarios.
\end{rem}
\section{Real data applications}\label{sec6realdata}
In this section, we utilize several real-world multi-layer networks with unknown ground-truth community labels to assess the performance of our methods. Before comparing various methods, we conducted a series of pre-processing steps on these datasets. Initially, for networks exhibiting directed or weighted characteristics, we transformed them into undirected and unweighted versions by disregarding their directions and weights. Subsequently, we narrowed our focus to nodes belonging to the largest connected component within the aggregate matrix $A_{\mathrm{sum}}$. Table \ref{realdata} presents the basic statistics of these multi-layer networks after pre-processing, where $\nu := \frac{\text{Total number of edges}}{L \times \frac{n(n-1)}{2}}$ is employed to quantify the sparsity of real-world multi-layer networks, where $L \times \frac{n(n-1)}{2}$ represents the maximum potential connections in a multi-layer network with $n$ nodes and $L$ layers. We observe that these networks vary widely in their sparsity, with some networks (like Lazega Law Firm) being relatively dense and others (like EU-Air and MUS-GPI) being extremely sparse. The paper's focus on spectral clustering methods and their performance under different sparsity levels provides insight into how these algorithms can be adapted for real-world multi-layer networks with varying degrees of sparsity. For visualization, we plot the adjacency matrices of all layers for Lazega Law Firm, C.Elegans, CS-Aarhus, and CSMMLN in Fig.~\ref{lazegaA}, Fig.~\ref{CelegansA}, Fig.~\ref{CSA}, and Fig.~\ref{CSMMLNA}, respectively. A concise overview of these multi-layer networks is provided below.
\begin{table}[h!]
\small
	\centering
	\caption{Basic information of real-world multi-layer networks considered in this paper.}
	\label{realdata}
	\begin{tabular}{cccccccccccc}
\hline\hline
Dataset&$n$&$L$&$\#$ Edges&$\nu$\\
\hline
Lazega Law Firm&71&3&2223&0.2982\\
C.Elegans&279&3&5863&0.0504\\
CS-Aarhus&61&5&620&0.0678\\
CSMMLN&260&8&43078&0.1599\\
EU-Air&417&37&3588&0.0011\\
FAO-trade&214&364&318346&0.0384\\
MUS-GPI&7747&7&19842&0.000094473\\
\hline\hline
\end{tabular}
\end{table}
\begin{figure*}
\centering
\subfigure[]{\includegraphics[width=0.33\textwidth]{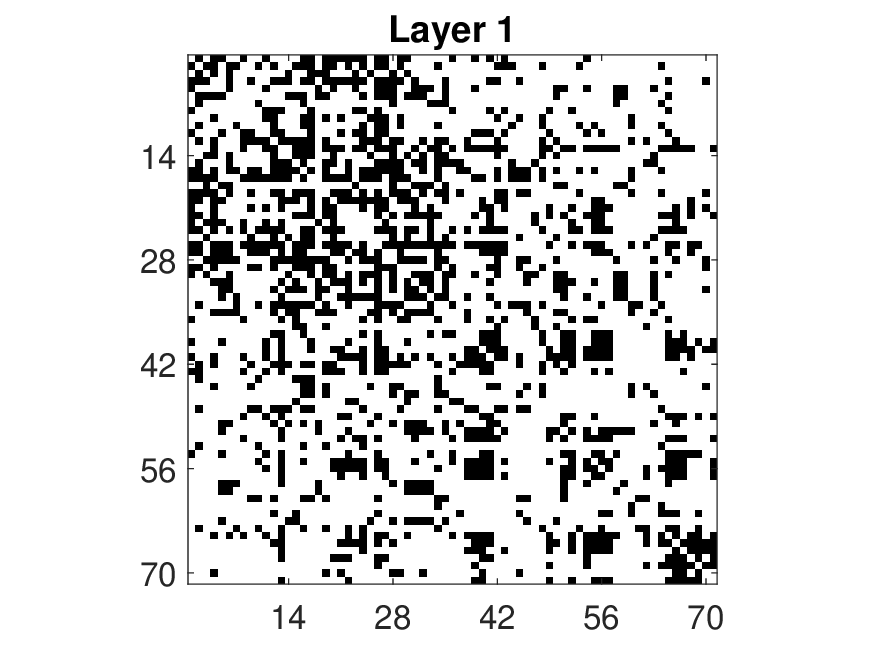}}
\subfigure[]{\includegraphics[width=0.33\textwidth]{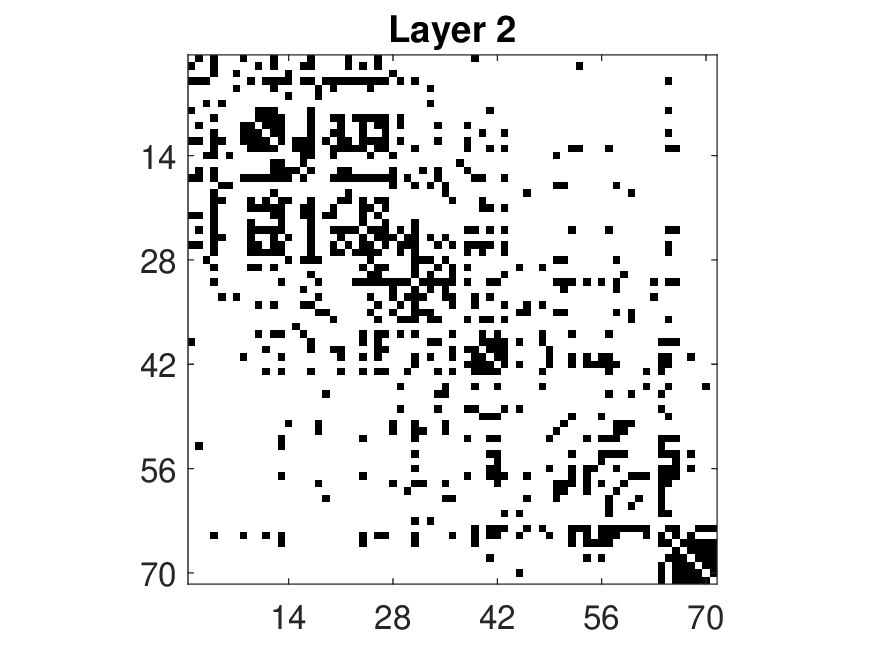}}
\subfigure[]{\includegraphics[width=0.33\textwidth]{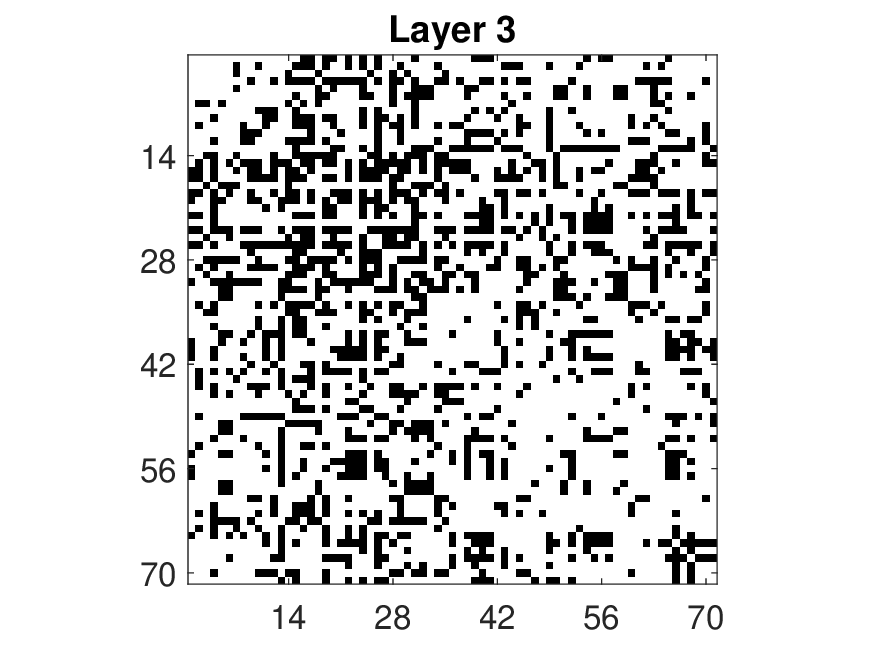}}
\caption{The adjacency matrices of the 3 layers of Lazega Law Firm. Matrix element values are displayed in gray scale with black corresponding to 1 and white to 0.}
\label{lazegaA} 
\end{figure*}
\begin{itemize}
\item \textbf{Lazega Law Firm}: It is a multi-layer social network that consists of 3 kinds of (Co-work, Friendship, and Advice) between partners and associates of a corporate law partnership \citep{snijders2006new}.
  \item \textbf{C.Elegans}:It is a multi-layer neuronal network that contains information about the connection of Caenorhabditis elegans, where the multiplex consists of layers corresponding to different synaptic junctions: electric (``ElectrJ"), chemical monadic (``MonoSyn"), and polyadic (``PolySyn") \citep{chen2006wiring}.
  \item \textbf{CS-Aarhus}: It is a multi-layer social network that consists of five kinds of online and offline relationships (Facebook, Leisure, Work, Co-authorship, Lunch) between the employees of the Computer Science department at Aarhus \citep{magnani2013combinatorial}.
  \item \textbf{Chinese Stock Market Multi-Layer Network (CSMMLN)}: It is a multi-layer stock network which we have constructed utilizing averaged monthly stock log return data from the Shenzhen Stock Exchange and the Shanghai Stock Exchange, covering the period from January 2022 to March 2024. A total of 260 valid stocks are identified over the 27 monthly observations. The data is sourced from RESSET (available at \url{https://www.resset.com/enindex}). For the $i$-th stock, let $Y_{i}$ represent a $1\times27$ vector, where $Y_{i}(t)$ records the averaged monthly log return for the $t$-th month for $t\in[27]$. Additionally, let $C$ denote a $260\times 260$ correlation matrix, where $C(i,j)$ represents the Pearson correlation coefficient between $Y_{i}$ and $Y_{j}$ for $i, j \in [260]$. Unlike works in \citep{huang2009network,chi2010network,li2022undirected}, which analyze the network structure of single-layer stock networks using a single threshold, we construct our multi-layer stock network using varying thresholds based on the Pearson correlation matrix $C$. Specifically, we consider eight threshold values: $\{0.4, 0.45, 0.5, 0.55, 0.6, 0.65, 0.7, 0.75\}$. Let $\tau$ be a $1\times8$ vector containing these thresholds. Let $A_{l}$ represent the adjacency matrix of the $l$-th stock network, where $A_{l}(i,j)=1$ if $|C(i,j)|>\tau(l)$, and $A_{l}(i,j)=0$ if $|C(i,j)|\leq\tau(l)$, for $l \in [8]$. In this way, we create a stock market multi-layer network with 260 nodes and 8 layers. As observed from Fig.~\ref{CSMMLNA}, the stock network becomes sparser with an increase in the threshold value.
  \item \textbf{EU-Air}: It is a multi-layer air-transportation network that is composed of 37 different layers each one corresponding to a different airline operating in Europe \citep{cardillo2013emergence}.
  \item \textbf{FAO-trade}: It is a multi-layer economic network that records different types of trade relationships among countries, obtained from FAO (Food and Agriculture Organization of the United Nations) \citep{de2015structural}. For this data, layers represent products, nodes are countries and edges represent import/export relationships of a specific food product among countries.
  \item \textbf{MUS multiplex GPI network (MUS-GPI)}: It is a multi-layer genetic network that collects seven types (physical association, association, direct interaction, colocalization, additive genetic interaction defined by inequality, synthetic genetic interaction defined by inequality, and suppressive genetic interaction defined by inequality) of genetic interactions for organisms \citep{stark2006biogrid}. The original data is directed and we transform it to undirected by ignoring its direction in this paper.
\end{itemize}
These real-world multi-layer networks considered in this paper can be downloaded from \url{https://manliodedomenico.com/data.php}. Given that edges in these real-world multi-layer networks typically represent relationships such as co-work, friendships, synaptic junctions, trade connections, air transportation, and genetic interactions, these networks can be regarded as assortative. This implies that for these multi-layer networks, nodes sharing common characteristics tend to form dense connections, thereby allowing the application of the averaged modularity $Q_{MNavrg}$ to assess the quality of community partitions generated by various algorithms for these real-world multi-layer networks.
\begin{figure*}
\centering
\subfigure[]{\includegraphics[width=0.33\textwidth]{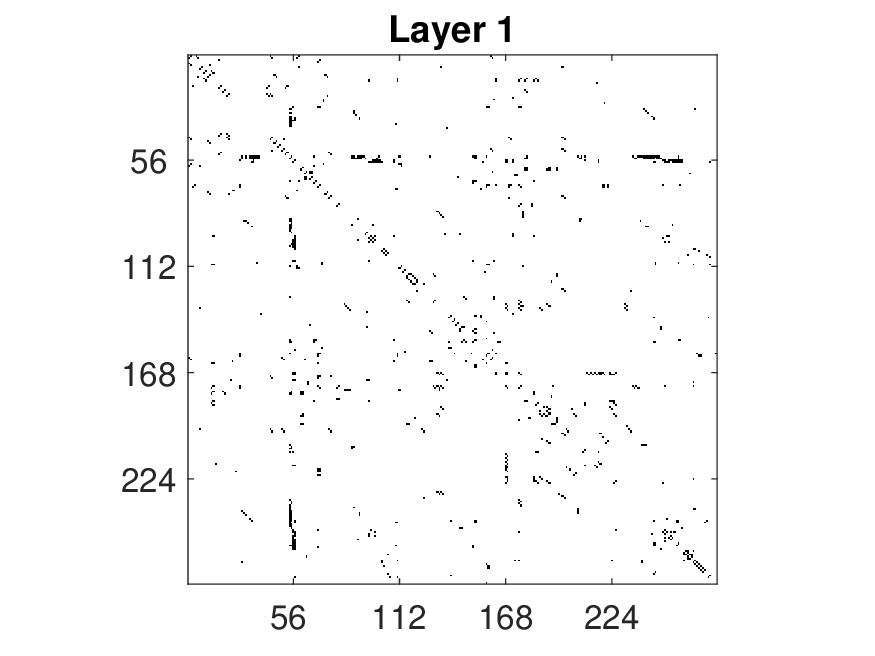}}
\subfigure[]{\includegraphics[width=0.33\textwidth]{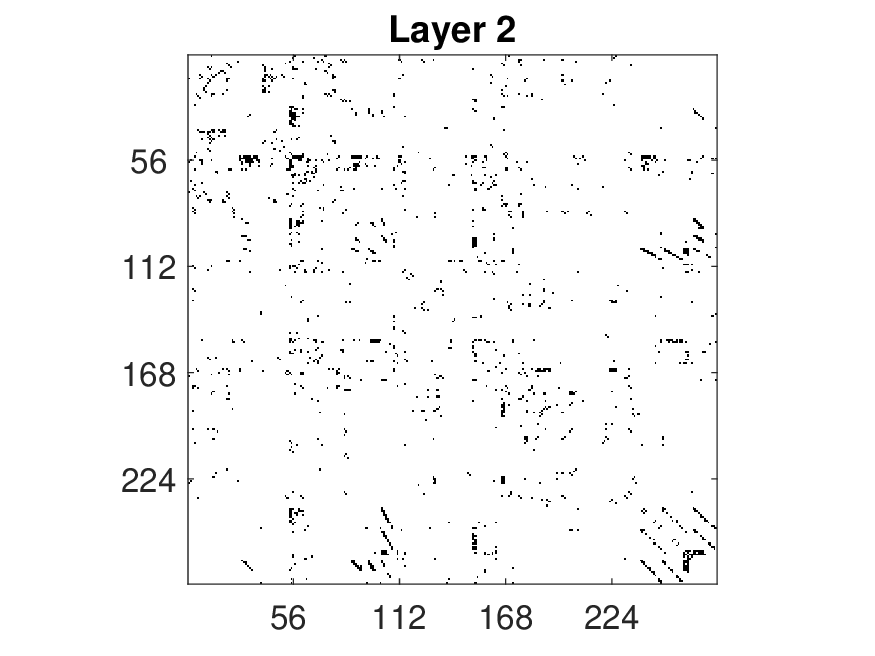}}
\subfigure[]{\includegraphics[width=0.33\textwidth]{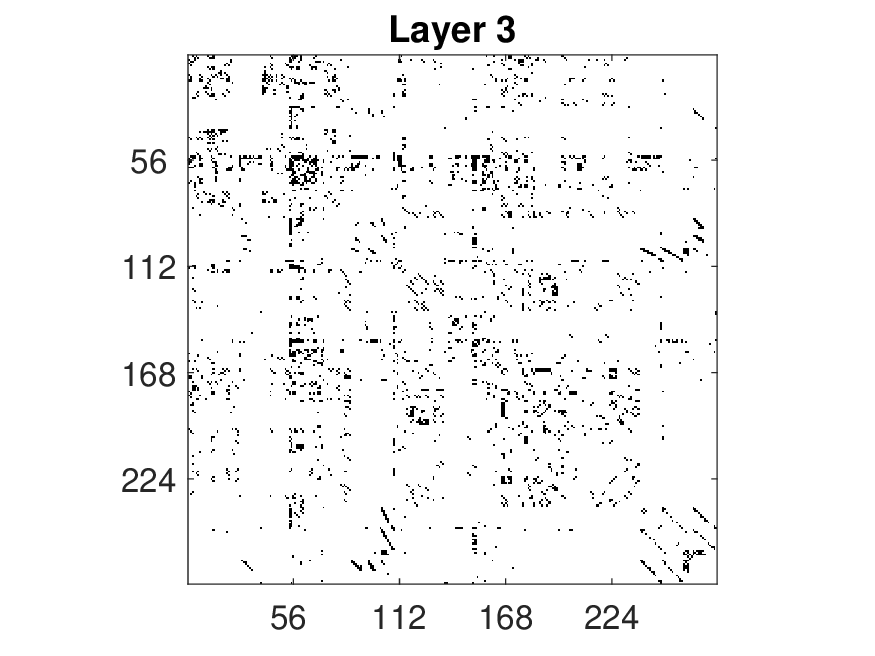}}
\caption{The adjacency matrices of the 3 layers of C.Elegans. Matrix element values are displayed in gray scale with black corresponding to 1 and white to 0.}
\label{CelegansA} 
\end{figure*}
\begin{table*}[h!]
\scriptsize
	\centering
	\caption{$(\hat{K},Q_{MNavrg})$ returned by each method for real data in Table \ref{realdata}, where $\hat{K}$ means estimated number of communities.}
	\label{realdataQ}
	\begin{tabular}{cccccccccccccc}
\hline\hline
Dataset&NSoA&NDSoSA&NSoSA&Sum&SoS-Debias&MASE&KMAM&CAMSBM&AWP&SNSoA&SNDSoSA\\
\hline
Lazega Law Firm&(3,0.2553)&(3,0.2553)&(3,0.2553)&(4,0.2495)&(3,0.2406)&(3,0.2377)&(3,0.2469)&(4,0.2416)&(3,0.2640)&(3,0.2553)&(3,0.2553)\\
C.Elegans&(2,0.3532)&(2,0.3570)&(2,0.3524)&(2,0.3524)&(2,0.3570)&(2,0.3426)&(6,0.2165)&(2,0.3502)&(8,0.2088)&(2,0.3532)&(2,0.3570)\\
CS-Aarhus&(5,0.4920)&(4,0.4913)&(6,0.4811)&(5,0.4862)&(5,0.4252)&(5,0.4392)&(4,0.4729)&(3,0.4381)&(5,0.4650)&(5,0.4920)&(4,0.4913)\\
CSMMLN&(3,0.3629)&(3,0.3582)&(3,0.3629)&(3,0.3405)&(3,0.3453)&(3,0.3405)&(3,0.3251)&(3,0.3432)&(4,0.3305)&(3,0.3629)&(3,0.3582)\\
EU-Air&(1,0)&(1,0)&(1,0)&(1,0)&(1,0)&(1,0)&(1,0)&(1,0)&(1,0)&(1,0)&(1,0)\\
FAO-trade&(2,0.1731)&(2,0.1658)&(2,0.1656)&(2,0.1315)&(2,0.1456)&(2,0.1131)&(3,0.0839)&(2,0.1154)&(5,0.0930)&(2,0.1699)&(2,0.1651)\\
MUS-GPI&(4,0.2813)&(7,0.2096)&(4,0.2044)&(6,0.0579)&(4,0.0576)&(5,0.0347)&(7,0.1019)&(6,0.0380)&(3,0.0601)&(4,0.1209)&(8,0.0911)\\
\hline\hline
\end{tabular}
\end{table*}

\begin{figure*}
\centering
\subfigure[]{\includegraphics[width=0.18\textwidth]{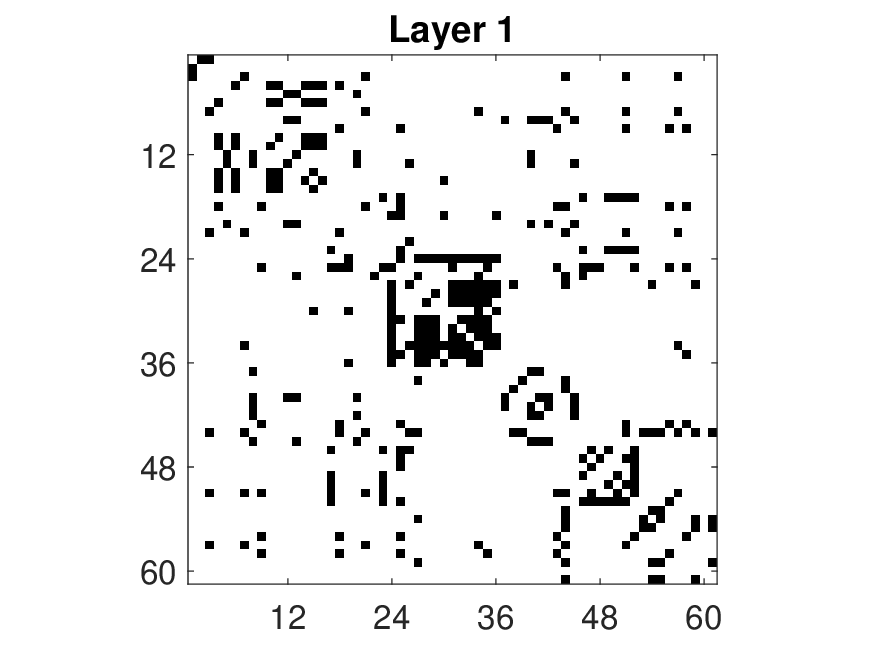}}
\subfigure[]{\includegraphics[width=0.18\textwidth]{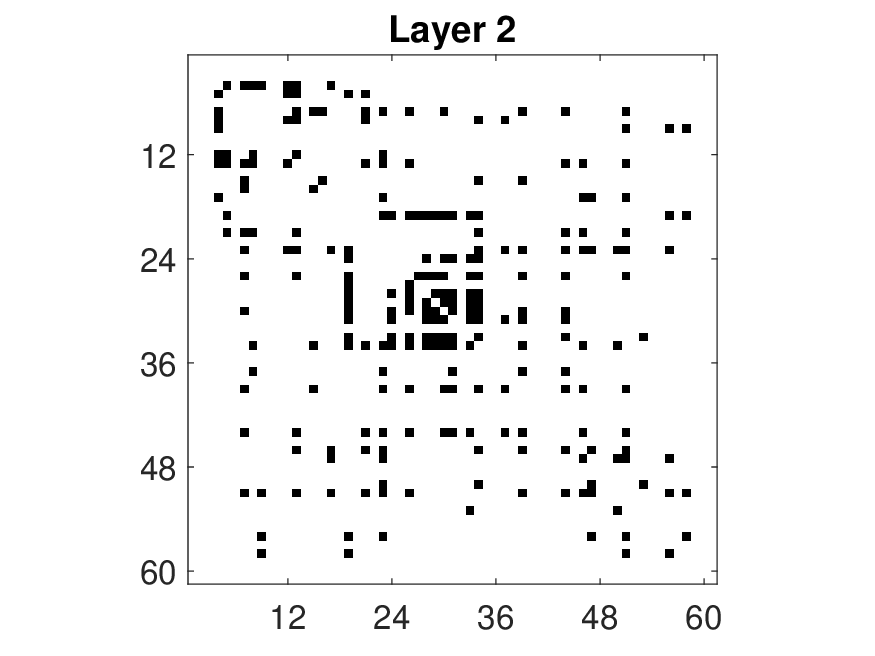}}
\subfigure[]{\includegraphics[width=0.18\textwidth]{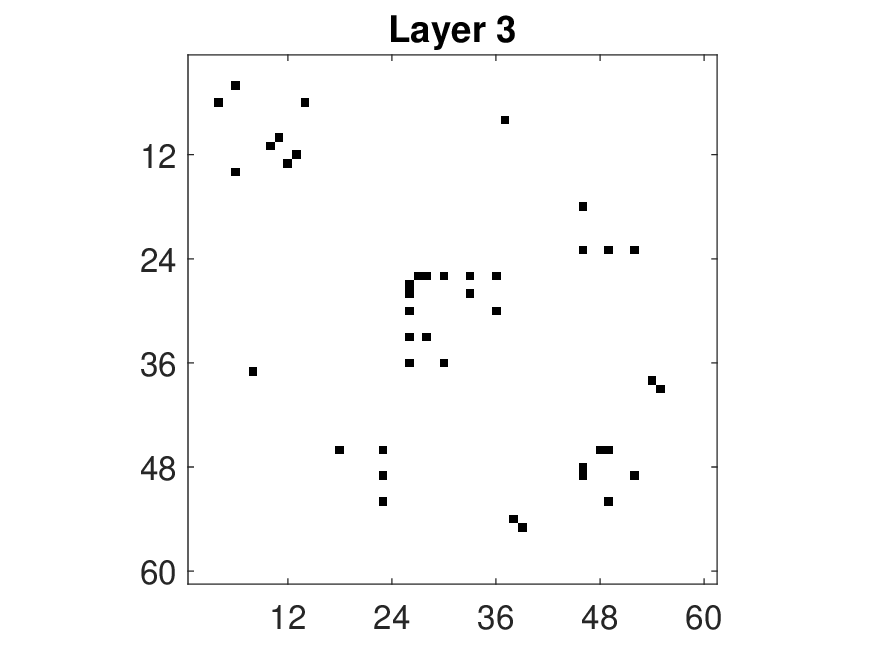}}
\subfigure[]{\includegraphics[width=0.18\textwidth]{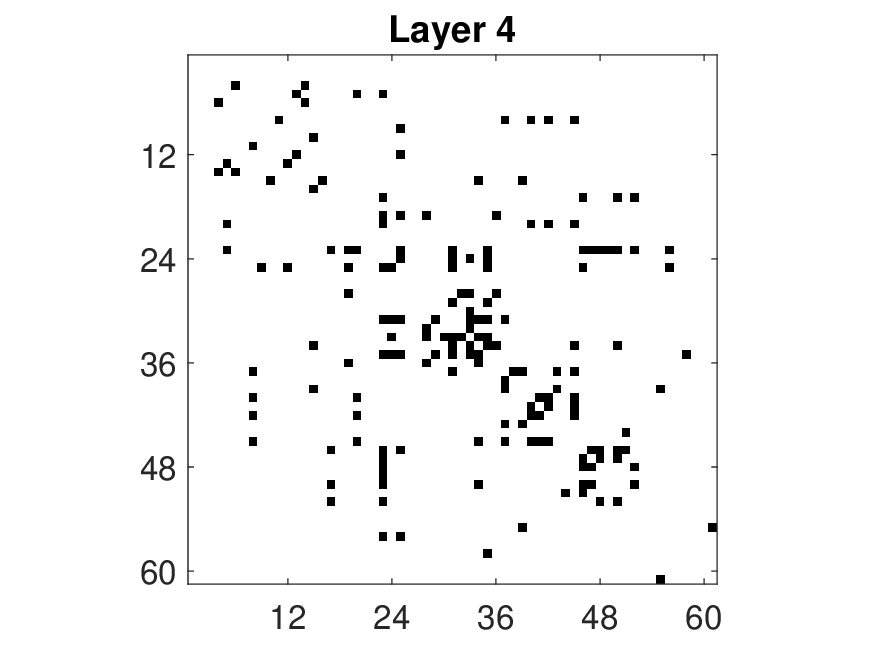}}
\subfigure[]{\includegraphics[width=0.18\textwidth]{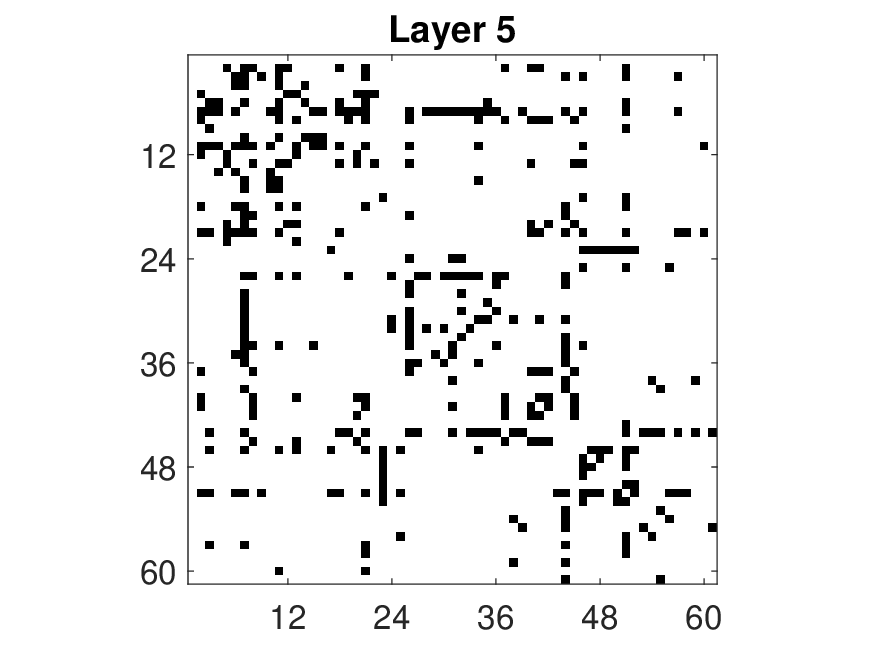}}
\caption{The adjacency matrices of the 5 layers of CS-Aarhus. Matrix element values are displayed in gray scale with black corresponding to 1 and white to 0.}
\label{CSA} 
\end{figure*}

\begin{figure*}
\centering
\subfigure[]{\includegraphics[width=0.245\textwidth]{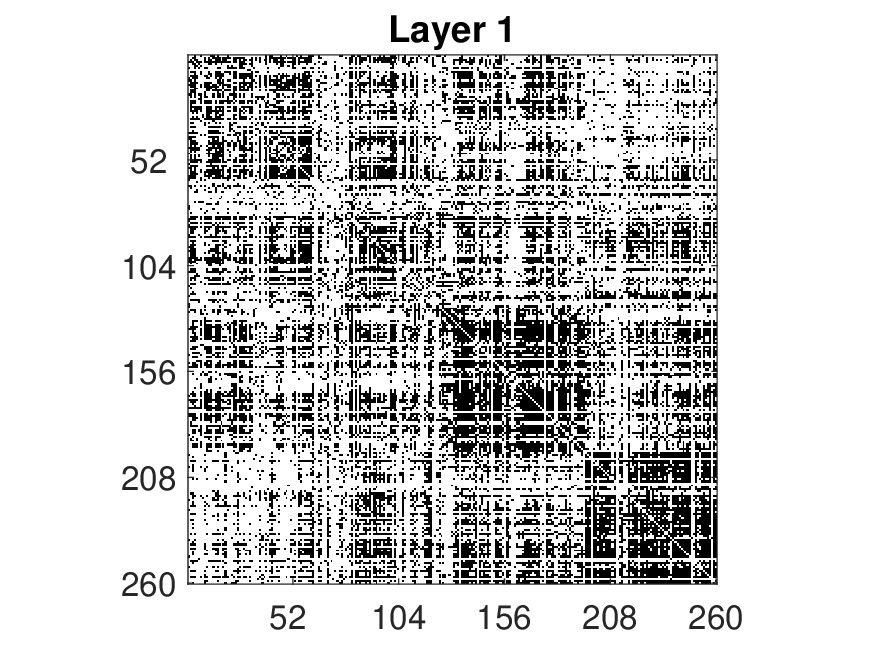}}
\subfigure[]{\includegraphics[width=0.245\textwidth]{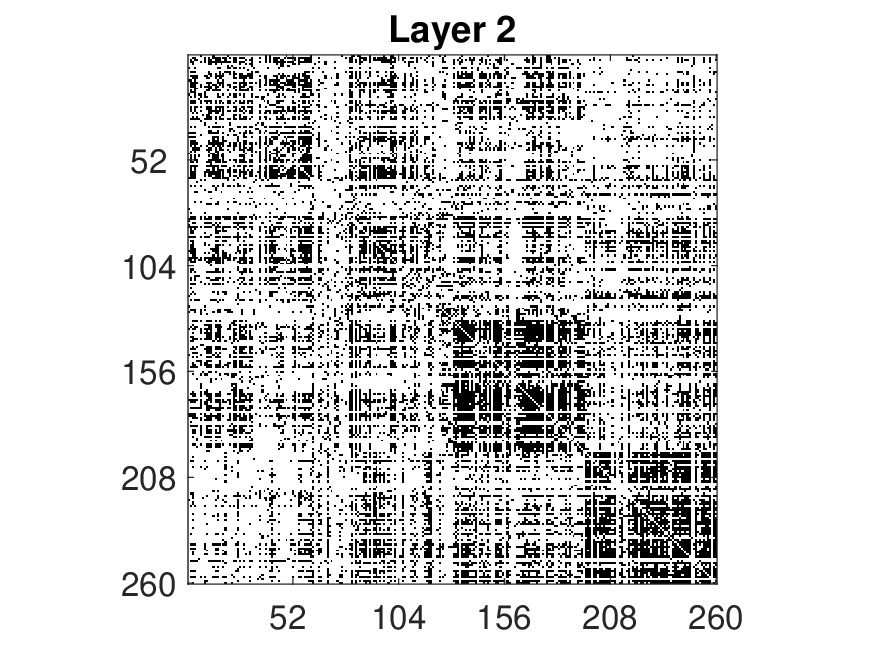}}
\subfigure[]{\includegraphics[width=0.245\textwidth]{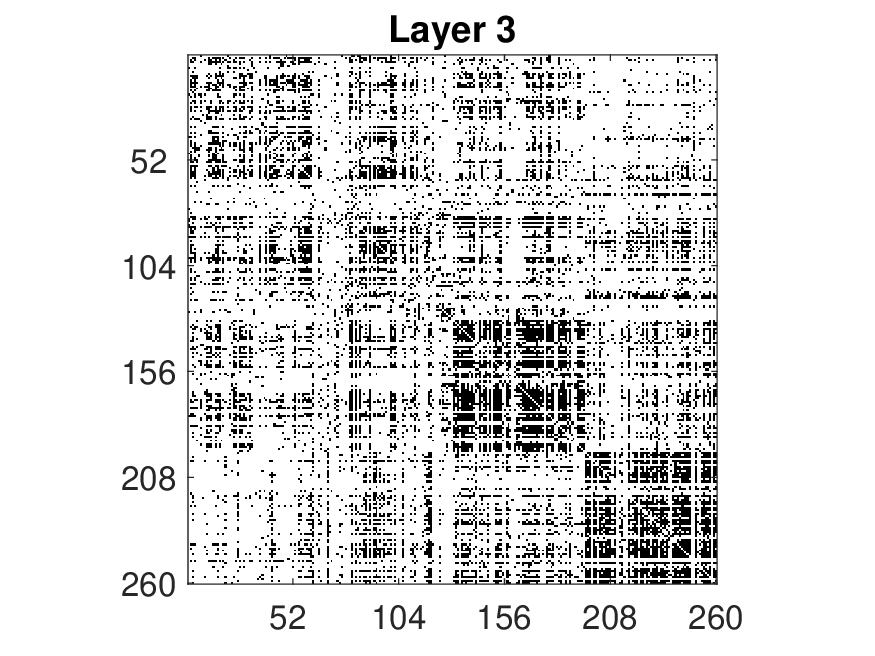}}
\subfigure[]{\includegraphics[width=0.245\textwidth]{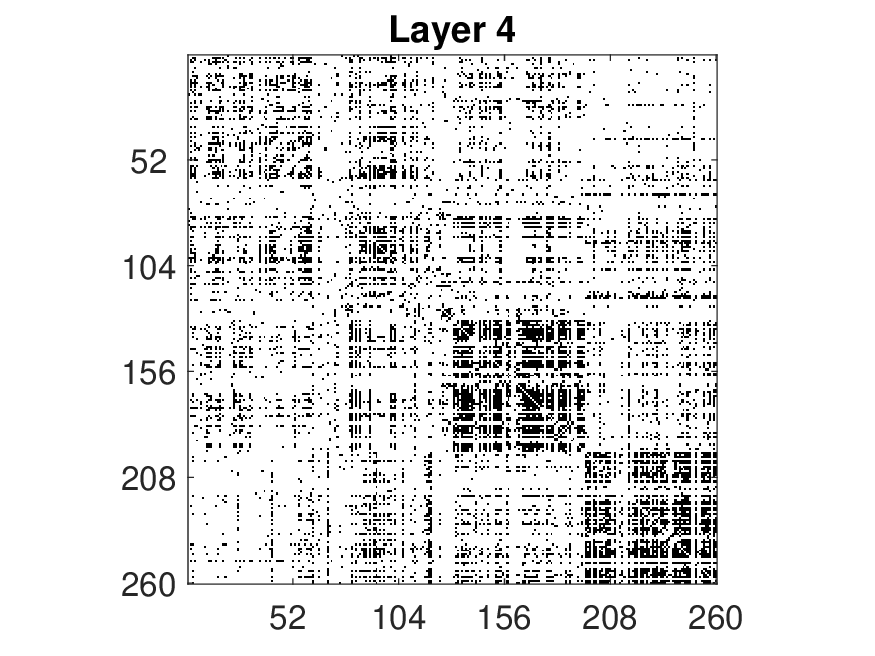}}
\subfigure[]{\includegraphics[width=0.245\textwidth]{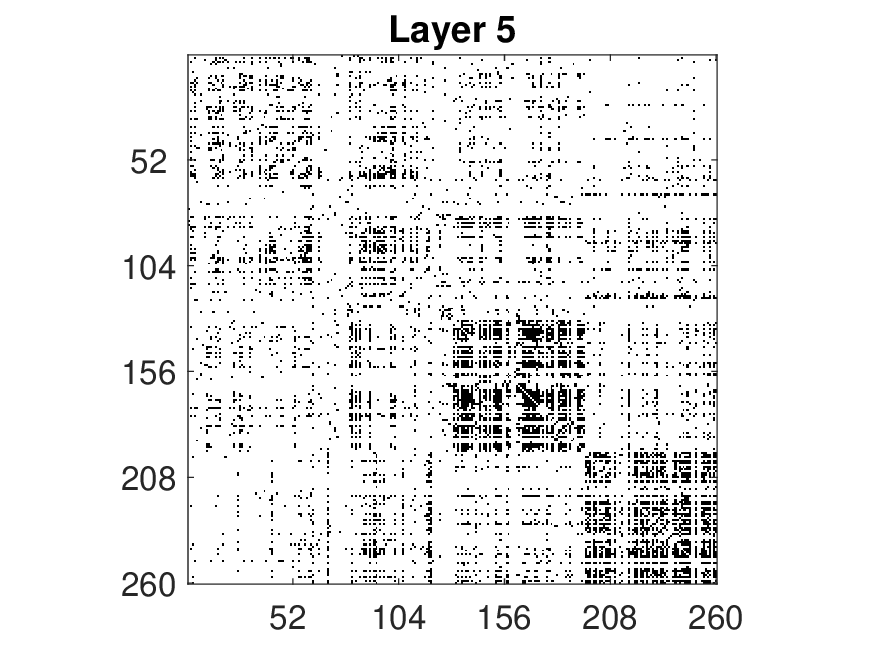}}
\subfigure[]{\includegraphics[width=0.245\textwidth]{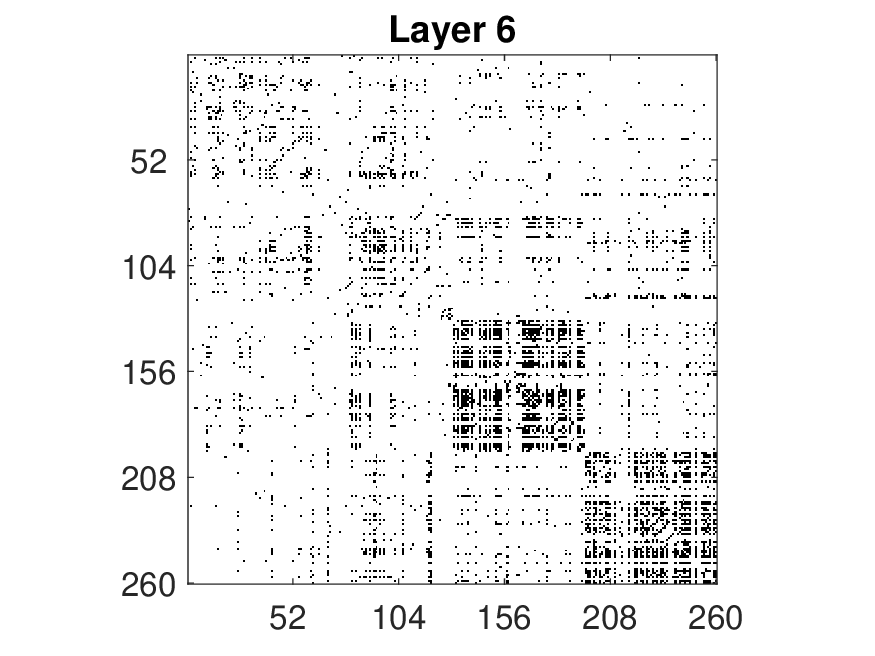}}
\subfigure[]{\includegraphics[width=0.245\textwidth]{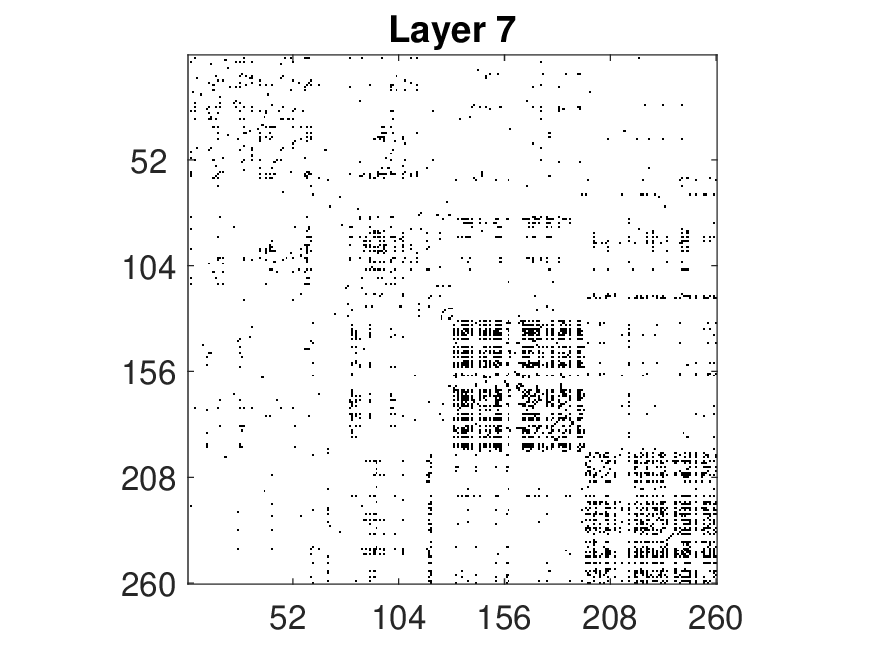}}
\subfigure[]{\includegraphics[width=0.245\textwidth]{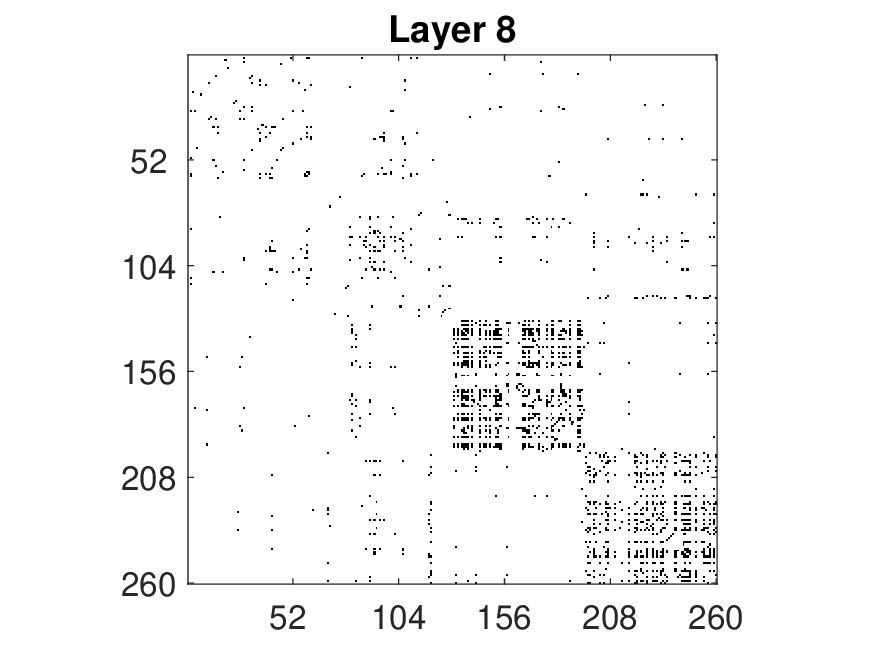}}
\caption{The adjacency matrices of the 8 layers of CSMMLN. Matrix element values are displayed in gray scale with black corresponding to 1 and white to 0.}
\label{CSMMLNA} 
\end{figure*}

Table \ref{realdataQ} reports the estimated number of communities and the respective averaged modularity $Q_{MNavrg}$ of each method for every real data used in this paper. Note that for networks with $n<500$ and $L<10$, SNSoA and SNDSoSA equal NSoA and NDSoSA, respectively. By analyzing Table \ref{realdataQ}, we have the following conclusions:
\begin{itemize}
  \item The averaged modularity of our NSoA and NDSoSA is no smaller than their competitors for all datasets except for C.Elegans.
  \item For Lazega Law Firm, NSoA, NDSoSA, and NSoSA estimate its $K$ as $3$ since their averaged modularity is larger than that of the other four methods.
  \item For C.Elegans, all methods except KMAM and AWP determine its $K$ as 2. Since the averaged modularity of KMAM and AWP is smaller than that of the other methods, as a result, we choose $K=2$ for this data.
  \item For CS-Aarhus, since NDSoSA's averaged modularity is slightly smaller than that of NSoA and NDSoSA always outperforms NSoA in simulations, we choose $K$ as 4 for this data.
  \item For CSMMLN, all approaches except AWP determines its $K$ as 3.
  \item For EU-Air, all methods suggest that these 450 airports in Europe belong to a giant community.
  \item For FAO-trade, all methods except KMAM and AWP determine its $K$ as 2.
  \item For MUS-GPI, NSoA and NDSoSA exhibit a higher average modularity than their competitors. Given that NSoA yields the highest average modularity, we tend to believe that the number of communities for this dataset is $4$.
\end{itemize}
\begin{table*}[h!]
\scriptsize
	\centering
	\caption{Running time (averaged of 10 repetitions) for each method in detecting communities of real data in Table \ref{realdata} using the optimal $K$ value identified in the previous analysis.}
	\label{realdataTime}
	\begin{tabular}{cccccccccccccc}
\hline\hline
Dataset&NSoA&NDSoSA&NSoSA&Sum&SoS-Debias&MASE&KMAM&CAMSBM&AWP&SNSoA&SNDSoSA\\
\hline
Lazega Law Firm&0.0629s&0.0585s&0.0534s&0.0570s&0.0568s&0.0633s&0.0746s&0.0642s&0.0069s&0.0629s&0.0585s\\
C.Elegans&0.1040s&0.1067s&0.1041s&0.1630s&0.1756s&0.1580s&0.2839s&0.1610s&0.0741s&0.1040s&0.1067s\\
CS-Aarhus&0.0594s&0.0565s&0.0546s&0.0593s&0.0616s&0.0632s&0.0725s&0.0690s&0.0090s&0.0594s&0.0565s\\
CSMMLN&0.1360s&0.1313s&0.1225s&0.0981s&0.1060s&0.1195s&0.2607s&0.1052s&0.0777s&0.1360s&0.1313s\\
EU-Air&0.0485s&0.1065s&0.0932s&0.0406s&0.1053s&0.0965s&0.1744s&0.2002s&0.8727s&0.0292s&0.0524s\\
FAO-trade&0.0896s&0.1939s&0.1790s&0.1011s&0.2129s&0.6204s&0.7738s&0.6796s&2.9203s&0.0677s&0.0808s\\
MUS-GPI&5.6786s&68.5446s&66.7791s&4.4409s&67.9146s&14.1555s&163.4885s&597.3381s&1508.9847s&1.9927s&5.2131s\\
\hline\hline
\end{tabular}
\end{table*}

Table \ref{realdataTime} records the running times of each method for the real-world multi-layer networks used in this paper. The following observations can be made:
\begin{itemize}
  \item For networks with $n>500$ or $L>10$, SNSoA outperforms SNDSoSA, and both methods surpass all other methods in terms of speed, suggesting their suitability for handling larger-scale multi-layer networks.
  \item The proposed methods, NSoA and NDSoSA, generally perform competitively or superiorly compared to NSoSA, Sum, SoS-DEbias, MASE, KMAM, CAMSBM, and AWP in terms of running time across all datasets, highlighting their effectiveness.
  \item For small-scale multi-layer networks (Lazega Law Firm, C.Elegans, CS-Aarhus, and CSMMLN), all methods exhibit rapid running times with insignificant differences, indicating their manageable computational complexity for small-scale datasets.
  \item For larger-scale multi-layer networks (EU-Air, FAO-trade, and MUS-GPI), running times diverge significantly. NSoA consistently outperforms NDSoSA, and both proposed methods are significantly faster than methods such as KMAM, CAMSBM, and AWP.
  \item The results show a clear trend of increasing running times with the size and complexity of the networks. This is expected, as larger networks typically require more computation to detect communities accurately.
\end{itemize}

Fig.~\ref{lazegaAC} plots the adjacency matrices of the 3 layers for Lazega Law Firm, sorted according to the common community assignment estimated by NDSoSA with 3 communities. The red grid lines mean the community partitions for the 3 estimated communities. We also plot the adjacency matrices sorted by community assignment from NDSoSA for C.Elegans, CS-Aarhus, and CSMMLN in Fig.~\ref{CelegansAC}, Fig.~\ref{CSAC}, and Fig.~\ref{CSMMLNAC}, respectively. From these three figures, it is easy to observe that nodes within the same communities (diagonal part in each adjacency matrix ) connect more than across communities (off-diagonal part in each adjacency matrix) for all layers. For visualization, Fig.~\ref{lazegaN}, Fig.~\ref{CelegansN}, Fig.~\ref{CSN}, and Fig.~\ref{CSMMLNN} show the clustering results of NDSoSA for the three networks Lazega Law Firm, C.Elegans, CS-Aarhus, and CSMMLN respectively, where nodes are colored by their community assignment.

Next, we conduct a more in-depth analysis of the CSMMLN network. Table \ref{CSMMLNIndustryDistribution} presents information on its communities, revealing two primary characteristics: (1) The majority of stocks in Communities 1, 2, and 3 originate from the Manufacturing, Finance, and Real Estate industries, respectively. (2) Stocks within the same industry can belong to different clusters. Tables \ref{CSMMLNCommunity1}, \ref{CSMMLNCommunity2}, and \ref{CSMMLNCommunity3} list the top 10 stocks in each of these communities, ranked by their total degree in the CSMMLN network. Our observations indicate that: (a) While most stocks in Community 1 are from the Manufacturing Industry, the top 10 stocks in this community predominantly come from the Construction and Finance industries. This suggests that stocks in these two industries have a significant impact on Manufacturing stocks within Community 1. (b) Nine out of the top 10 stocks in Community 2 belong to the Finance Industry, indicating its strong influence on various stocks in this community. (c) All top 10 stocks in Community 3 are from the Real Estate Industry, highlighting its dominant influence in this community. Fig.~\ref{stockcommunity123} illustrates the average monthly log returns of the top 10 stocks in the three communities. This figure reveals that stock returns within each community exhibit similar trends over time, indicating a high correlation in their performance. This characteristic is advantageous for portfolio management, as it allows investors to diversify their holdings within a community while maintaining consistent exposure to market movements. For example, an investor can create a portfolio with multiple stocks from Community 1, confident that these stocks will likely exhibit similar performance. Furthermore, investors can utilize the community structure to hedge their portfolios by holding stocks from different communities, thereby offsetting potential losses in one community with gains in another and reducing overall portfolio risk. By identifying stocks with highly correlated returns, investors can construct more efficient and diversified portfolios. The NDSoSA method, which leverages the multi-layer structure of financial networks, offers a powerful tool for uncovering these hidden community structures and aiding investors in making informed decisions.

\begin{figure*}
\centering
\subfigure[]{\includegraphics[width=0.33\textwidth]{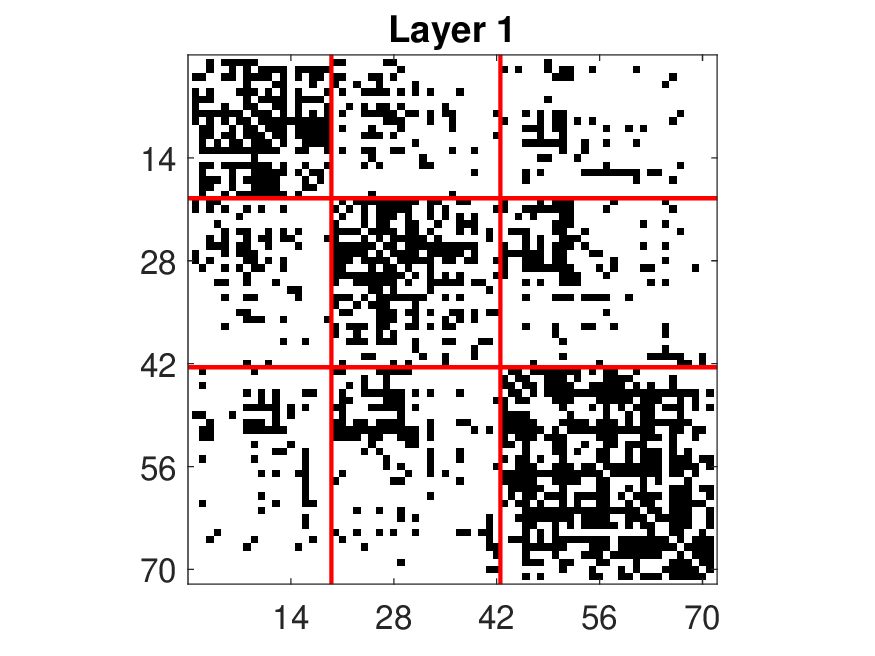}}
\subfigure[]{\includegraphics[width=0.33\textwidth]{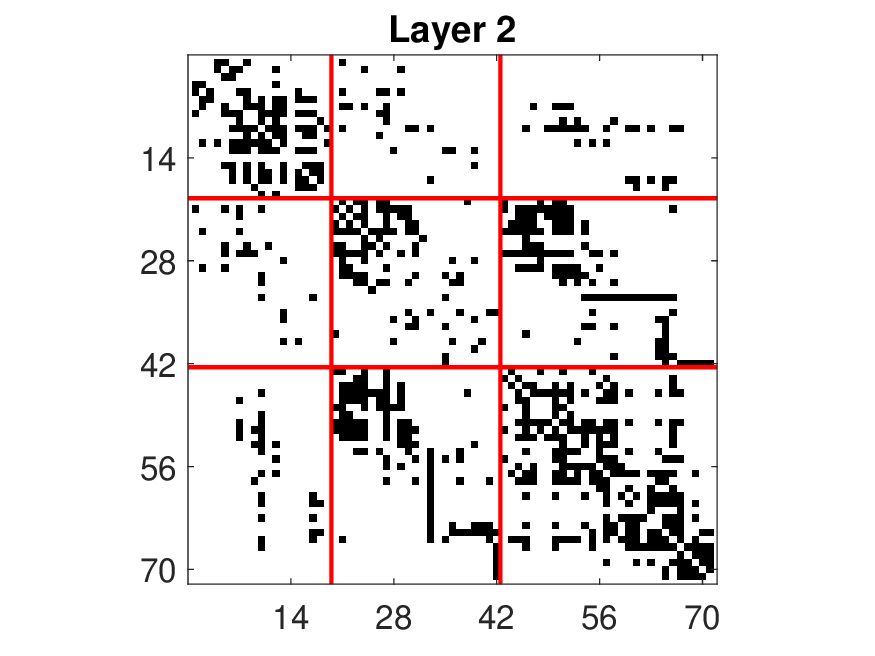}}
\subfigure[]{\includegraphics[width=0.33\textwidth]{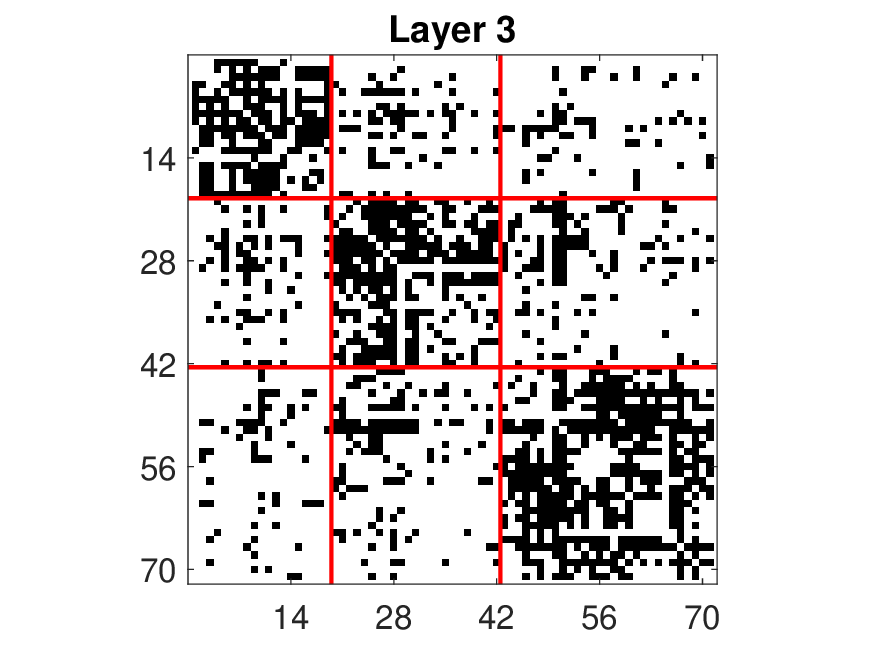}}
\caption{The adjacency matrices of the 3 layers with the rows and columns being sorted according to the common community labels obtained from NDSoSA with 3 communities for Lazega Law Firm. The colored red lines indicate community partitions. Matrix element values are displayed in gray scale with black corresponding to 1 and white to 0.}
\label{lazegaAC} 
\end{figure*}
\begin{figure*}
\centering
\subfigure[]{\includegraphics[width=0.33\textwidth]{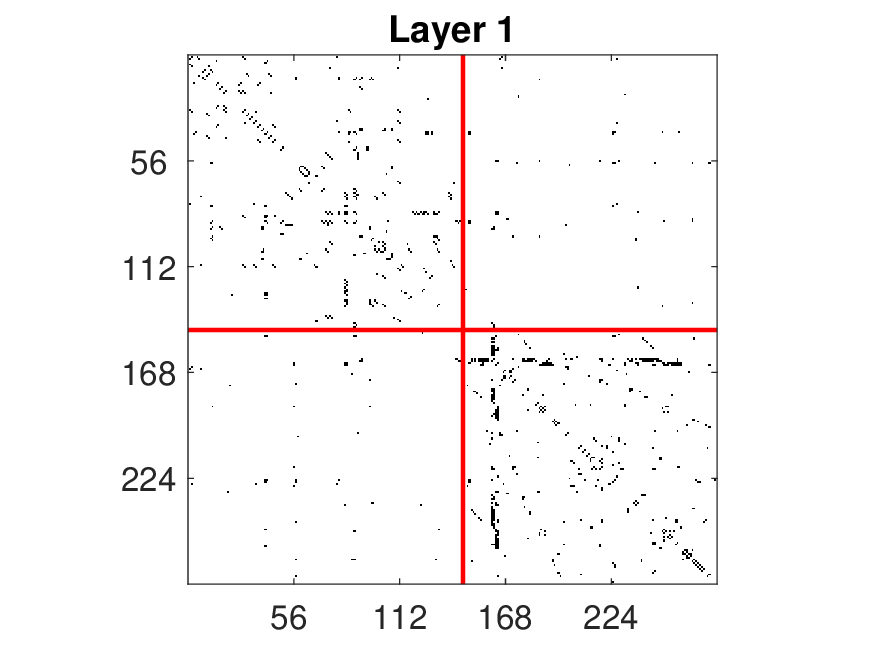}}
\subfigure[]{\includegraphics[width=0.33\textwidth]{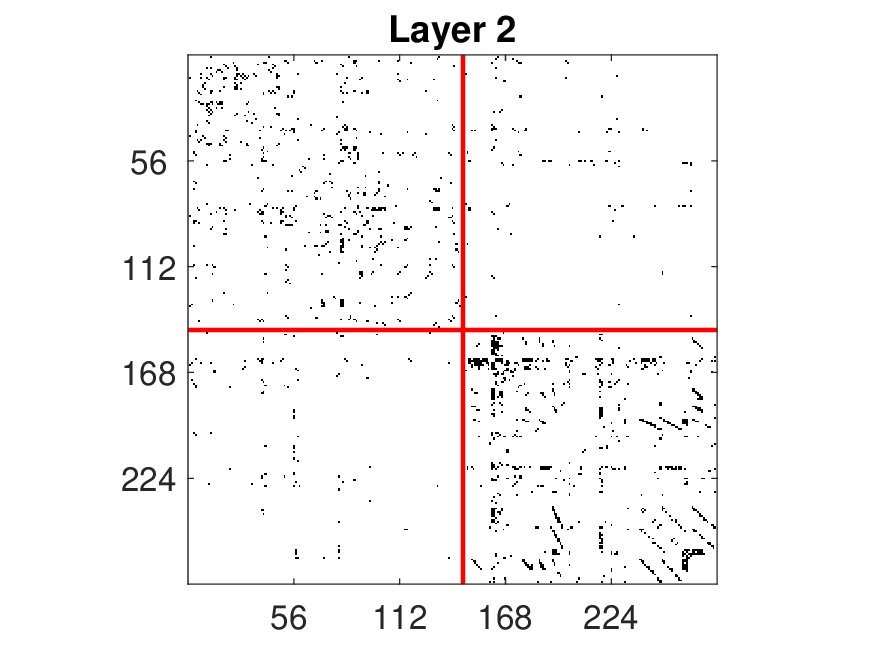}}
\subfigure[]{\includegraphics[width=0.33\textwidth]{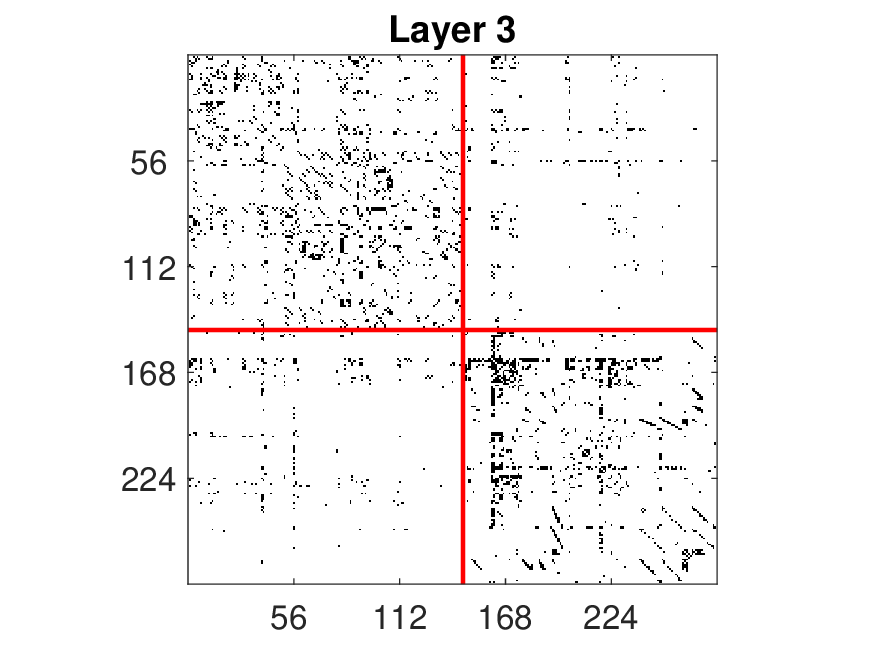}}
\caption{The adjacency matrices of the 3 layers sorted according to the common community labels obtained from NDSoSA with 2 communities for C.Elegans. The colored red lines indicate community partitions. Matrix element values are displayed in gray scale with black corresponding to 1 and white to 0.}
\label{CelegansAC} 
\end{figure*}
\begin{figure*}
\centering
\subfigure[]{\includegraphics[width=0.18\textwidth]{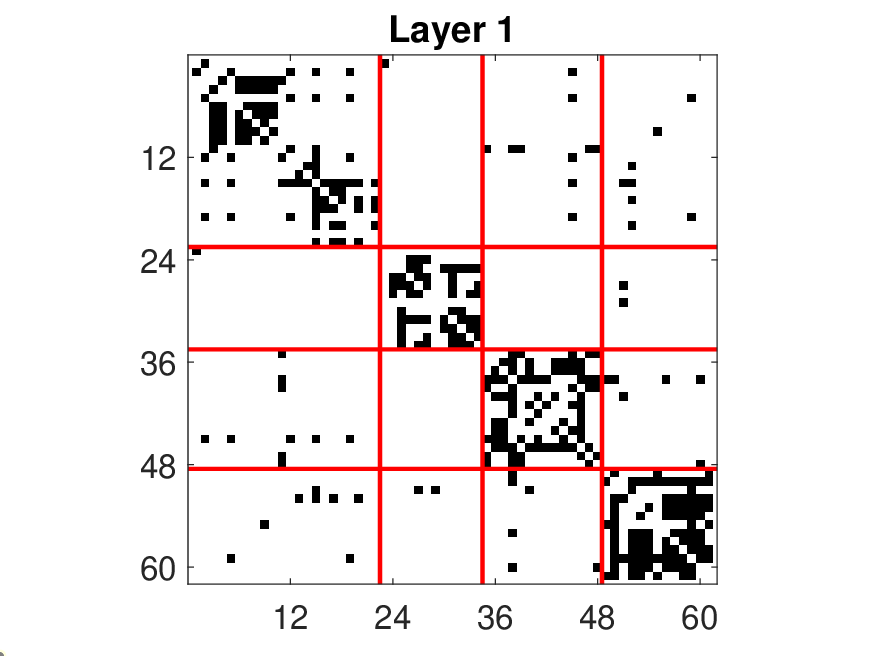}}
\subfigure[]{\includegraphics[width=0.18\textwidth]{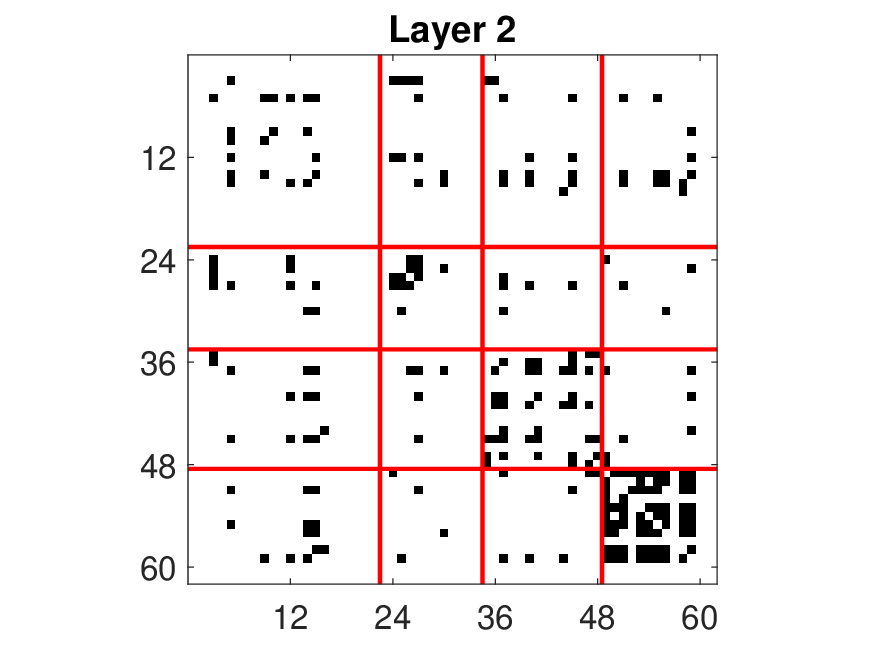}}
\subfigure[]{\includegraphics[width=0.18\textwidth]{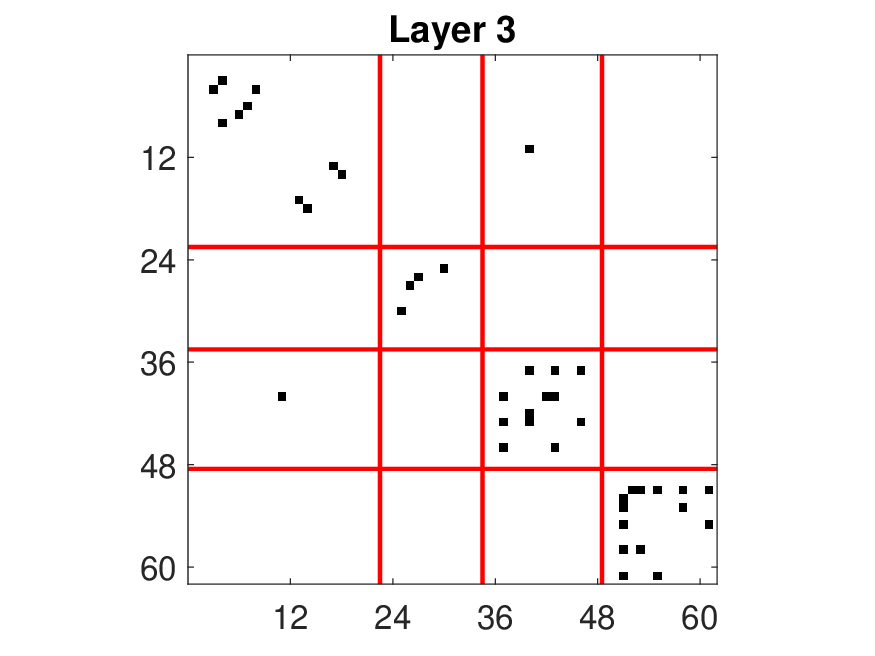}}
\subfigure[]{\includegraphics[width=0.18\textwidth]{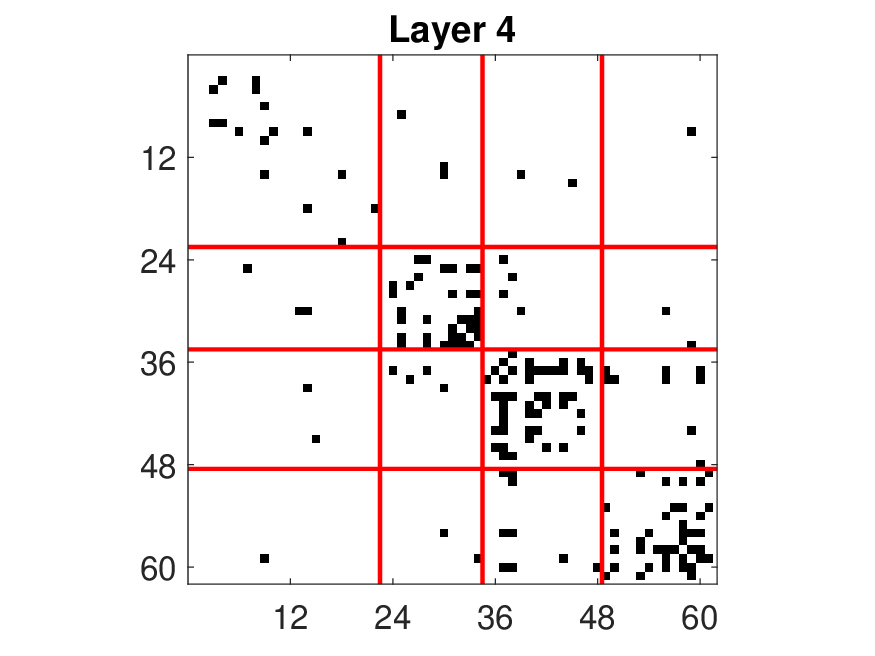}}
\subfigure[]{\includegraphics[width=0.18\textwidth]{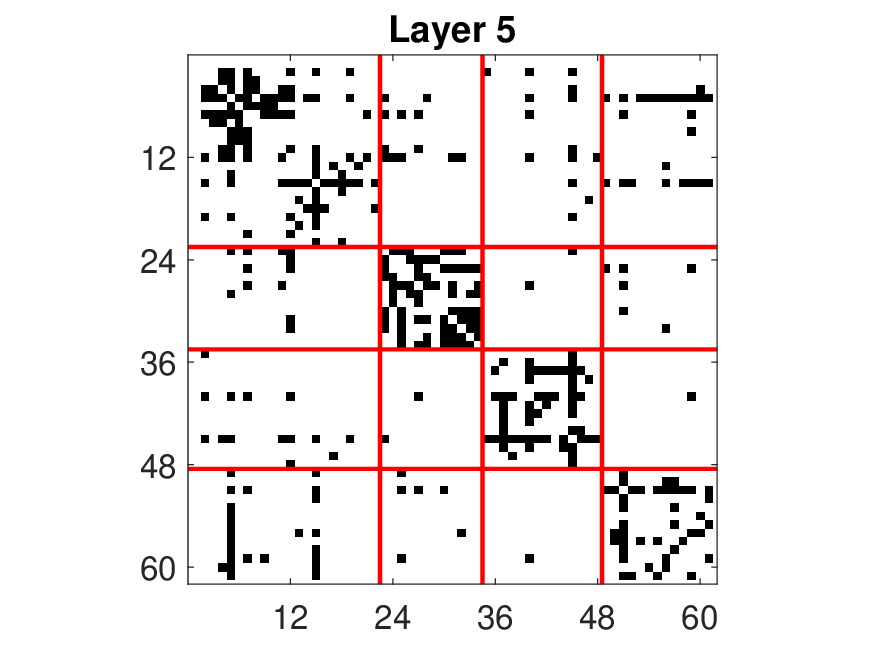}}
\caption{The adjacency matrices of the 3 layers sorted according to the common community labels obtained from NDSoSA with 4 communities for CS-Aarhus. The colored red lines indicate community partitions. Matrix element values are displayed in gray scale with black corresponding to 1 and white to 0.}
\label{CSAC} 
\end{figure*}

\begin{figure*}
\centering
\subfigure[]{\includegraphics[width=0.245\textwidth]{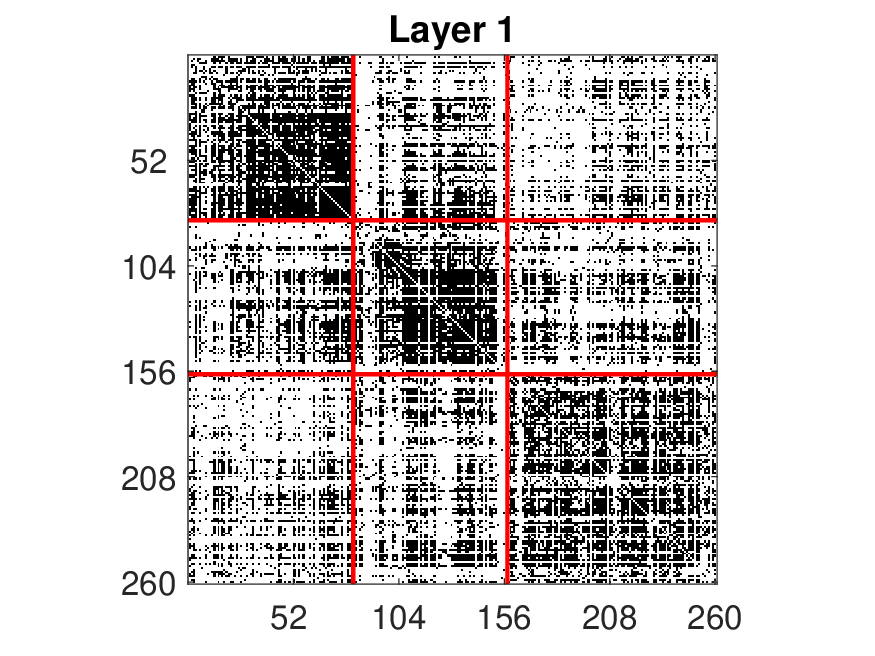}}
\subfigure[]{\includegraphics[width=0.245\textwidth]{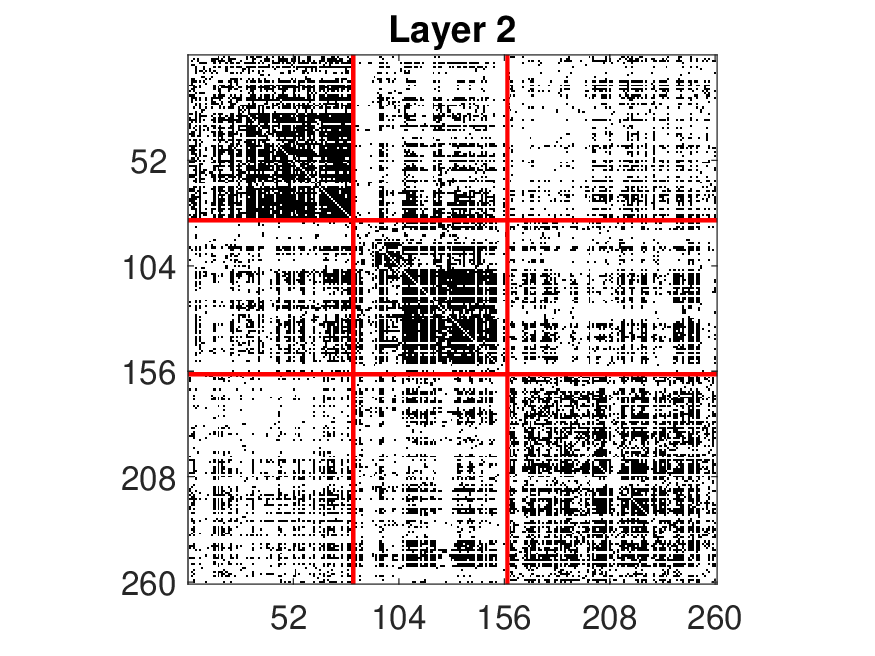}}
\subfigure[]{\includegraphics[width=0.245\textwidth]{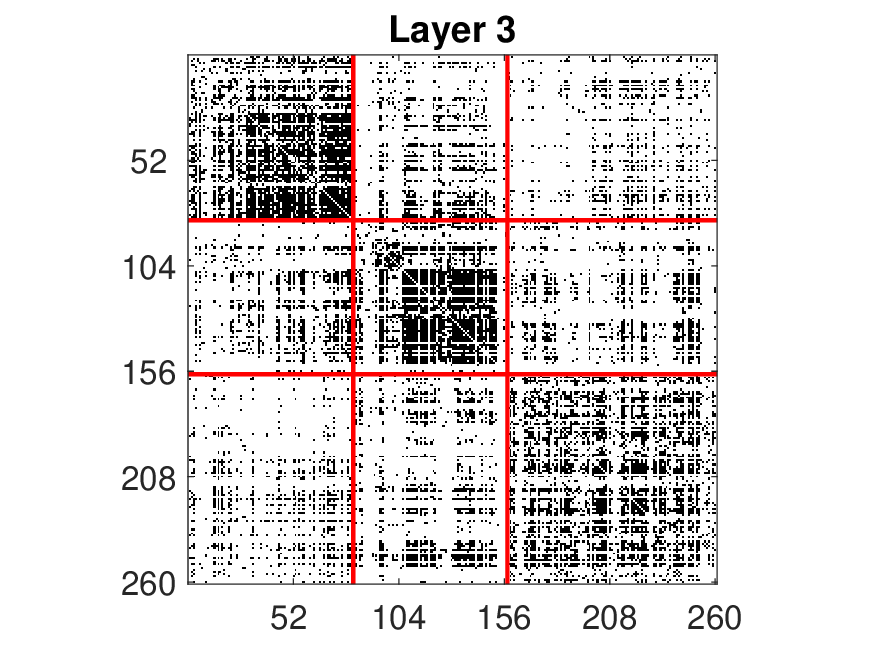}}
\subfigure[]{\includegraphics[width=0.245\textwidth]{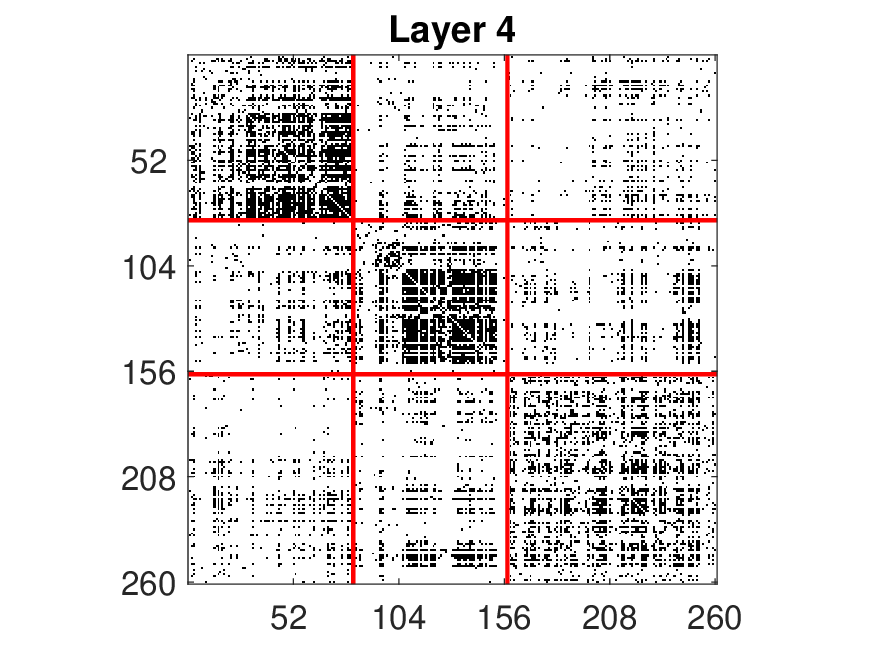}}
\subfigure[]{\includegraphics[width=0.245\textwidth]{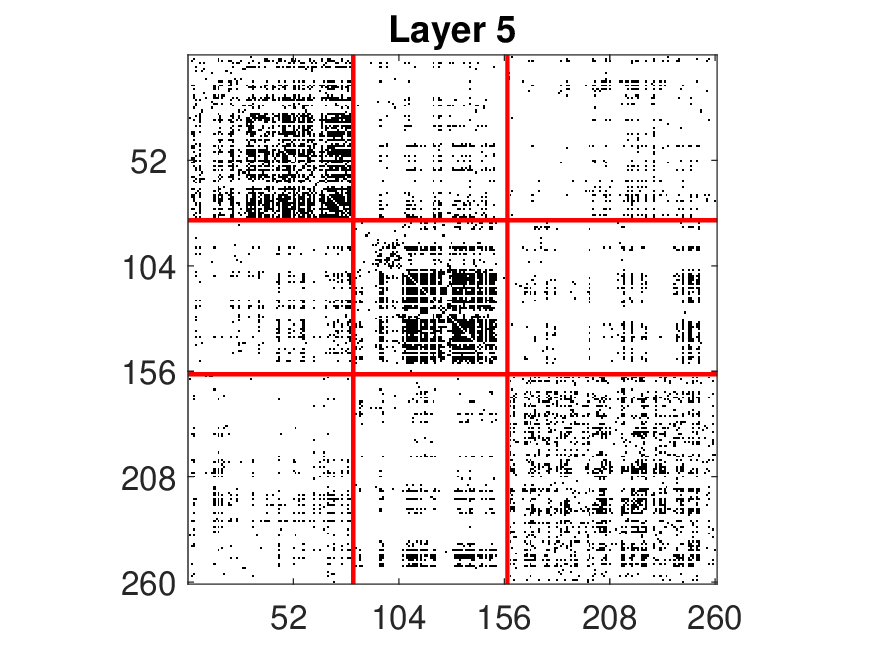}}
\subfigure[]{\includegraphics[width=0.245\textwidth]{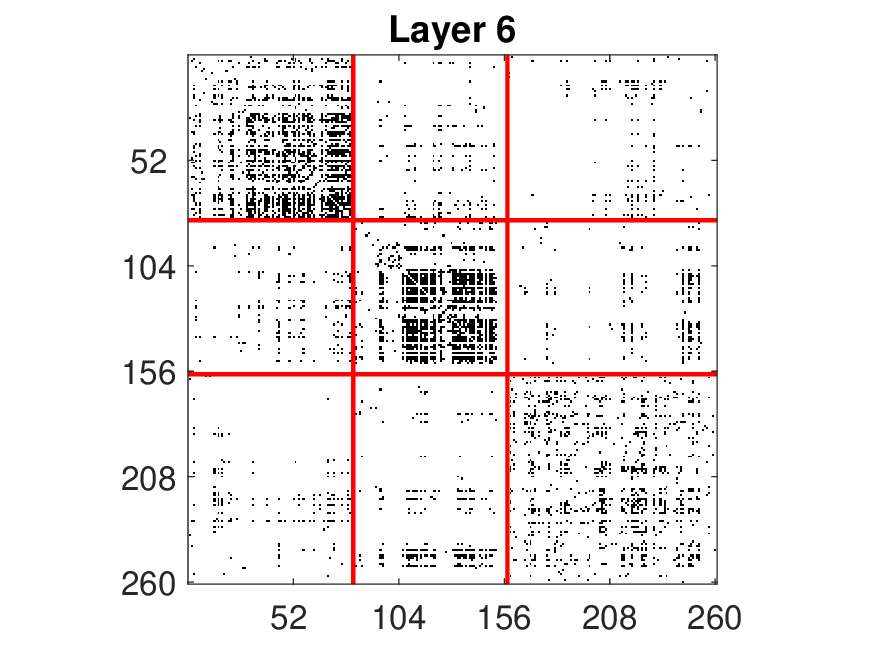}}
\subfigure[]{\includegraphics[width=0.245\textwidth]{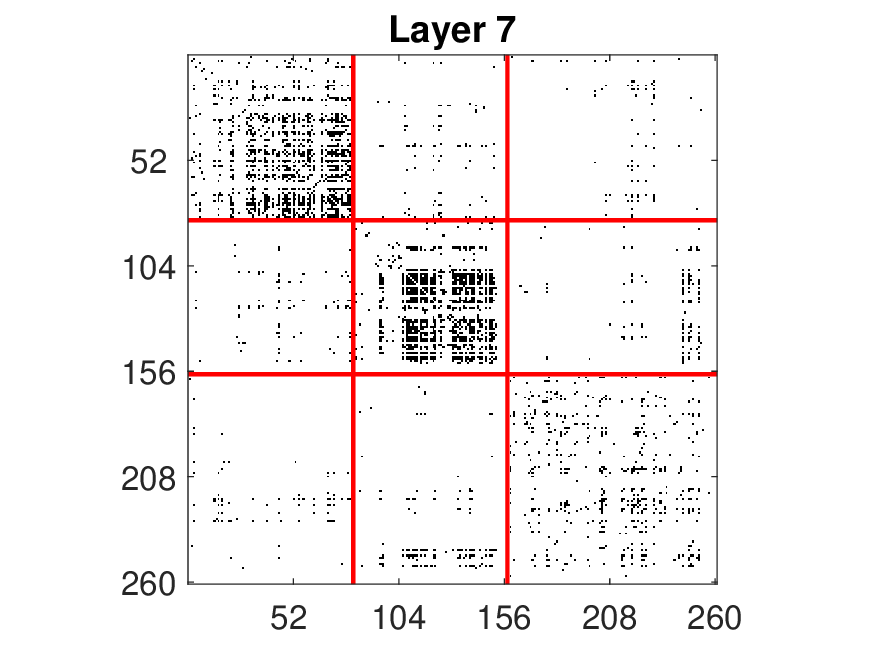}}
\subfigure[]{\includegraphics[width=0.245\textwidth]{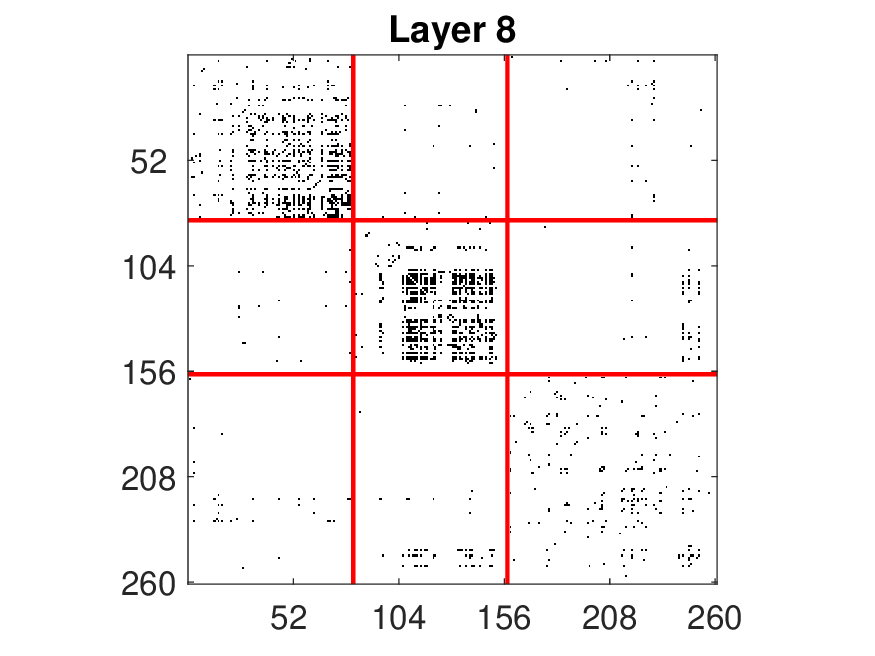}}
\caption{The adjacency matrices of the 8 layers sorted according to the common community labels obtained from NDSoSA with 3 communities for CSMMLN. The colored red lines indicate community partitions. Matrix element values are displayed in gray scale with black corresponding to 1 and white to 0.}
\label{CSMMLNAC} 
\end{figure*}

\begin{figure*}
\centering
\subfigure[]{\includegraphics[width=0.33\textwidth]{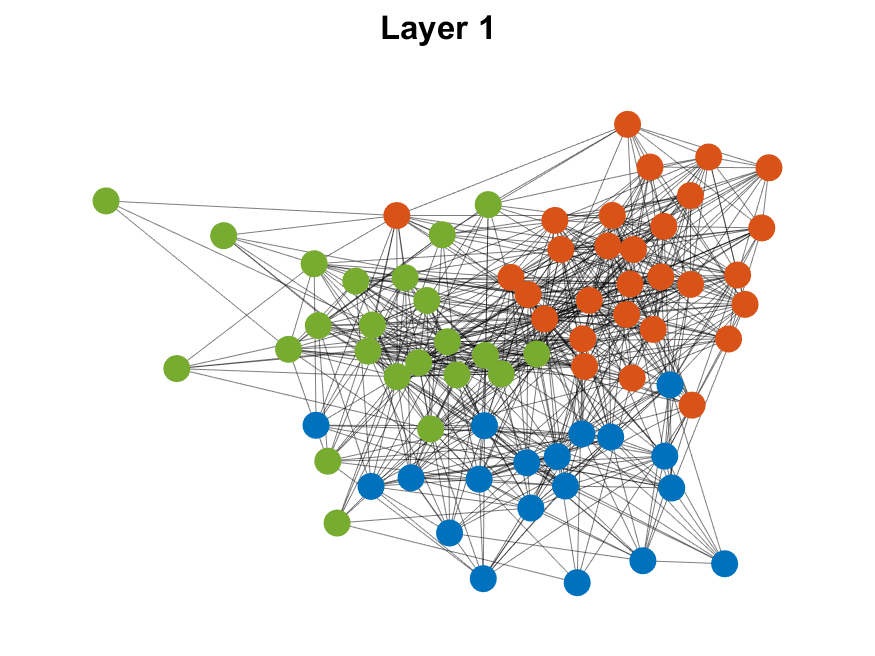}}
\subfigure[]{\includegraphics[width=0.33\textwidth]{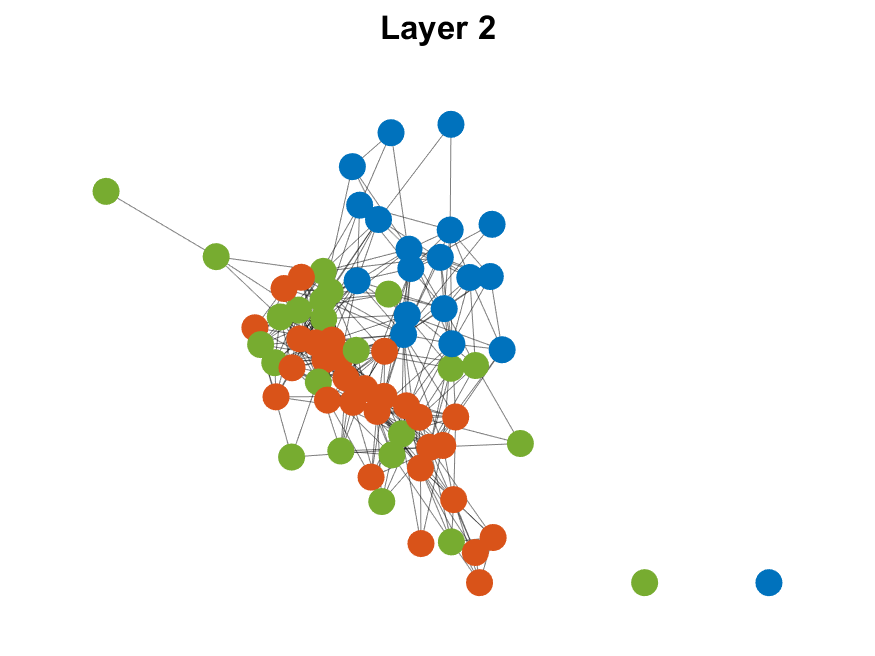}}
\subfigure[]{\includegraphics[width=0.33\textwidth]{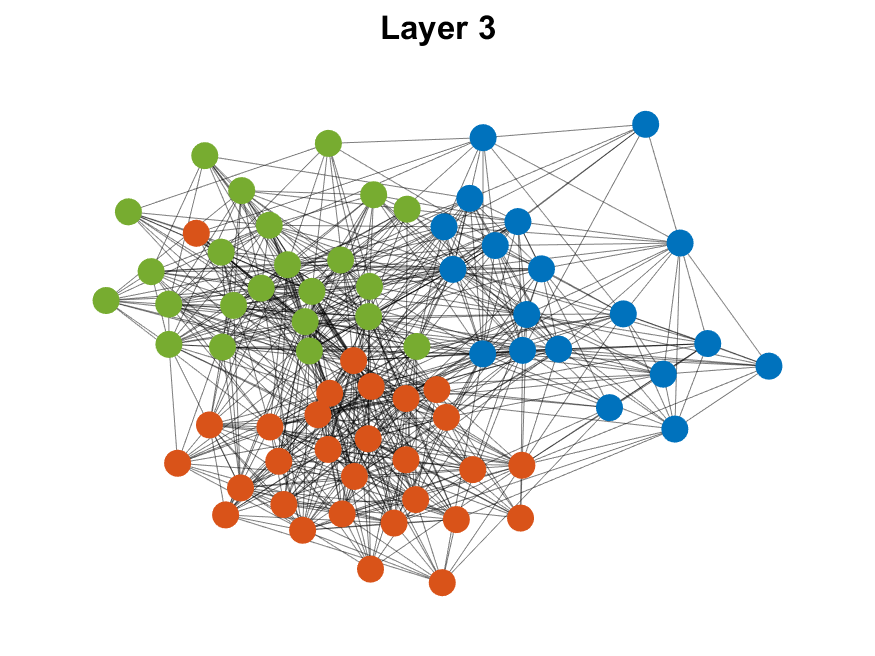}}
\caption{Visualization of the 3 layers of Lazega Law Firm, where nodes are colored by the common community labels obtained from NDSoSA with 3 communities.}
\label{lazegaN} 
\end{figure*}

\begin{figure*}
\centering
\subfigure[]{\includegraphics[width=0.33\textwidth]{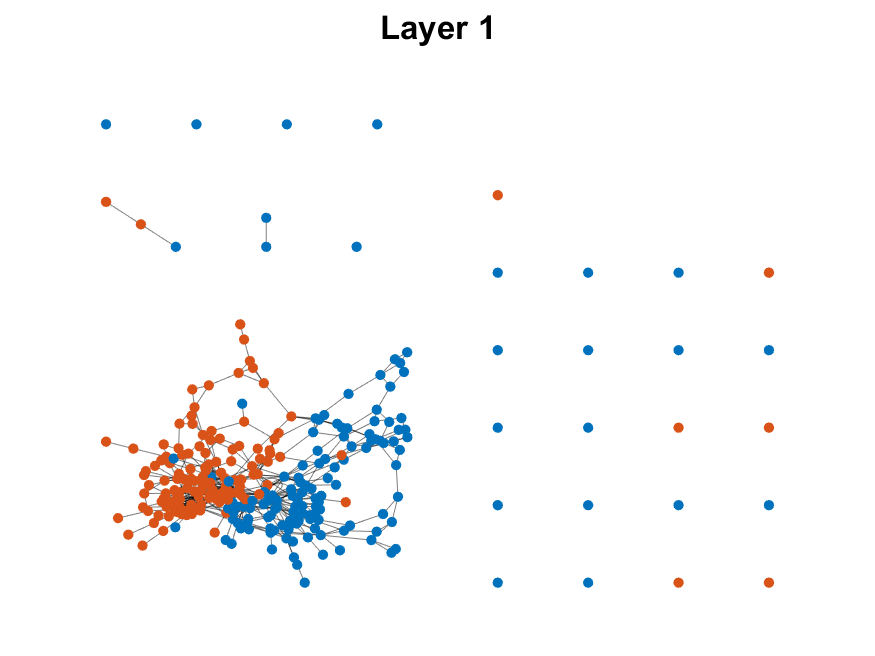}}
\subfigure[]{\includegraphics[width=0.33\textwidth]{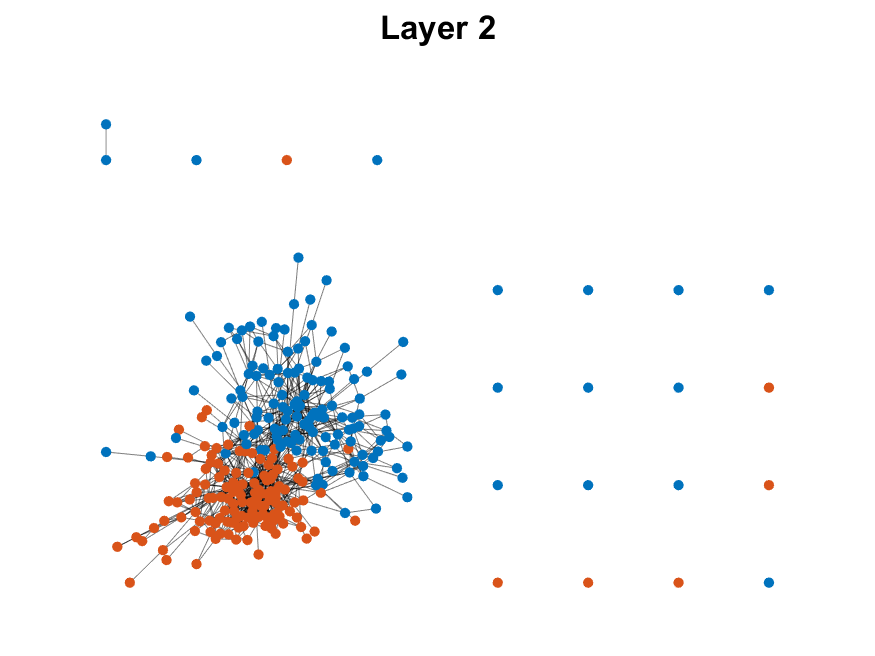}}
\subfigure[]{\includegraphics[width=0.33\textwidth]{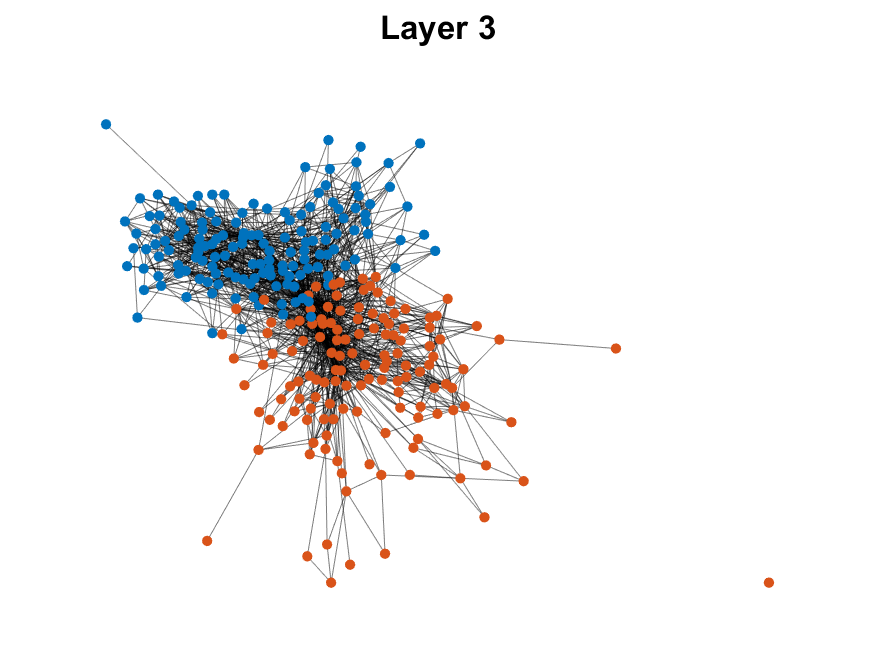}}
\caption{Visualization of the 3 layers of C.Elegans, where nodes are colored by the common community labels obtained from NDSoSA with 2 communities.}
\label{CelegansN} 
\end{figure*}

\begin{figure*}
\centering
\subfigure[]{\includegraphics[width=0.18\textwidth]{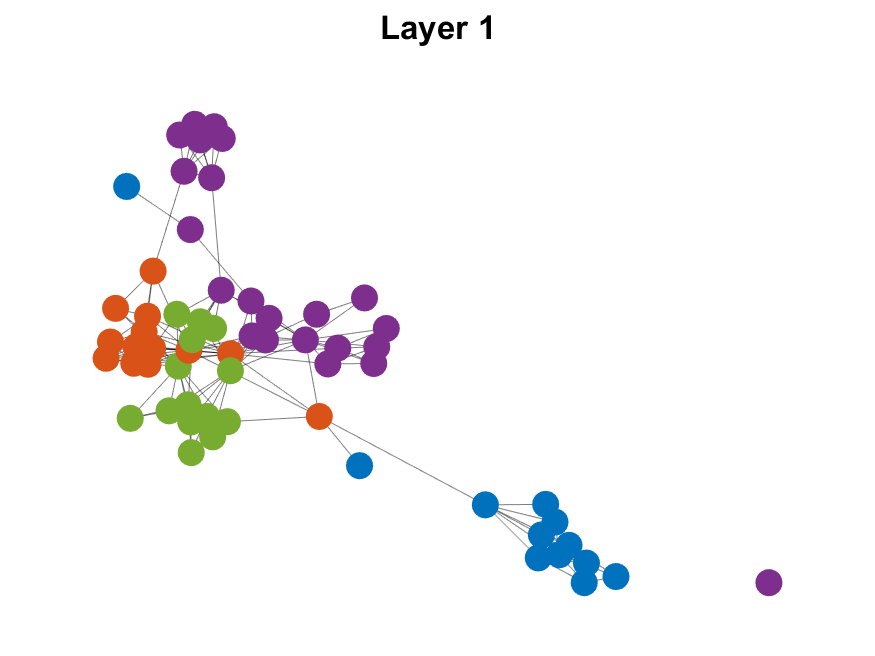}}
\subfigure[]{\includegraphics[width=0.18\textwidth]{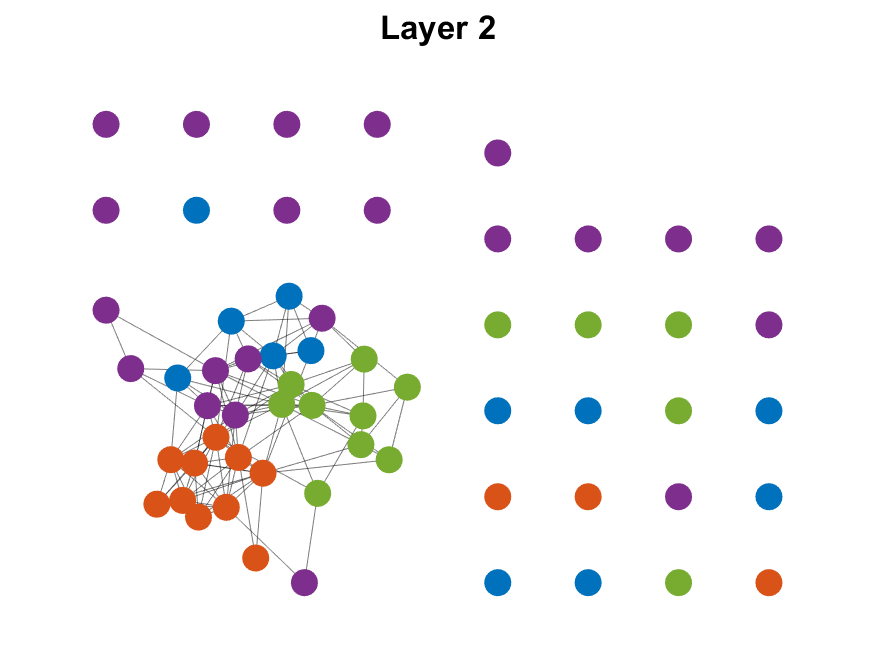}}
\subfigure[]{\includegraphics[width=0.18\textwidth]{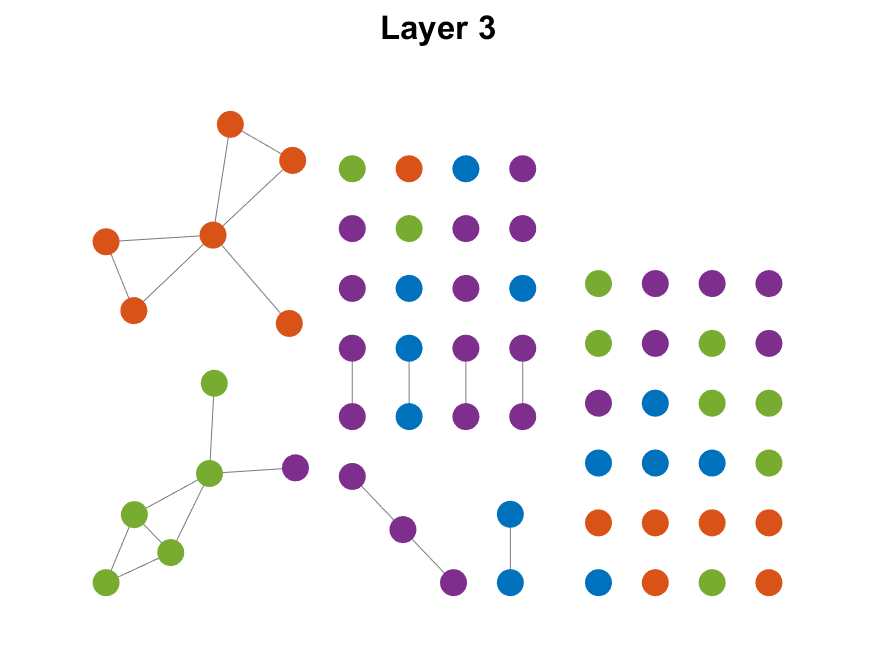}}
\subfigure[]{\includegraphics[width=0.18\textwidth]{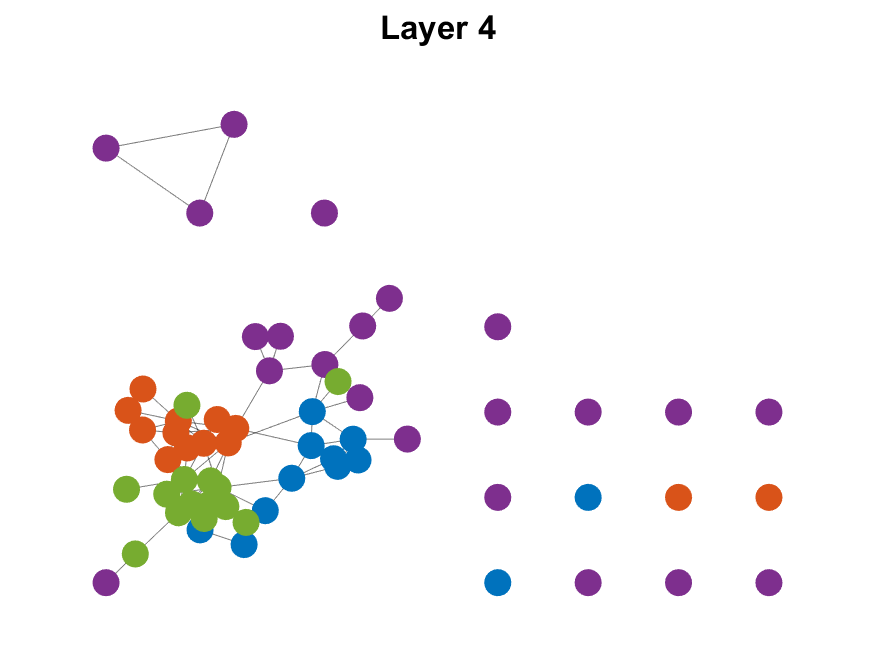}}
\subfigure[]{\includegraphics[width=0.18\textwidth]{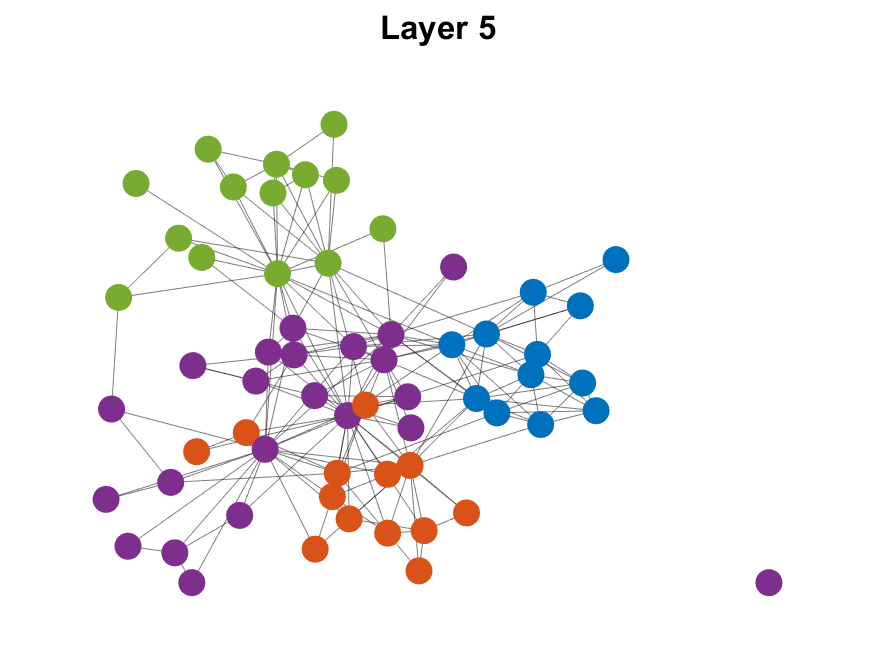}}
\caption{Visualization of the 5 layers of CS-Aarhus, where nodes are colored by the common community labels obtained from NDSoSA with 4 communities.}
\label{CSN} 
\end{figure*}

\begin{figure*}
\centering
\subfigure[]{\includegraphics[width=0.245\textwidth]{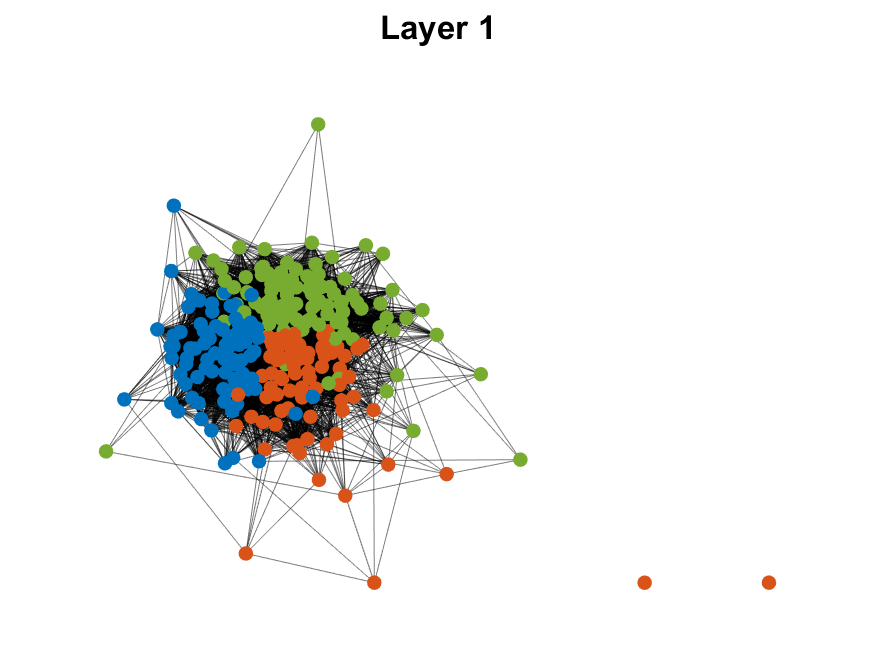}}
\subfigure[]{\includegraphics[width=0.245\textwidth]{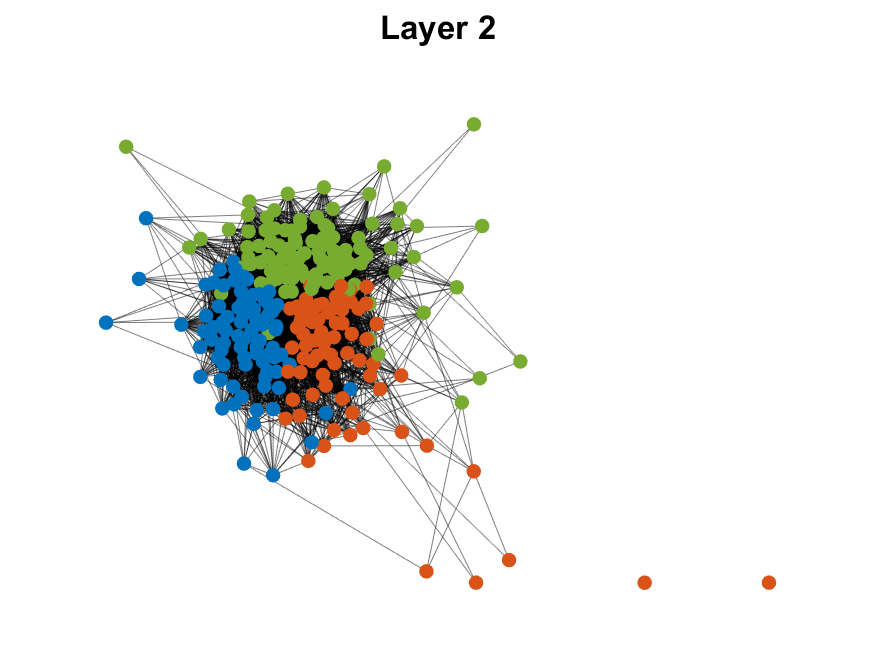}}
\subfigure[]{\includegraphics[width=0.245\textwidth]{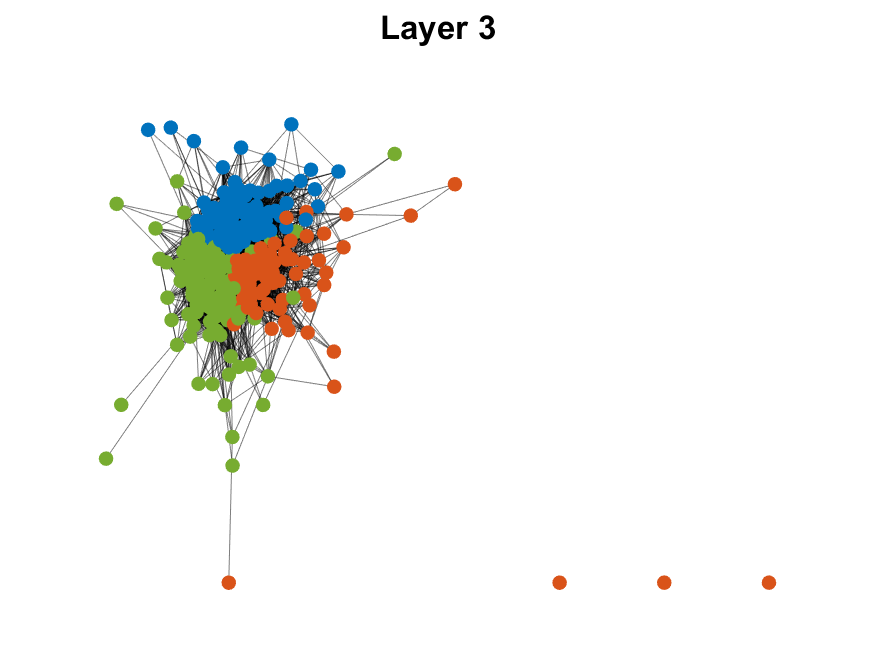}}
\subfigure[]{\includegraphics[width=0.245\textwidth]{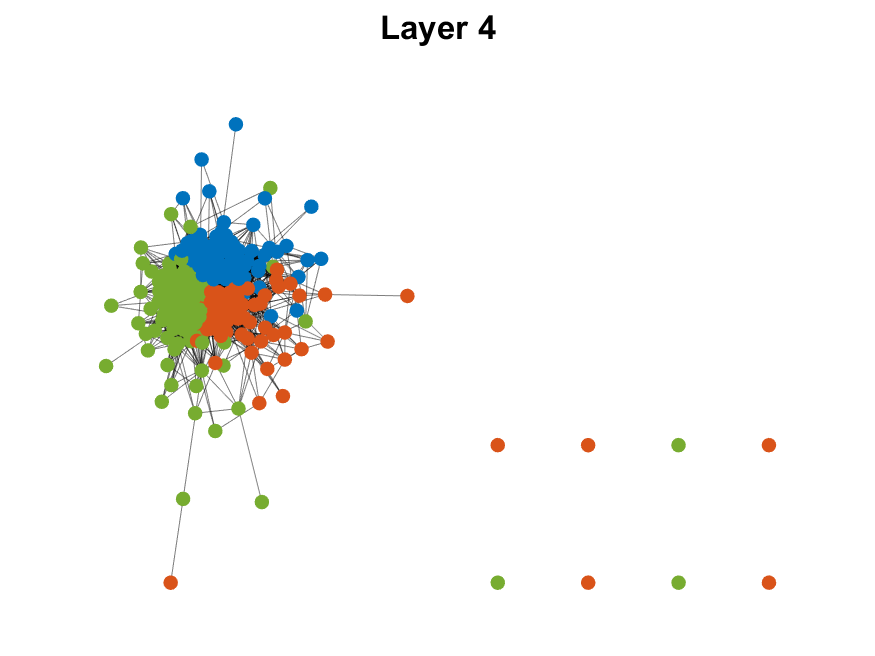}}
\subfigure[]{\includegraphics[width=0.245\textwidth]{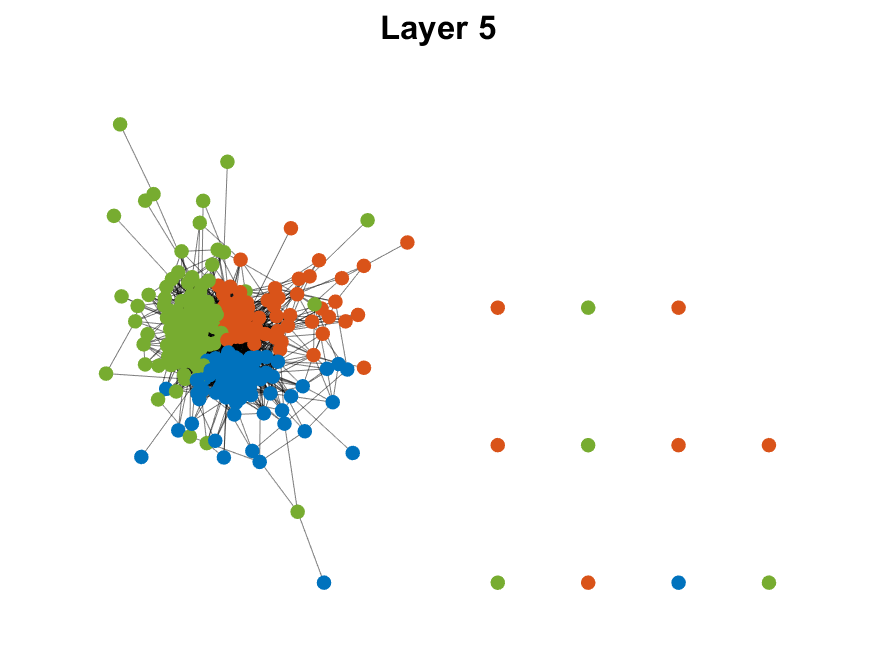}}
\subfigure[]{\includegraphics[width=0.245\textwidth]{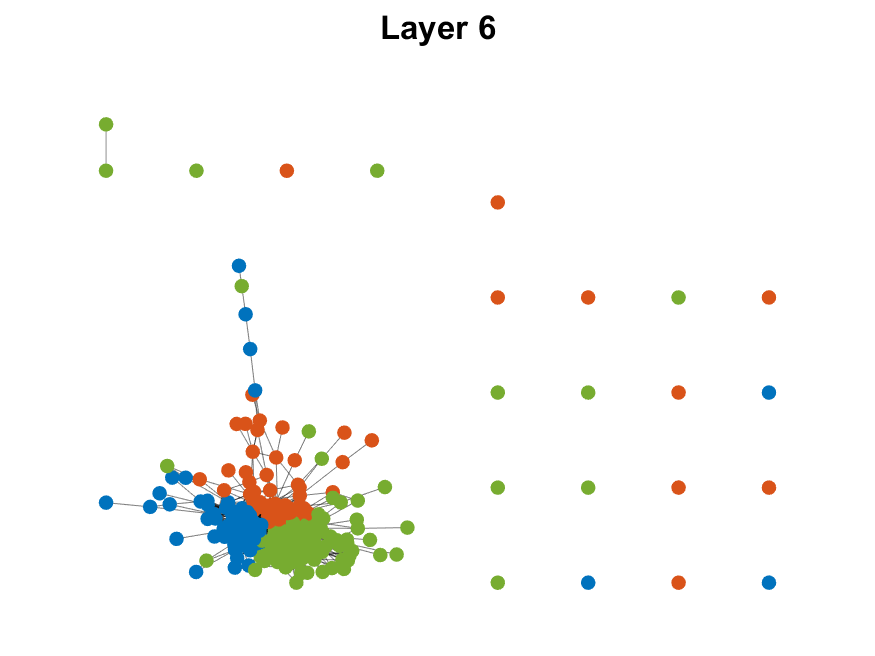}}
\subfigure[]{\includegraphics[width=0.245\textwidth]{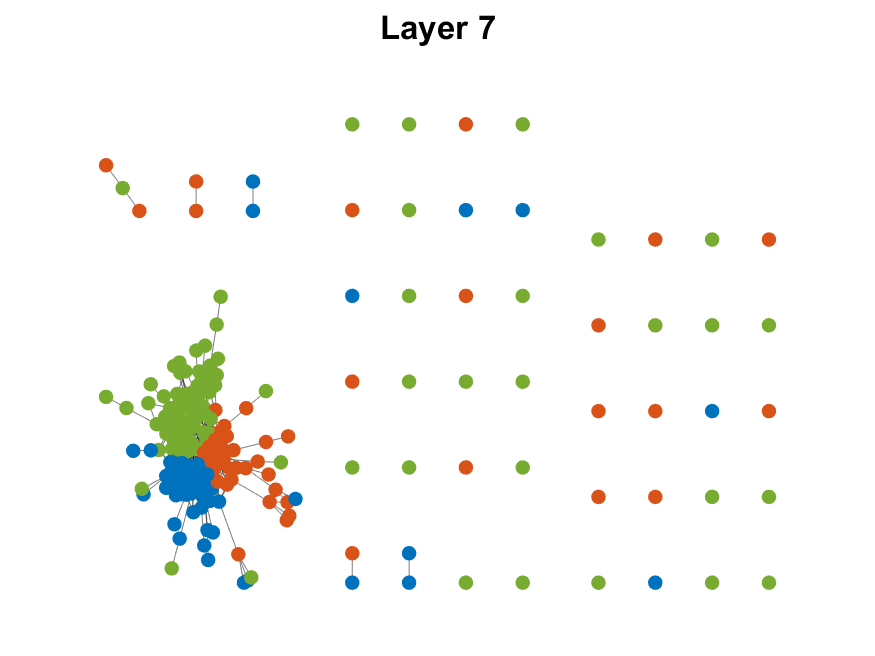}}
\subfigure[]{\includegraphics[width=0.245\textwidth]{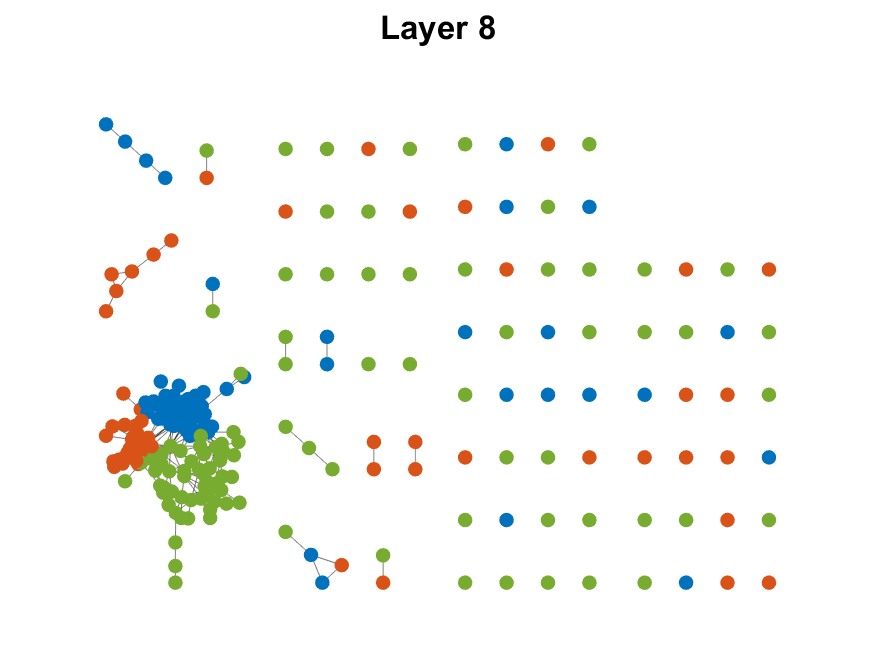}}
\caption{Visualization of the 8 layers of CSMMLN, where nodes are colored by the common community labels obtained from NDSoSA with 3 communities.}
\label{CSMMLNN} 
\end{figure*}

\begin{table*}[h!]
	\centering
	\caption{Communities information in the CSMMLN network}
	\label{CSMMLNIndustryDistribution}
	\begin{tabular}{cccccccccccccc}
\hline\hline
Cluster&Size&Manufacturing Industry&Construction Industry&Finance Industry&Real estate Industry\\
\hline
Community 1&103&51&31&14&7\\
Community 2&76&9&15&47&5\\
Community 3&81&5&19&4&53\\
\hline\hline
\end{tabular}
\end{table*}

\begin{table*}[h!]
	\centering
	\caption{The top 10 stocks of Community 1 ranked by the total degree for the CSMMLN network, where the total degree of node $i$ is defined as $\sum_{l\in[L]}D_{l}(i,i)$ for $i\in[n]$. Note that stock symbols with “sz” represent
Shenzhen Stock Exchange and “sh” represent Shanghai Stock Exchange.}
	\label{CSMMLNCommunity1}
	\begin{tabular}{cccccccccccccc}
\hline\hline
Ranking&Stock&Stock symbol&Total degree&Industry\\
\hline
1&Pubang Landscape Architecture&002663.sz&830&Construction \\
2&Ruida Futures&002961.sz&691&Finance \\
3&Beijing Water Business Doctor&300055.sz&669&Construction \\
4&Zhejiang Orient Financial Holdings Group&600120.sz&669&Finance\\
5&Zhejiang Yasha Decoration&002375.sz&645&Construction \\
6&Shandong Meichen Science and Technology&300237.sz&633&Construction \\
7&Anhui Xinli Finance&600318.sh&618&Finance \\
8&Yongan Futures&600927.sh&601&Finance \\
9&Weiye Construction Group&300621.sz&570&Construction \\
10&Luxin Venture Capital Group&600783.sh&556&Finance \\
\hline\hline
\end{tabular}
\end{table*}

\begin{table*}[h!]
	\centering
	\caption{The top 10 stocks of Community 2 ranked by the total degree for the CSMMLN network.}
	\label{CSMMLNCommunity2}
	\begin{tabular}{cccccccccccccc}
\hline\hline
Ranking&Stock&Stock symbol&Total degree&Industry\\
\hline
1&Shaanxi International Trust&000563.sz&729&Finance\\
2&Shanghai AJ Group&600643.sh&715&Finance\\
3&Suzhou Gold Mantis Construction and Decoration&002081.sz&713&Construction\\
4&Bank of Zhengzhou&002936.sz&680&Finance\\
5&Huaxi Securities&002926.sz&666&Finance\\
6&Sinolink Securities&600109.sh&660&Finance\\
7&Shanxi Securities&002500.sz&633&Finance\\
8&Hubei Biocause Pharmaceutical&000627.sz&617&Finance\\
9&Sealand Securities&000750.sz&611&Finance\\
10&Huaan Securities&600909.sh&606&Finance\\
\hline\hline
\end{tabular}
\end{table*}

\begin{table*}[h!]
	\centering
	\caption{The top 10 stocks of Community 3 ranked by the total degree for the CSMMLN network.}
	\label{CSMMLNCommunity3}
	\begin{tabular}{cccccccccccccc}
\hline\hline
Ranking&Stock&Stock symbol&Total degree&Industry\\
\hline
1&Suning Universal&000718.sz&762&Real estate\\
2&Deluxe Family&600503.sh&741&Real estate\\
3&China Fortune Land Development&600340.sh&726&Real estate\\
4&Hubei Fuxing Science and Technology&000926.sz&704&Real estate\\
5&Wolong Resources Group&600173.sh&699&Real estate\\
6&Chongqing DIMA Industry&600565.sh&668&Real estate\\
7&Hefei Urban Construction Development&002208.sz&655&Real estate\\
8&Citychamp Dartong Advanced Materials&600067.sh&635&Real estate\\
9&Shenzhen Centralcon Investment Holding&000042.sz&632&Real estate\\
10&Beijing Urban Construction Investment and Development&600266.sh&600&Real estate\\
\hline\hline
\end{tabular}
\end{table*}

\begin{figure*}
\centering
\subfigure[]{\includegraphics[width=0.33\textwidth]{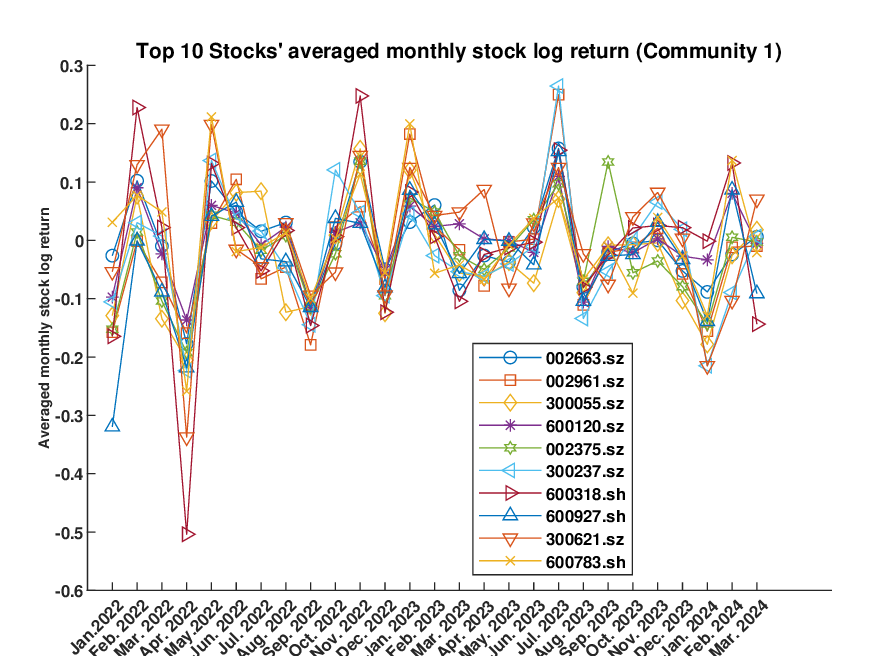}}
\subfigure[]{\includegraphics[width=0.33\textwidth]{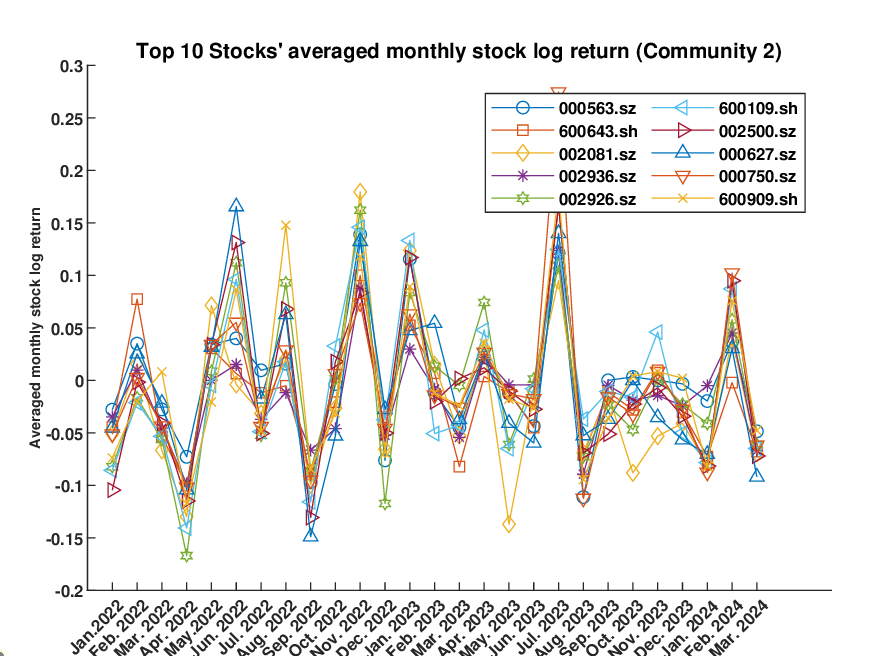}}
\subfigure[]{\includegraphics[width=0.33\textwidth]{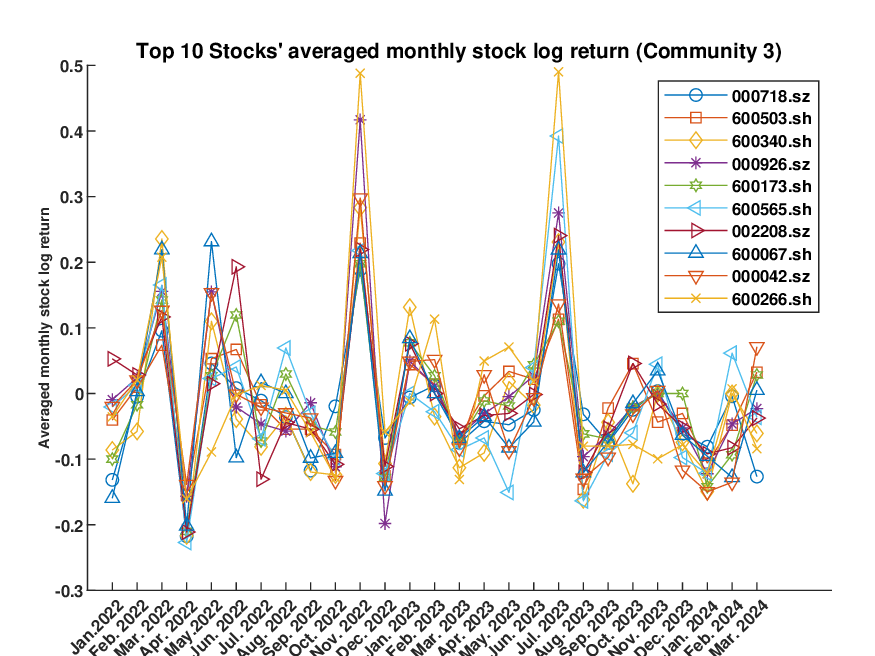}}
\caption{Top 10 Stocks' averaged monthly stock log return for the three communities of the CSMMLN network.}
\label{stockcommunity123} 
\end{figure*}
\section{Conclusion}\label{sec7}
In this paper, we have analyzed the performance of spectral clustering algorithms for community detection in multi-layer networks under the multi-layer degree-corrected stochastic block model. We established consistency results for two algorithms, namely NSoA and NDSoSA, with NSoA relying on the sum of adjacency matrices and NDSoSA utilizing the debiased sum of squared adjacency matrices. Our results highlight the benefits of leveraging multiple layers to enhance the accuracy of community detection, emphasizing the richness of information contained within the multi-layer structure. We further show that NDSoSA generally outperforms NSoA. To detect communities in large-scale multi-layer networks, we present the accelerated versions of NSoA and NDSoSA by randomly selecting nodes and layers. To identity the number of communities in multi-layer networks, we provide a strategy through maximizing the averaged modularity score of a community detection algorithm. Numerical simulations show that NDSoSA outperforms existing methods, underscoring the significance of incorporating the debiased sum of squared adjacency matrices for community detection in multi-layer networks. The numerical results also underscore the high computational efficiency of the accelerated algorithms when applied to large-scale multi-layer networks, as well as the high accuracy of our strategy in determining the number of communities. Moreover, the analysis of several real-world multi-layer networks provides meaningful insights into the applicability of our methods.

There are several potential directions for future work. Firstly,  it would be worthwhile to explore the extension of our methods to other types of multi-layer networks, such as those with directed or weighted edges. \citep{qing2023community,qing2024bipartite} introduced distribution-free models to characterize directed weighted networks, while \citep{su2023spectral} expanded the MLSBM model to directed multi-layer networks and proposed an approach as a directed adaptation of the debiased spectral clustering method proposed by \citep{lei2023bias}. By integrating the distribution-free concept from \citep{qing2023community,qing2024bipartite} with the spectral co-clustering concept from \citep{su2023spectral}, it becomes feasible to extend works presented in this paper to encompass multi-layer directed or weighted networks. Secondly, our method for estimating $K$ in this paper maximizes averaged modularity but lacks theoretical guarantees. Developing methods with the estimation consistency for the true number of communities for multi-layer networks under MLSBM and MLDCSBM is challenging and appealing. For single-layer networks under SBM and DCSBM, methods with theoretical guarantees do exist, such as the pseudo-likelihood ratio statistic combined with spectral clustering \citep{ma2021determining} and the stepwise goodness-of-fit (StGoF) combined with spectral clustering \citep{jin2023optimal}. Supposing that the true number of communities is $K_{0}$ and denoting $\hat{Z}_{K}$ as the $n\times K$ estimated community membership matrix obtained by applying NDSoSA on $\{A_{l}\}^{L}_{l=1}$ with $K$ communities, we outline potential directions to extend the works of \citep{ma2021determining,jin2023optimal} to multi-layer networks:
\begin{itemize}
  \item Building on \citep{ma2021determining}, we can obtain the $n\times (K+1)$ estimated community membership matrix $\hat{Z}^{b}_{K+1}$ by applying the binary segmentation technique introduced in \citep{ma2021determining} to the row-normalized version of the leading $(K+1)$ eigenvectors of $S_{\mathrm{sum}}$. Subsequently, we compute the pseudo-likelihood ratio for each layer using Equation (2) in \citep{ma2021determining}. By summing these ratios across all layers, we can derive the $R(K)$ metric defined in Equation (2) of \citep{ma2021determining}. The $K$ that minimizes $R(K)$ is expected to be a good estimate of the true number of communities.
  \item Alternatively, based on $\hat{Z}_{K}$, we can calculate the StGoF statistic defined by Equation (2.8) in \citep{jin2023optimal} for each layer. Denote the StGoF statistic for the $l$-th layer as $\psi^{(K)}_{l}$ for $l\in[L]$. As shown in \citep{jin2023optimal}, $\psi^{(K)}_{l}$ converges to the standard normal distribution $N(0,1)$ for $K=K_{0}$ and diverges to $\infty$ for $K<K_{0}$ under DCSBM for single-layer networks. Naturally, if we define a new statistic as $\psi^{(K)}\equiv\sum_{l\in[L]}(\psi^{(K)}_{l})^{2}$, $\psi^{(K)}$ should follow $\chi^{2}_{L}$ (Chi-Square distribution with $L$ degrees of freedom) when $K=K_{0}$ and diverges to $\infty$ for $K<K_{0}$. Thus, following step (c) of the StGoF algorithm in \citep{jin2023optimal}, the $K$ that satisfies $\psi^{(K)}<\chi^{2}_{L}(\alpha)$ should be a good estimate of the true number of communities, where $\chi^{2}_{L}(\alpha)$ denotes the $\alpha$ upper-quantile of $\chi^{2}_{L}$ and $\alpha$ is a predetermined value, such as $1\%$ or $5\%$.
\end{itemize}
Although we have provided possible directions for extending the works in \citep{ma2021determining,jin2023optimal} to estimate the number of communities from single-layer networks to multi-layer networks, the details of these extensions require careful consideration, and establishing their theoretical guarantees is challenging. Thirdly, developing theoretical lower bounds for the subsample sizes $n_{\mathrm{sample}}$ and $L_{\mathrm{sample}}$, as well as building the estimation consistencies of SNSoA and SNDSoSA,is an appealing research direction. Fourthly, in this paper, we remove the effect of the degree heterogeneity parameter $\theta$ by normalizing each rows of $\hat{U}$ (and $\hat{V}$). An alternative approach could involve designing algorithms based on the entry-wise ratios using $\hat{U}$ (and $\hat{V}$) to detect communities under MLDCSBM, where the concept of entry-wise ratios was first introduced in \citep{SCORE} to design an efficient algorithm for community estimation under DCSBM for single-layer networks. Lastly, conducting empirical studies on a wider range of real-world multi-layer networks to validate the effectiveness and scalability of the proposed algorithms would be crucial for practical applications. For instance, our methods can be utilized to detect areas of activity within a multi-layer network that describes the movement of various entities in a city \citep{yildirimoglu2018identification}. For further exploration of application areas, please refer to \citep{fortunato2010community,papadopoulos2012community,boccaletti2014structure,gasparetti2021community,huang2021survey}.
\section*{CRediT authorship contribution statement}
\textbf{Huan Qing:} Conceptualization; Data curation; Formal analysis; Funding acquisition; Methodology; Project administration; Resources; Software; Validation; Visualization; Writing-original draft; Writing-review $\&$ editing.
\section*{Declaration of competing interest}
The author declares no competing interests.
\section*{Data availability}
Data and code will be made available on request.
\section*{Acknowledgements}
H.Q. was supported by the Scientific Research Foundation of Chongqing University of Technology (Grant No. 2024ZDR003), and the Science and Technology Research Program of Chongqing Municipal Education Commission (Grant No. KJQN202401168).
\appendix
\section{Proofs}\label{SecProofs}
\subsection{Proof of Lemma \ref{EigOsum}}
\begin{proof}
By Lemma 3 of \citep{qing2023community}, this lemma holds.
\end{proof}
\subsection{Proof of Lemma \ref{EigOsum2}}
\begin{proof}
Set $P=\sum_{l\in[L]}B_{l}Z'\Theta^2ZB_{l}$. We have $\tilde{S}_{\mathrm{sum}}=\Theta ZPZ'\Theta$. Let $\Delta$ be a $K\times K$ diagonal matrix such that $\Delta(k,k)=\frac{\|\Theta Z(:,k)\|_{F}}{\|\theta\|_{F}}$ for $k\in[K]$. Let $\Gamma$ be an $n\times K$ matrix such that $\Gamma(:,k)=\frac{\Theta Z(:,k)}{\|\Theta Z(:,k)\|_{F}}$ for $k\in[K]$. We have $\Gamma'\Gamma=I_{K\times K}$ and $\tilde{S}_{\mathrm{sum}}=\|\theta\|^{2}_{F}\Gamma \Delta P\Delta\Gamma'$. Since $(\Theta Z)'(\Theta Z)$ is a positive definite diagonal matrix, the rank of $\sum_{l\in[L]}B^{2}_{l}$ is equal to that of $\sum_{l\in[L]}B_{l}Z'\Theta^2ZB_{l}$. By Lemma 3 of \citep{qing2023community}, we know that Lemma \ref{EigOsum2} holds as long as $\mathrm{rank}(P)=K$. Hence, when Lemma \ref{EigOsum2} holds, the rank of $\sum_{l\in[L]}B^{2}_{l}$ should be $K$.
\end{proof}
\subsection{Proof of Lemma \ref{boundAsum}}
\begin{proof}
We use the Matrix Bernstein theorem below to bound $\|A_{\mathrm{sum}}-\Omega_{\mathrm{sum}}\|$.
\begin{thm}\label{Bern}
(Theorem 1.4 of \citep{tropp2012user}) Consider a finite sequence $\{X_{i}\}$ of independent, random, self-adjoint matrices with dimension $n$. Supposes that
\begin{align*}
\mathbb{E}[X_{i}]=0, \mathrm{and~}\|X_{i}\|\leq R~\mathrm{almost~surely}.
\end{align*}
Then, for all $t\geq 0$,
\begin{align*}
\mathbb{P}(\|\sum_{i}X_{i}\|\geq t)\leq n\cdot \mathrm{exp}(\frac{-t^{2}/2}{\sigma^{2}+Rt/3}),
\end{align*}
where $\sigma^{2}:=\|\sum_{i}\mathbb{E}[X^{2}_{i}]\|$.
\end{thm}

Let $E^{(ij)}$ be a (deterministic) matrix with 1 in the $(i,j)$th position and 0 everywhere else for $i\in[n],j\in[n]$. We have $A_{\mathrm{sum}}-\Omega_{\mathrm{sum}}=\sum_{l\in[L]}(A_{l}-\Omega_{l})=\sum_{l\in[L]}\sum_{i\in[n]}\sum_{j\in[n]}(A_{l}(i,j)-\Omega_{l}(i,j))E^{(ij)}$. Let $X^{(ijl)}=(A_{l}(i,j)-\Omega_{l}(i,j))E^{(ij)}$ for all $l\in[L], i\in[n], j\in[n]$.We have the following three conclusions:
\begin{itemize}
  \item $\mathbb{E}[X^{(ijl)}]=0$.
  \item $\|X^{(ijl)}\|=|A_{l}(i,j)-\Omega_{l}(i,j)|\leq1$, i.e., $R=1$.
  \item For $\sigma^{2}$, under $\mathrm{MLDCSBM}(Z,\Theta,\mathbb{B})$, we have
      \begin{align*}
      \sigma^{2}&=\|\sum_{l\in[L]}\sum_{i\in[n]}\sum_{j\in[n]}\mathbb{E}[(X^{(ijl)})^{2}]\|\\
      &=\|\sum_{l\in[L]}\sum_{i\in[n]}\sum_{j\in[n]}E^{(ii)}\mathbb{E}[(A_{l}(i,j)-\Omega_{l}(i,j))^{2}]\|\\
      &=\|\sum_{l\in[L]}\sum_{i\in[n]}\sum_{j\in[n]}E^{(ii)}\Omega_{l}(i,j)(1-\Omega_{l}(i,j))\|\\
      &=\mathrm{max}_{i\in[n]}\sum_{l\in[L]}\sum_{j\in[n]}\Omega_{l}(i,j)(1-\Omega_{l}(i,j))\\
      &\leq\mathrm{max}_{i\in[n]}\sum_{l\in[L]}\sum_{j\in[n]}\Omega_{l}(i,j)\\
      &=\mathrm{max}_{i\in[n]}\sum_{l\in[L]}\sum_{j\in[n]}\theta(i)\theta(j)B_{l}(\ell(i),\ell(j))\\
      &\leq\theta_{\mathrm{max}}\sum_{l\in[L]}\sum_{j\in[n]}\theta(j)=\theta_{\mathrm{max}}\|\theta\|_{1}L.
      \end{align*}
\end{itemize}
By Theorem \ref{Bern}, for any $t\geq0$,we have
\begin{align*}
\mathbb{P}(\|A_{\mathrm{sum}}-\Omega_{\mathrm{sum}}\|\geq t)\leq n\cdot \mathrm{exp}(\frac{-t^{2}/2}{\theta_{\mathrm{max}}\|\theta\|_{1}L+t/3})
\end{align*}
Set $t=\frac{\alpha+1+\sqrt{(\alpha+1)(\alpha+19)}}{3}\sqrt{\theta_{\mathrm{max}}\|\theta\|_{1}L\mathrm{log}(n+L)}$ for any $\alpha\geq0$, we have
\begin{align*}
&\mathbb{P}(\|A_{\mathrm{sum}}-\Omega_{\mathrm{sum}}\|\geq t)\leq n\cdot \mathrm{exp}(\frac{-t^{2}/2}{\theta_{\mathrm{max}}\|\theta\|_{1}L+t/3})\\
&=n\cdot\mathrm{exp}(\frac{-(\alpha+1)\mathrm{log}(n+L)}{\frac{18}{(\sqrt{\alpha+1}+\sqrt{\alpha+19})^{2}}+\frac{2\sqrt{\alpha+1}}{\sqrt{\alpha+1}+\sqrt{\alpha+19}}\sqrt{\frac{\mathrm{log}(n+L)}{\theta_{\mathrm{max}}\|\theta\|_{1}L}}})\\
&\leq n\cdot\mathrm{exp}(-(\alpha+1)\mathrm{log}(n+L))=\frac{n}{(n+L)^{\alpha+1}}\leq\frac{1}{(n+L)^{\alpha}},
\end{align*}
where we have used Assumption \ref{Assum1} and the fact that $\frac{18}{(\sqrt{\alpha+1}+\sqrt{\alpha+19})^{2}}+\frac{2\sqrt{\alpha+1}}{\sqrt{\alpha+1}+\sqrt{\alpha+19}}\sqrt{\frac{\mathrm{log}(n+L)}{\theta_{\mathrm{max}}\|\theta\|_{1}L}}\leq\frac{18}{(\sqrt{\alpha+1}+\sqrt{\alpha+19})^{2}}+\frac{2\sqrt{\alpha+1}}{\sqrt{\alpha+1}+\sqrt{\alpha+19}}=1$. Set $\alpha=1$, this lemma holds.
\end{proof}
\subsection{Proof of Lemma \ref{boundSsum}}
\begin{proof}
Recall that $S_{\mathrm{sum}}=\sum_{l\in[L]}(A^{2}_{l}-D_{l})$ and $\tilde{S}_{\mathrm{sum}}=\sum_{l\in[L]}\Omega^{2}_{l}$, we see that $\tilde{S}_{\mathrm{sum}}$ is not the expectation of $S_{\mathrm{sum}}$ under MLDCSBM. Therefore, we cannot apply the Matrix Bernstein inequality in Theorem \ref{Bern} directly to bound $\|S_{\mathrm{sum}}-\tilde{S}_{\mathrm{sum}}\|$. Instead, we consider the following strategy to decompose $S_{\mathrm{sum}}-\tilde{S}_{\mathrm{sum}}$ into two parts:
\begin{align*}
&S_{\mathrm{sum}}-\tilde{S}_{\mathrm{sum}}=\sum_{l\in[L]}(A^{2}_{l}-D_{l}-\Omega^{2}_{l})=\sum_{l\in[L]}\sum_{i\in[n]}\sum_{j\in[n]}\sum_{m\in[n]}(A_{l}(i,m)A_{l}(m,j)\\
&~~~-\Omega_{l}(i,m)\Omega_{l}(m,j))E^{(ij)}-\sum_{l\in[L]}D_{l}\\
&=\sum_{l\in[L]}\sum_{i\in[n],j\in[n],i\neq j}\sum_{m\in[n]}(A_{l}(i,m)A_{l}(m,j)-\Omega_{l}(i,m)\Omega_{l}(m,j))E^{(ij)}\\
&~~~+\sum_{l\in[L]}\sum_{m\in[n]}\sum_{i\in[n]}(A^{2}_{l}(i,m)-\Omega^{2}_{l}(i,m))E^{(ii)}-\sum_{l\in[L]}D_{l}\\
&=\sum_{l\in[L]}\sum_{i\in[n],j\in[n],i\neq j}\sum_{m\in[n]}(A_{l}(i,m)A_{l}(m,j)-\Omega_{l}(i,m)\Omega_{l}(m,j))E^{(ij)}\\
&~~~+\sum_{l\in[L]}\sum_{m\in[n]}\sum_{i\in[n]}(A_{l}(i,m)-\Omega^{2}_{l}(i,m))E^{(ii)}-\sum_{l\in[L]}D_{l}\\
&=\sum_{l\in[L]}\sum_{i\in[n],j\in[n],i\neq j}\sum_{m\in[n]}(A_{l}(i,m)A_{l}(m,j)-\Omega_{l}(i,m)\Omega_{l}(m,j))E^{(ij)}\\
&~~~-\sum_{l\in[L]}\sum_{m\in[n]}\sum_{i\in[n]}\Omega^{2}_{l}(i,m)E^{(ii)}+\sum_{l\in[L]}\sum_{i\in[n]}\sum_{m\in[n]}A_{l}(i,m)E^{(ii)}-\sum_{l\in[L]}D_{l}\\
&=\sum_{l\in[L]}\sum_{i\in[n],j\in[n],i\neq j}\sum_{m\in[n]}(A_{l}(i,m)A_{l}(m,j)-\Omega_{l}(i,m)\Omega_{l}(m,j))E^{(ij)}\\
&~~~-\sum_{l\in[L]}\sum_{m\in[n]}\sum_{i\in[n]}\Omega^{2}_{l}(i,m)E^{(ii)}+\sum_{l\in[L]}D_{l}-\sum_{l\in[L]}D_{l}\\
&=\underbrace{\sum_{l\in[L]}\sum_{i\in[n],j\in[n],i\neq j}\sum_{m\in[n]}(A_{l}(i,m)A_{l}(m,j)-\Omega_{l}(i,m)\Omega_{l}(m,j))E^{(ij)}}_{S1}\\
&~~~-\underbrace{\sum_{l\in[L]}\sum_{m\in[n]}\sum_{i\in[n]}\Omega^{2}_{l}(i,m)E^{(ii)}}_{S2}.
\end{align*}
Therefore, we have $S_{\mathrm{sum}}-\tilde{S}_{\mathrm{sum}}=S1-S2$, where $S1=\sum_{l\in[L]}\sum_{i\in[n],j\in[n],i\neq j}\sum_{m\in[n]}(A_{l}(i,m)A_{l}(m,j)-\Omega_{l}(i,m)\Omega_{l}(m,j))E^{(ij)}$ and $S2=\sum_{l\in[L]}\sum_{m\in[n]}\sum_{i\in[n]}\Omega^{2}_{l}(i,m)E^{(ii)}$. In $S1$, since $i\neq j$, $A_{l}(i,m)$ and $A_{l}(m,j)$ are independent. Hence $\mathbb{E}(S1)=0$, and we can apply the Matrix Bernstein inequality to $S{1}$ to obtain the upper bound of $\|S1\|$. Since $\Omega_{l}$ is non-random, $S2$ is deterministic, and we can bound $\|S2\|$ directly. The above analysis suggests that
\begin{align*}
&\|S_{\mathrm{sum}}-\tilde{S}_{\mathrm{sum}}\|=\|S1-S2\|\leq\|S1\|+\|S2\|\\
&=\|S1\|+\|\sum_{l\in[L]}\sum_{m\in[n]}\sum_{i\in[n]}\Omega^{2}_{l}(i,m)E^{(ii)}\|\\
&=\|S1\|+\mathrm{max}_{i\in[n]}\sum_{l\in[L]}\sum_{m\in[n]}\Omega^{2}_{l}(i,m)\\
&=\|S1\|+\mathrm{max}_{i\in[n]}\sum_{l\in[L]}\sum_{m\in[n]}\theta^{2}(i)\theta^{2}(m)B^{2}_{l}(\ell(i),\ell(m))\\
&\leq\|S1\|+\mathrm{max}_{i\in[n]}\sum_{l\in[L]}\sum_{m\in[n]}\theta^{2}(i)\theta^{2}(m)\\
&\leq\|S1\|+\theta^{2}_{\mathrm{max}}\sum_{l\in[L]}\sum_{m\in[n]}\theta^{2}(m)\\
&=\|S1\|+\theta^{2}_{\mathrm{max}}\|\theta\|^{2}_{F}L.
\end{align*}
Next, we apply Theorem \ref{Bern} to bound $\|S1\|$. Set $X^{(ijml)}=(A_{l}(i,m)A_{l}(m,j)-\Omega_{l}(i,m)\Omega_{l}(m,j))E^{(ij)}$ for all $i\in[n],j\in[n],m\in[n],l\in[L]$. We have the following three conclusions:
\begin{itemize}
  \item $\mathbb{E}[X^{(ijml)}]=0$.
  \item $\|X^{(ijml)}\|=|A_{l}(i,m)A_{l}(m,j)-\Omega_{l}(i,m)\Omega_{l}(m,j)|\leq1$. Set $R_{s}=1$.
  \item Set $\sigma^{2}_{s}=\|\sum_{l\in[L]}\sum_{i\in[n],j\in[n],i\neq j}\sum_{m\in[n]}\mathbb{E}[(X^{(ijml)})^{2}]\|$. Under MLDCSBM, we have
      \begin{align*}
      &\sigma^{2}_{s}=\|\sum_{l\in[L]}\sum_{i\in[n],j\in[n],i\neq j}\sum_{m\in[n]}\mathbb{E}[(A_{l}(i,m)A_{l}(m,j)\\
      &~~~-\Omega_{l}(i,m)\Omega_{l}(m,j))^{2}E^{(ii)}]\|\\
      &=\|\sum_{l\in[L]}\sum_{i\in[n],j\in[n],i\neq j}\sum_{m\in[n]}E^{(ii)}\mathbb{E}[A^{2}_{l}(i,m)A^{2}_{l}(m,j)\\
      &~~~-2A_{l}(i,m)A_{l}(m,j)\Omega_{l}(i,m)\Omega_{l}(m,j)\\
      &~~~+\Omega^{2}_{l}(i,m)\Omega^{2}_{l}(m,j)]\|\\
      &=\|\sum_{l\in[L]}\sum_{i\in[n],j\in[n],i\neq j}\sum_{m\in[n]}E^{(ii)}(\mathbb{E}[A^{2}_{l}(i,m)A^{2}_{l}(m,j)]\\
      &~~~-\Omega^{2}_{l}(i,m)\Omega^{2}_{l}(m,j))\|\\
      &=\|\sum_{l\in[L]}\sum_{i\in[n],j\in[n],i\neq j}\sum_{m\in[n]}E^{(ii)}(\mathbb{E}[A^{2}_{l}(i,m)]\mathbb{E}[A^{2}_{l}(m,j)]\\
      &~~~-\Omega^{2}_{l}(i,m)\Omega^{2}_{l}(m,j))\|\\
      &=\|\sum_{l\in[L]}\sum_{i\in[n],j\in[n],i\neq j}\sum_{m\in[n]}E^{(ii)}((\Omega_{l}(i,m)(1-\Omega_{l}(i,m))\\
      &~~~+\Omega^{2}_{l}(i,m))(\Omega_{l}(m,j)(1-\Omega_{l}(m,j))\\
      &~~~+\Omega^{2}_{l}(m,j))-\Omega^{2}_{l}(i,m)\Omega^{2}_{l}(m,j))\|\\
      &=\|\sum_{l\in[L]}\sum_{i\in[n],j\in[n],i\neq j}\sum_{m\in[n]}E^{(ii)}\Omega_{l}(i,m)\Omega_{l}(m,j)\\
      &~~~(1-\Omega_{l}(i,m)\Omega_{l}(m,j))\|\\
      &=\mathrm{max}_{i\in[n]}\sum_{l\in[L]}\sum_{j\in[n],j\neq i}\sum_{m\in[n]}\Omega_{l}(i,m)\Omega_{l}(m,j)\\
      &~~~(1-\Omega_{l}(i,m)\Omega_{l}(m,j))\\
      &\leq\mathrm{max}_{i\in[n]}\sum_{l\in[L]}\sum_{j\in[n],j\neq i}\sum_{m\in[n]}\Omega_{l}(i,m)\Omega_{l}(m,j)\\
      &=\mathrm{max}_{i\in[n]}\sum_{l\in[L]}\sum_{j\in[n],j\neq i}\sum_{m\in[n]}\theta(i)\theta^{2}(m)\theta(j)B_{l}(\ell(i),\ell(m))B_{l}(\ell(m),\ell(j))\\
      &\leq\mathrm{max}_{i\in[n]}\sum_{l\in[L]}\sum_{j\in[n],j\neq i}\sum_{m\in[n]}\theta(i)\theta^{2}(m)\theta(j)\\
      &\leq\theta_{\mathrm{max}}\sum_{l\in[L]}\sum_{j\in[n]}\sum_{m\in[n]}\theta^{2}(m)\theta(j)\\
      &=\theta_{\mathrm{max}}\|\theta\|_{1}\|\theta\|^{2}_{F}L.
      \end{align*}
\end{itemize}
For any $t_{s}\geq0$, by Theorem \ref{Bern}, we have
\begin{align*}
\mathbb{P}(\|S1\|\geq t_{s})\leq n\cdot\mathrm{exp}(\frac{-t^{2}_{s}/2}{\sigma^{2}_{s}+R_{s}t_{s}/3})\leq n\cdot\mathrm{exp}(\frac{-t^{2}_{s}/2}{\theta_{\mathrm{max}}\|\theta\|_{1}\|\theta\|^{2}_{F}L+t_{s}/3})
\end{align*}
Set $t_{s}=\frac{\alpha+1+\sqrt{(\alpha+1)(\alpha+19)}}{3}\sqrt{\theta_{\mathrm{max}}\|\theta\|_{1}\|\theta\|^{2}_{F}L\mathrm{log}(n+L)}$ for any $\alpha\geq0$, by Assumption \ref{Assum2}, we have
\begin{align*}
\mathbb{P}(\|S1\|\geq t_{s})&\leq n\cdot\mathrm{exp}(\frac{-(\alpha+1)\mathrm{log}(n+L)}{\frac{18}{(\sqrt{\alpha+1}+\sqrt{\alpha+19})^{2}}+\frac{2\sqrt{\alpha+1}}{\sqrt{\alpha+1}+\sqrt{\alpha+19}}\sqrt{\frac{\mathrm{log}(n+L)}{\theta_{\mathrm{max}}\|\theta\|_{1}\|\theta\|^{2}_{F}L}}})\\
&\leq\frac{n}{(n+L)^{\alpha+1}}\leq\frac{1}{(n+L)^{\alpha}}.
\end{align*}
Set $\alpha=1$, this lemma holds.
\end{proof}
\subsection{Proof of Theorem \ref{mainNSoA}}
\begin{proof}
First, we present the following two lemmas.
\begin{lem}\label{BoundUhatU}
Under $\mathrm{MLDCSBM}(Z,\Theta,\mathcal{B})$, we have
\begin{align*}
\|\hat{U}Q-U\|_{F}=O(\frac{\sqrt{K}\|A_{\mathrm{sum}}-\Omega_{\mathrm{sum}}\|}{\theta^{2}_{\mathrm{min}}n_{\mathrm{min}}L}), \end{align*}
where $Q$ is a $K\times K$ orthogonal matrix.
\end{lem}
\begin{proof}
By Lemma 5.1 in \citep{lei2015consistency}, we have
\begin{align*}
\|\hat{U}Q-U\|_{F}\leq\frac{2\sqrt{2K}\|A_{\mathrm{sum}}-\Omega_{\mathrm{sum}}\|}{|\lambda_{K}(\Omega_{\mathrm{sum}})|}.
\end{align*}
For $|\lambda_{K}(\Omega_{\mathrm{sum}})|$, we have
\begin{align*}
|\lambda_{K}(\Omega_{\mathrm{sum}})|&=\sqrt{\lambda_{K}(\Omega^{2}_{\mathrm{sum}})}=\sqrt{\lambda_{K}((\Theta Z(\sum_{l\in[L]}B_{l})Z'\Theta)^{2})}\\
&\geq\sqrt{\lambda^{4}_{K}(\Theta)\lambda^{2}_{K}(Z'Z)\lambda^{2}_{K}(\sum_{l\in[L]}B_{l})}\\
&=\lambda^{2}_{K}(\Theta)\lambda_{K}(Z'Z)|\lambda_{K}(\sum_{l\in[L]}B_{l})|\\
&=\theta^{2}_{\mathrm{min}}n_{\mathrm{min}}|\lambda_{K}(\sum_{l\in[L]}B_{l})|.
\end{align*}
By Assumption \ref{Assum11}, we have $|\lambda_{K}(\Omega_{\mathrm{sum}})|\geq c_{1}\theta^{2}_{\mathrm{min}}n_{\mathrm{min}}L$. This proves the claim.
\end{proof}
\begin{lem}\label{BoundUstarhatUsatr}
Under $\mathrm{MLDCSBM}(Z,\Theta,\mathcal{B})$, we have
\begin{align*}
\|\hat{U}_{*}Q-U_{*}\|_{F}=O(\frac{\theta_{\mathrm{max}}\sqrt{Kn_{\mathrm{max}}}\|A_{\mathrm{sum}}-\Omega_{\mathrm{sum}}\|}{\theta^{3}_{\mathrm{min}}n_{\mathrm{min}}L}).
\end{align*}
\end{lem}
\begin{proof}
Set $\mu=\mathrm{min}_{i\in[n]}\|U(i,:)\|_{F}$. We have $\|\hat{U}_{*}Q-U_{*}\|_{F}\leq\frac{2\|\hat{U}Q-U\|_{F}}{\mu}$ by basic algebra. Since $\frac{1}{\mu}\leq\frac{\theta_{\mathrm{max}}\sqrt{n_{\mathrm{max}}}}{\theta_{\mathrm{min}}}$ by the proof of Lemma 7 \citep{qing2023community}, we have $\|\hat{U}_{*}Q-U_{*}\|_{F}\leq\frac{2\theta_{\mathrm{max}}\sqrt{n_{\mathrm{max}}}\|\hat{U}Q-U\|_{F}}{\theta_{\mathrm{min}}}$. By Lemma \ref{BoundUhatU}, this lemma holds.
\end{proof}

By the proof of Theorem 2 \citep{qing2023community}, we know that $\hat{f}_{NSoA}=O(\frac{K}{n_{\mathrm{min}}}\|\hat{U}_{*}Q-U_{*}\|^{2}_{F})$. Since $\|\hat{U}_{*}Q-U_{*}\|_{F}=O(\frac{\theta_{\mathrm{max}}\sqrt{n_{\mathrm{max}}}\|\hat{U}Q-U\|_{F}}{\theta_{\mathrm{min}}})=O(\frac{\theta_{\mathrm{max}}\sqrt{Kn_{\mathrm{max}}}\|A_{\mathrm{sum}}-\Omega_{\mathrm{sum}}\|}{\theta^{3}_{\mathrm{min}}n_{\mathrm{min}}|\lambda_{K}(\sum_{l\in[L]}B_{l})|})$ by the proofs of Lemmas \ref{BoundUhatU} and \ref{BoundUstarhatUsatr}, we have $\hat{f}_{NSoA}=O(\frac{K^{2}\theta^{2}_{\mathrm{max}}n_{\mathrm{max}}\|A_{\mathrm{sum}}-\Omega_{\mathrm{sum}}\|^{2}}{\theta^{6}_{\mathrm{min}}n^{3}_{\mathrm{min}}\lambda^{2}_{K}(\sum_{l\in[L]}B_{l})})$. By Assumption \ref{Assum11}, we get $\hat{f}_{NSoA}=O(\frac{\theta^{2}_{\mathrm{max}}K^{2}n_{\mathrm{max}}\|A_{\mathrm{sum}}-\Omega_{\mathrm{sum}}\|^{2}}{\theta^{6}_{\mathrm{min}}n^{3}_{\mathrm{min}}L^{2}})$. Finally, this theorem holds by Lemma \ref{boundAsum}.
\end{proof}
\subsection{Proof of Theorem \ref{mainNDSoSA}}
\begin{proof}
Similar to the proof of Theorem \ref{mainNSoA}, first we present the following two lemmas.
\begin{lem}\label{BoundVhatV}
Under $\mathrm{MLDCSBM}(Z,\Theta,\mathcal{B})$, we have
\begin{align*}
\|\hat{V}Q_{s}-V\|_{F}=O(\frac{\sqrt{K}\|S_{\mathrm{sum}}-\tilde{S}_{\mathrm{sum}}\|}{\theta^{4}_{\mathrm{min}}n^{2}_{\mathrm{min}}L}),
\end{align*}
where $Q_{s}$ is a $K$-by-$K$ orthogonal matrix.
\end{lem}
\begin{proof}
According to Lemma 5.1 \citep{lei2015consistency}, there exists n orthogonal matrix $Q_{s}$ such that
\begin{align*}
\|\hat{V}Q_{s}-V\|_{F}\leq\frac{2\sqrt{2K}\|S_{\mathrm{sum}}-\tilde{S}_{\mathrm{sum}}\|}{|\lambda_{K}(\tilde{S}_{\mathrm{sum}})|}.
\end{align*}
By Assumption \ref{Assum22} and the fact that $Z'\Theta^{2}Z$ is a diagonal matrix, we have
\begin{align*}
&|\lambda_{K}(\tilde{S}_{\mathrm{sum}})|=\sqrt{\lambda_{K}((\sum_{l\in[L]}\Omega^{2}_{l})^{2})}\\
&=\sqrt{\lambda_{K}((\sum_{l\in[L]}\Theta ZB_{l}Z'\Theta^{2}ZB_{l}Z'\Theta)^{2})}\\
&=\sqrt{\lambda_{K}(\Theta^{2}(\sum_{l\in[L]}ZB_{l}Z'\Theta^{2}ZB_{l}Z')^{2}\Theta^{2})}\\
&=\sqrt{\lambda_{K}(\Theta^{4}(\sum_{l\in[L]}ZB_{l}Z'\Theta^{2}ZB_{l}Z')^{2})}\\
&\geq\sqrt{\lambda_{K}(\Theta^{4})\lambda_{K}((\sum_{l\in[L]}ZB_{l}Z'\Theta^{2}ZB_{l}Z')^{2})}\\
&=\sqrt{\lambda_{K}(\Theta^{4})\lambda_{K}(Z'Z(\sum_{l\in[L]}B_{l}Z'\Theta^{2}ZB_{l})^{2}Z'Z)}\\
&\geq\sqrt{\lambda_{K}(\Theta^{4})\lambda^{2}_{K}(Z'Z)\lambda_{K}((\sum_{l\in[L]}B_{l}Z'\Theta^{2}ZB_{l})^{2})}\\
&=O(\sqrt{\lambda_{K}(\Theta^{4})\lambda^{2}_{K}(Z'Z)\lambda_{K}((Z'\Theta^{2}Z)^{2}(\sum_{l\in[L]}B^{2}_{l})^{2})})\\
&=O(\sqrt{\lambda^{8}_{K}(\Theta)\lambda^{4}_{K}(Z'Z)\lambda^{2}_{K}(\sum_{l\in[L]}B^{2}_{l})})\\
&=O(\lambda^{4}_{K}(\Theta)\lambda^{2}_{K}(Z'Z)|\lambda_{K}(\sum_{l\in[L]}B^{2}_{l})|)\\
&=O(\theta^{4}_{\mathrm{min}}n^{2}_{\mathrm{min}}|\lambda_{K}(\sum_{l\in[L]}B^{2}_{l})|)\\
&=O(\theta^{4}_{\mathrm{min}}n^{2}_{\mathrm{min}}L).
\end{align*}
Then we get $\|\hat{V}Q_{s}-V\|_{F}=O(\frac{\sqrt{K}\|S_{\mathrm{sum}}-\tilde{S}_{\mathrm{sum}}\|}{\theta^{4}_{\mathrm{min}}n^{2}_{\mathrm{min}}L})$.
\end{proof}
\begin{lem}\label{BoundVstarhatVsatr}
Under $\mathrm{MLDCSBM}(Z,\Theta,\mathcal{B})$, we have
\begin{align*}
\|\hat{V}_{*}Q_{s}-V_{*}\|_{F}=O(\frac{\theta_{\mathrm{max}}\sqrt{Kn_{\mathrm{max}}}\|S_{\mathrm{sum}}-\tilde{S}_{\mathrm{sum}}\|}{\theta^{5}_{\mathrm{min}}n^{2}_{\mathrm{min}}L}),
\end{align*}
\end{lem}
\begin{proof}
Set $\mu_{s}=\mathrm{min}_{i\in[n]}\|V(i,:)\|_{F}$. We have $\|\hat{V}_{*}Q_{s}-V_{*}\|_{F}\leq\frac{2\|\hat{V}Q_{s}-V\|_{F}}{\mu_{s}}$. Since $\frac{1}{\mu_{s}}\leq\frac{\theta_{\mathrm{max}}\sqrt{n_{\mathrm{max}}}}{\theta_{\mathrm{min}}}$ by the proof of Lemma 7 \citep{qing2023community}, this lemma holds based on the results in Lemma \ref{BoundVhatV}.
\end{proof}

According to the proof of Theorem 2 \citep{qing2023community}, $\hat{f}_{NDSoSA}=O(\frac{K}{n_{\mathrm{min}}}\|\hat{V}_{*}Q_{s}-V_{*}\|^{2}_{F})$. By Lemma \ref{BoundVstarhatVsatr}, we get $\hat{f}_{NDSoSA}=O(\frac{\theta^{2}_{\mathrm{max}}K^{2}n_{\mathrm{max}}\|S_{\mathrm{sum}}-\tilde{S}_{\mathrm{sum}}\|^{2}}{\theta^{10}_{\mathrm{min}}n^{5}_{\mathrm{min}}\lambda^{2}_{K}(\sum_{l\in[L]}B^{2}_{l})})=O(\frac{\theta^{2}_{\mathrm{max}}K^{2}n_{\mathrm{max}}\|S_{\mathrm{sum}}-\tilde{S}_{\mathrm{sum}}\|^{2}}{\theta^{10}_{\mathrm{min}}n^{5}_{\mathrm{min}}L^{2}})$. Finally, by Lemma \ref{boundSsum}, we complete the proof.
\end{proof}
\subsection{Consistency results for NSoSA}\label{MainNSoSA}
Here, we provide consistency results for NSoSA and subsequently demonstrate why NDSoSA consistently outperforms NSoSA. Recall that NSoSA employs the K-means algorithm on the normalized version of the matrix constructed from the leading $K$ eigenvectors of $\sum_{l\in[L]}A^{2}_{l}$, to establish NSoSA's consistency, we need to obtain the upper bound of $\|\sum_{_{l\in[L]}}(A^{2}_{l}-\Omega^{2}_{l})\|$. Since $\|\sum_{_{l\in[L]}}(A^{2}_{l}-\Omega^{2}_{l})\|=\|S_{\mathrm{sum}}-\tilde{S}_{\mathrm{sum}}+\sum_{l\in[L]}D_{l}\|$, we have $|\sum_{_{l\in[L]}}(A^{2}_{l}-\Omega^{2}_{l})\|\leq\|S_{\mathrm{sum}}-\tilde{S}_{\mathrm{sum}}\|+\|\sum_{l\in[L]}D_{l}\|=\|S_{\mathrm{sum}}-\tilde{S}_{\mathrm{sum}}\|+\mathrm{max}_{i\in[n]}\sum_{l\in[L]}D_{l}(i,i)=|S_{\mathrm{sum}}-\tilde{S}_{\mathrm{sum}}\|+\mathrm{max}_{i\in[n]}\sum_{l\in[L]}\sum_{j\in[n]}A_{l}(i,j)$, where the upper bound of $\|S_{\mathrm{sum}}-\tilde{S}_{\mathrm{sum}}\|$ is given in Lemma \ref{boundSsum}. The following lemma bounds $\mathrm{max}_{i\in[n]}\sum_{l\in[L]}\sum_{j\in[n]}A_{l}(i,j)$.
\begin{lem}\label{maxDil}
Under $\mathrm{MLDCSBM}(Z,\Theta,\mathcal{B})$, when $\theta_{\mathrm{max}}\|\theta\|_{1}L\geq\mathrm{log}(n+L)$, with probability at least $1-O(\frac{1}{n+L})$, we have
\begin{align*}
\|\sum_{l\in[L]}D_{l}\|=O(\theta_{\mathrm{max}}\|\theta\|_{1}L).
\end{align*}
\end{lem}
\begin{proof}
Since $\mathbb{E}[A_{l}(i,j)-\Omega_{l}(i,j)]=0, |A_{l}(i,j)-\Omega_{l}(i,j)|\leq1$ for $l\in[L], i\in[n], j\in[n]$, and $\sigma_{0}=\sum_{l\in[L]}\sum_{j\in[n]}\mathbb{E}[(A_{l}(i,j)-\Omega_{l}(i,j))^{2}]=\sum_{l\in[L]}\sum_{j\in[n]}\Omega_{l}(i,j)(1-\Omega_{l}(i,j))\leq\sum_{l\in[L]}\sum_{j\in[n]}\theta(i)\theta(j)\leq\theta_{\mathrm{max}}\|\theta\|_{1}L$, for any $t_{0}>0$, by Theorem \ref{Bern}, we have
\begin{align*}
\mathbb{P}(|\sum_{l\in[L]}\sum_{j\in[n]}(A_{l}(i,j)-\Omega_{l}(i,j))|\geq t_{0})\leq\mathrm{exp}(\frac{-t^{2}_{0}/2}{\theta_{\mathrm{max}}\|\theta\|_{1}L+t_{0}/3}).
\end{align*}
Set $t_{0}=\frac{\alpha+1+\sqrt{(\alpha+1)(\alpha+19)}\sqrt{\theta_{\mathrm{max}}\|\theta\|_{1}L\mathrm{log}(n+L)}}{3}$ for any $\alpha\geq0$, we have
\begin{align*}
&\mathbb{P}(|\sum_{l\in[L]}\sum_{j\in[n]}(A_{l}(i,j)-\Omega_{l}(i,j))|\geq t_{0})\\
&\leq\mathrm{exp}(\frac{-(\alpha+1)\mathrm{log}(n+L)}{\frac{18}{(\sqrt{\alpha+1}+\sqrt{\alpha+19})^{2}}+\frac{2\sqrt{\alpha+1}}{\sqrt{\alpha+1}+\sqrt{\alpha+19}}\sqrt{\frac{\mathrm{log}(n+L)}{\theta_{\mathrm{max}}\|\theta\|_{1}L}}})\\
&\leq\frac{1}{(n+L)^{\alpha+1}},
\end{align*}
where the last inequality holds by the assumption $\theta_{\mathrm{max}}\|\theta\|_{1}L\geq\mathrm{log}(n+L)$. Thus, with probability at least $1-O(\frac{1}{n+L})$, we have
\begin{align*}
|\sum_{l\in[L]}\sum_{j\in[n]}(A_{l}(i,j)-\Omega_{l}(i,j))|\leq t_{0},
\end{align*}
which gives that
\begin{align*}
\sum_{l\in[L]}\sum_{j\in[n]}A_{l}(i,j)&\leq t_{0}+\sum_{l\in[L]}\sum_{j\in[n]}\Omega_{l}(i,j)\\
&\leq t_{0}+\sum_{l\in[L]}\sum_{j\in[n]}\theta(i)\theta(j)\\
&\leq t_{0}+\theta_{\mathrm{max}}\|\theta\|_{1}L\\
&=O(\theta_{\mathrm{max}}\|\theta\|_{1}L),
\end{align*}
where the last equality holds by the assumption $\theta_{\mathrm{max}}\|\theta\|_{1}L\geq\mathrm{log}(n+L)$. Set $\alpha=0$, this lemma holds.
\end{proof}
By Lemma \ref{maxDil}, we have
\begin{align*}
\|\sum_{l\in[L]}(A^{2}_{l}-\Omega^{2}_{l})\|\leq\|S_{\mathrm{sum}}-\tilde{S}_{\mathrm{sum}}\|+O(\theta_{\mathrm{max}}\|\theta\|_{1}L).
\end{align*}
Following a similar proof of Theorem \ref{mainNDSoSA}, we know that NSoSA's error rate is
\begin{align*}
\hat{f}_{NSoSA}&=O(\frac{\theta^{2}_{\mathrm{max}}K^{2}n_{\mathrm{max}}\|\sum_{l\in[L]}(A^{2}_{l}-\Omega^{2}_{l})\|^{2}}{\theta^{10}_{\mathrm{min}}n^{5}_{\mathrm{min}}L^{2}})\\
&=O(\frac{\theta^{2}_{\mathrm{max}}K^{2}n_{\mathrm{max}}(\|S_{\mathrm{sum}}-\tilde{S}_{\mathrm{sum}}\|+O(\theta_{\mathrm{max}}\|\theta\|_{1}L))^{2}}{\theta^{10}_{\mathrm{min}}n^{5}_{\mathrm{min}}L^{2}})\\
&=\hat{f}_{NDSoSA}+O(\frac{K^{2}\theta^{4}_{\mathrm{max}}\|\theta\|^{2}_{1}n_{\mathrm{max}}}{\theta^{10}_{\mathrm{min}}n^{5}_{\mathrm{min}}}).
\end{align*}
It is clear that $\hat{f}_{NSoSA}\geq\hat{f}_{NDSoSA}$ and this explains why NDSoSA almost always outperforms NSoSA. Furthermore, under the same conditions as Corollary \ref{CorSsum}, we have
\begin{align*}
\hat{f}_{NSoSA}&=\hat{f}_{NDSoSA}+O(\frac{1}{\rho^{2}n^{2}})\\
&=O(\frac{\mathrm{log}(n+L)}{\rho^{2}n^{2}L})+O(\frac{1}{n^{2}})+O(\frac{1}{\rho^{2}n^{2}})\\
&=O(\frac{\mathrm{log}(n+L)}{\rho^{2}n^{2}L})+O(\frac{1}{\rho^{2}n^{2}}).
\end{align*}
Following a similar analysis as that of Section \ref{CompareDA}, we find that NSoA outperforms NSoSA only when $L\gg\mathrm{log}(n+L)$.
\bibliographystyle{elsarticle-num}
\bibliography{refMLDCSBM}
\end{document}